%% file: main.tex
\pdfoutput=1 
\documentclass[11pt]{article}
\usepackage[pagebackref, 
            colorlinks=true, 
            linkcolor=blue,
            urlcolor=blue,
            citecolor=red, 
            bookmarks, 
            hypertexnames=false]{hyperref}

\usepackage{tikz}
\usetikzlibrary{shapes,arrows.meta,positioning,shapes.geometric}

\usepackage[blocks]{authblk}

\usepackage{cite}
\usepackage{graphicx}

\usepackage{amsmath, amsfonts, amssymb, dsfont, braket, amsthm}

\usepackage{cleveref}

\usepackage{subcaption}

\usepackage{autonum}

\usepackage{url}

\usepackage{ifthen}

\usepackage[margin=1in]{geometry}

\newtheoremstyle{newdefinition}{}{}{\normalfont}{}{\bfseries}{}{\newline}
{\thmname{#1} \thmnumber{#2}\thmnote{ (#3)}}

\newtheoremstyle{newplain}{}{}{\itshape}{}{\bfseries}{}{1em}
{\thmname{#1} \thmnumber{#2}\thmnote{ (#3)}}

\newtheoremstyle{newremark}{}{}{\normalfont}{}{\bfseries}{}{1em}
{\thmname{#1}}

\theoremstyle{newdefinition}
\newtheorem{definition}{Definition}[section]

\theoremstyle{newplain}
\newtheorem{theorem}[definition]{Theorem}
\newtheorem{lemma}[definition]{Lemma}
\newtheorem{proposition}[definition]{Proposition}
\newtheorem{corollary}[definition]{Corollary}
\newtheorem{remark}[definition]{Remark}

\DeclareMathOperator{\R}{\mathbb{R}} 
\DeclareMathOperator{\N}{\mathbb{N}}

\DeclareMathOperator{\C}{\mathbb{C}}

\DeclareMathOperator{\D}{\mathbb{D}}

\DeclareMathOperator{\cH}{\mathcal{H}}

\DeclareMathOperator{\cB}{\mathcal{B}}

\DeclareMathOperator{\cS}{\mathcal{S}}

\DeclareMathOperator{\identity}{\mathds{1}}
\DeclareMathOperator{\rank}{rank}
\DeclareMathOperator{\supp}{supp}

\newcommand{\norm}[1]{\left\Vert#1\right\Vert}
\newcommand{\Cnorm}[1]{\left|#1\right|}
\newcommand{\twoline}[2]{\genfrac{}{}{0pt}{}{#1}{#2}}
\newcommand{\dyad}[1]{\left|#1\right\rangle\!\!\left\langle#1\right|}
\newcommand{\tr}[1][~]{
    \ifthenelse{\equal{#1}{~}}{\operatorname{Tr}}{\text{Tr}\left[#1\right]}
}

\definecolor{cyan}{RGB}{0,204,204}
\definecolor{lightgreen}{RGB}{0,255,128}
\definecolor{midgreen}{RGB}{178,255,102}
\definecolor{midblue}{RGB}{0,102,204}

\begin{document}

\title{Continuity of quantum entropic quantities via almost convexity\thanks{This paper was presented at Beyondiid10 2022 and is a long version of \cite{BluhmCapelGondolfPerezHernandez-ShortContinuityBounds-2023}.}}
\author[1,2]{Andreas Bluhm\thanks{andreas.bluhm@univ-grenoble-alpes.fr}}
\author[3]{{\'A}ngela Capel\thanks{angela.capel@uni-tuebingen.de}}
\author[3]{Paul Gondolf\thanks{paul.gondolf@student.uni-tuebingen.de}}
\author[4]{Antonio Pérez-Hernández\thanks{antperez@ind.uned.es}}
\affil[1]{QMATH, Department of Mathematical Sciences, University of Copenhagen, Universitetsparken 5, 2100 Copenhagen, Denmark}
\affil[2]{Univ. Grenoble Alpes, CNRS, Grenoble INP, LIG, 38000 Grenoble, France}
\affil[3]{Fachbereich Mathematik, Universit\"at T\"ubingen, 72076 T\"ubingen, Germany}
\affil[4]{Departamento de Matem\'{a}tica Aplicada I, Escuela T\'{e}cnica Superior de Ingenieros Industriales, Universidad Nacional de Educación a Distancia, calle Juan del Rosal 12, 28040 Madrid (Ciudad Universitaria), Spain}
\date{\today}

\maketitle

\begin{abstract}
    Based on the proofs of the continuity of the conditional entropy by Alicki, Fannes, and Winter,  we introduce in this work the almost locally affine (ALAFF) method. This method allows us to prove a great variety of continuity bounds for the derived entropic quantities. First, we apply the ALAFF method to the Umegaki relative entropy. This way, we recover known almost tight bounds, but also some new continuity bounds for the relative entropy. Subsequently, we apply our method to the Belavkin-Staszewski relative entropy (BS-entropy). This yields novel explicit bounds in particular for the BS-conditional entropy, the BS-mutual and BS-conditional mutual information. On the way, we prove almost concavity for the Umegaki relative entropy and the BS-entropy, which might be of independent interest. We conclude by showing some applications of these continuity bounds in various contexts within quantum information theory. 
\end{abstract}

\newpage

\tableofcontents

\section{Introduction}
\input{sections/introduction}

\section{Main results}\label{sec:main-results}
\input{sections/main-results}

\section{Preliminaries}\label{sec:prelim}
\input{sections/preliminaries}

\section{From almost convexity to continuity bounds}\label{sec:ALAFF_method}
\input{sections/from-almost-concav-to-cont-bounds}

\section{Almost concavity and continuity bounds for the Umegaki relative entropy}\label{sec:umegaki}
\input{sections/almost-concav-rel-ent}

\section[Almost concavity and continuity bounds for the Belavkin-Staszewski entropy]{Almost concavity and continuity bounds for the Belavkin\\-Staszewski entropy}\label{sec:BS}
\input{sections/almost-concav-bs-ent}

\section{Applications}\label{sec:applications}
\input{sections/applications}

\section{Outlook}\label{sec:outlook}
\input{sections/outlook}

\vspace{0.5cm}

\noindent \textbf{Acknowledgements:}
The authors are grateful to Li Gao for the interesting discussions, to Peter Brown for the example used in \cref{prop:discontinuity_conditional_BS_entropy}, and to Mark Wilde for the discussion on the difference between BS-mutual information and its variational counterpart. Moreover, the authors would like to thank {\'A}lvaro Alhambra for pointing them to \cite{gour2022role}. The authors also express their thanks to Ludovico Lami and Marco Tomamichel for spotting an error in Lemma 5.8 of an earlier version of the paper. A.B.~acknowledges financial support from the European Research Council (ERC Grant Agreement No. 81876) and VILLUM FONDEN via the QMATH Centre of Excellence (Grant No.10059). A.P.H acknowledges financial support from the Spanish Ministerio de Ciencia e Innovación (grant PID2020-
113523GB-I00) and Comunidad de Madrid (grant
QUITEMAD-CMS2018/TCS-4342). This work was partially funded by the Deutsche Forschungsgemeinschaft (DFG, German Research Foundation) – Project-ID 470903074 – TRR 352.

\bibliographystyle{abbrv}
\bibliography{lit}

\newpage
\appendix
\input{sections/appendix}

\end{document}

%% file: sections/introduction.tex
Entropic quantities have proven essential in characterizing the information-processing capabilities both of classical and quantum systems. As the real world cannot be measured to infinite precision, such quantities need to be continuous to contain meaningful information about physical systems. Often, however, we do not only want to know whether an entropic quantity is continuous but also to quantify this continuity. That means we are interested in estimating for an entropic quantity $f$
\begin{equation}
    \sup\{|f(\rho) - f(\sigma)|: \rho, \sigma \in \mathcal S_0, d(\rho, \sigma) \leq \varepsilon\}.
\end{equation}
for some subset $\mathcal S_0$ of the quantum states and some appropriate distance measure $d$ such as the trace distance, for example. 

Already in 1973, Fannes \cite{Fannes-ContinuityEntropy-1973} proved that the von Neumann entropy is uniformly continuous and gave a concrete dimension-dependent bound, which was later improved to a sharp version in \cite{Audenaert-ContinuityEstimateEntropy-2007, Petz2008}. Similar results in the line of almost concavity for the von Neumann entropy were provided in \cite{Kim-ConvexityEstimates-2013},  \cite{CarlenLieb-EntropyInequalities-2014}, \cite{KimRuskai-ConcavityEntropy-2014} or \cite{audenaert2014quantum}, among others. Another example of a concrete continuity estimate is the Alicki-Fannes inequality for the conditional entropy \cite{AlickiFannes-2004}, which was subsequently improved to an almost tight version by Winter \cite{Winter-AlickiFannes-2016}. Applications of this kind of continuity bounds include, but are not limited to, entanglement measures \cite{Nielsen2000cont-bounds} and the capacities of quantum channels \cite{LeungSmith2009capacities, shirokov2017tight}. We refer the reader to textbooks such as \cite{WildeFromClassicalToQuantumInformation_2016} for a thorough discussion of continuity bounds and their applications.

The importance of the Alicki-Fannes result in \cite{AlickiFannes-2004} goes beyond its quantification of the continuity of the conditional entropy, but their method and its improved versions \cite{mosonyi2011quantum, synak2006asymptotic, Winter-AlickiFannes-2016} work quite generally for entropic quantities. Most clearly, this has been realized by Shirokov, who has named this approach the \textit{Alicki-Fannes-Winter method} \cite{Shirokov-AFWmethod-2020,Shirokov-ContinuityReview-2022}. We continue this line of work by generalising the Shirokov approach further to what we call the \textit{almost locally affine (ALAFF) method}. The aim of this generalization is to apply it to entropic quantities beyond those derived from the Umegaki relative entropy \cite{Umegaki-RelativeEntropy-1962}, such as the conditional entropy. In particular, we are interested in the Belavkin-Staszewski relative entropy (BS-entropy) \cite{BelavkinStaszewski-BSentropy-1982} and its derived entropic quantities. As the Umegaki relative entropy, it generalizes the Kullback-Leibler relative entropy of classical systems \cite{KullbackLeibler-KLD-1951}, but it is less well studied (see \cite{Bluhm2020, BluhmCapelPerezHernandez-WeakQFBSentropy-2021, Hiai2017,  matsumoto2010reverse,Matsumoto2018} for some recent results). The BS-entropy and the related family of geometric R{\'e}nyi divergences have recently found an application for estimating channel capacities \cite{FangFawzi-GeometricRenyiDivergences-2019}. Moreover, generalizations of the mutual information and other entropic quantities based on the BS-entropy have been defined \cite{Scalet2021, ZhaiYangXi-BSEntropy-2022}. The BS-mutual information has been instrumental in proving that the mutual information in one-dimensional quantum Gibbs states of finite-range, translation-invariant Hamiltonians decays exponentially fast \cite{bluhm2022exponential} and that Davies generators in one dimension which converge to those Gibbs states, in the commuting case, satisfy a positive modified logarithmic Sobolev inequality at every temperature, and thus rapid mixing \cite{bardet2021entropy, bardet2021rapid}.

A short version of the current manuscript, with new applications in the context of quantum entropic uncertainty relations, has been published in \cite{BluhmCapelGondolfPerezHernandez-ShortContinuityBounds-2023}.

%% file: sections/main-results.tex
This section summarizes the main results of this article. The focus of this work is not so much on the continuity bounds themselves, but more on the introduction of the method which allows deriving all of them in a systematic way (\cref{sec:ALAFF_method}). Our approach is summarized in \cref{fig:fig_flow_chart}. For a given divergence, in this paper either the Umegaki relative entropy \cite{Umegaki-RelativeEntropy-1962} or the BS-entropy \cite{BelavkinStaszewski-BSentropy-1982}, we need to prove two properties: its (joint) convexity and its almost (joint) concavity. Both of these properties, under certain conditions on the remainder function, then directly translate into almost local affinity (\cref{def:ALAFF-function}) of the entropic quantities derived from the divergence at hand on a suitably defined subset $\cS_0$ of $\cS(\cH)$. Serving as input to the ALAFF method, the remainder estimates get translated into continuity bounds for said quantities. The entropic quantities include, for example, versions of the conditional entropy and the (conditional) mutual information, as defined in \cref{fig:fig_flow_chart}. The necessity of $\cS_0$ as a restriction of $\cS(\cH)$ becomes obvious when trying to prove continuity bounds for the Umegaki relative entropy, for example. It is known not to be continuous on the set of all pairs of states $(\rho, \sigma)$, which makes a careful choice of $\cS_0$ inevitable. To this end, we define $s$-perturbed $\Delta$-invariant convex subsets of $\cS(\cH)$ (\cref{def:def_perturbed_delta_invariant_subset}) for which we can show that the ALAFF method works and which are general enough to capture all situations of interest. For the formal statement of the ALAFF method, we refer the reader to \cref{theo:theo_alaff_method}.

\begin{figure}[ht!]
    \centering
    \scalebox{0.8}{
    \begin{tikzpicture}[font=\small,thick,scale=1, every node/.style={scale=1}]
        \node[draw ,
            diamond,
            fill=midblue!30!white,
            minimum width=2.5cm,
            minimum height=1cm,
            inner sep=0.1cm,
            align=center] (block1) {Divergence\\$\D(\cdot\Vert\cdot)$};
            
        \node[draw,
            rounded rectangle,
            fill=cyan!30!white,
            below left=of block1,
            minimum width=2.5cm,
            inner sep=0.1cm,
            align=center] (block2) {Convexity\\{\tiny$p \D(\rho_1\Vert \sigma_1) + (1 - p) \D(\rho_2\Vert \sigma_2) \ge \D(\rho\Vert \sigma)$}};
            
        \node[draw,
            rounded rectangle,
            fill=cyan!30!white,
            below right=of block1,
            minimum width=2.5cm,
            inner sep=0.1cm,
            align=center] (block3) {Almost concavity\\{\tiny$\D(\rho\Vert\sigma) \ge p \D(\rho_1\Vert \sigma_1) + (1 - p) \D(\rho_2\Vert \sigma_2) - f(p)$}};
         
        \node[draw,
            trapezium, 
            fill=lightgreen!50!white,
            below=2cm of block1,
            trapezium left angle = 65,
            trapezium right angle = 115,
            trapezium stretches,
            minimum width=3.5cm,
            minimum height=1cm] (block4) {ALAFF method};
        
        \node[coordinate,below=0.5cm of block4] (block5) {};
        
        \node[draw,
            below=of block5,
            fill=midgreen!50!white,
            minimum width=2.5cm,
            minimum height=1cm,
            align=center] (block6) {Conditional\\mutual information\\{\tiny$ \mathbb I_\rho(A:C|B) :=  $}\\{\tiny$ \D_\rho(B|C) - \D_\rho(AB|C) $}};
        
        \node[draw,
            right=0.5cm of block6,
            fill=midgreen!50!white,
            minimum width=2.5cm,
            minimum height=1cm,
            align=center] (block7) {Mutual information\\{\tiny$ \mathbb I_\rho(A:B) :=  $}\\{\tiny$  \D(\rho_{AB} \| \rho_A \otimes \rho_B) $}};
        
        \node[draw,
            right=0.5cm of block7,
            fill=midgreen!50!white,
            minimum width=2.5cm,
            minimum height=1cm,
            align=center] (block8) {Conditional divergence\\{\tiny$\mathbb{H}_\rho(A|B) :=  $}\\{\tiny$ - \D(\rho_{AB} \| \identity_A \otimes \rho_B) $}};
            
        \node[draw,
            left=0.5cm of block6,
            fill=midgreen!50!white,
            minimum width=2.5cm,
            minimum height=1cm,
            align=center] (block9) {Divergence\\(fixed second argument)\\{\tiny$ \left| \D (\rho_1 \| \sigma ) -  \D (\rho_2 \| \sigma ) \right| $}\\{\tiny$  \leq f_{D,2}(\| \rho_1 - \rho_2 \|_1)$}};
        
        \node[draw,
            left=0.5cm of block9,
            fill=midgreen!50!white,
            minimum width=2.5cm,
            minimum height=1cm,
            align=center] (block10) {Divergence\\(fixed first argument)\\{\tiny$ \left| \D (\rho \| \sigma_1 ) -  \D (\rho \| \sigma_2 ) \right| $}\\{\tiny$  \leq f_{D,1}(\| \sigma_1 - \sigma_2 \|_1)$}};

        \node[circle, draw, fill,inner sep=1.5pt,below=0.4cm of block10] (block11) {};
        
        \node[draw,
            below=of block10,
            fill=midgreen!50!white,
            minimum width=2.5cm,
            minimum height=1cm,
            align=center] (block12) {Divergence\\{\tiny$ \left| \D (\rho_1 \| \sigma_1 ) -  \D (\rho_2 \| \sigma_2 ) \right| $}\\{\tiny$  \leq f_D(\| \rho_1 - \rho_2 \|_1,\| \sigma_1 - \sigma_2 \|_1)$}};

        \node[draw,
            below=of block9,
            fill=midgreen!50!white,
            minimum width=2.5cm,
            minimum height=1cm,
            align=center] (block13) {Divergence bound\\{\tiny$  \D (\rho \| \sigma )    \leq f_{DB}(\| \rho - \sigma \|_1)$}};
        
        \draw[-latex] (block1) -| (block2)
            node[pos=0.25,fill=white,inner sep=0]{Proof};
         
        \draw[-latex] (block1) -| (block3)
            node[pos=0.25,fill=white,inner sep=0]{Proof};

        \draw[-latex] (block2) -| (block4);
        \draw[-latex] (block3) -| (block4);
    
        \draw[-Turned Square] (block5) -| (block6);
        \draw[-Turned Square] (block5) -| (block7);
        \draw[-Turned Square] (block5) -| (block8);
        \draw[-Turned Square] (block5) -| (block9);
        \draw[-Turned Square] (block5) -| (block10);
        
        \draw[-latex] (block10) -- (block12);
        \draw[-latex] (block9) -- (block13);
        \draw (block9) |- (block11);
        
        \draw (block4) -- (block5) 
            node[pos=1,fill=white,inner sep=0.1cm] {Uniform continuity \& Continuity bounds};
    \end{tikzpicture}
    }
    \caption{A flow chart demonstrating how convexity and almost concavity of a divergence can be used to obtain uniform continuity and explicit continuity bounds on entropic quantities derived from that divergence. The subscripts of the functions $f_{D, 1/2}$ and $f_{DB}$ stand for divergence first, second argument and divergence bound respectively.}
    \label{fig:fig_flow_chart}
\end{figure}
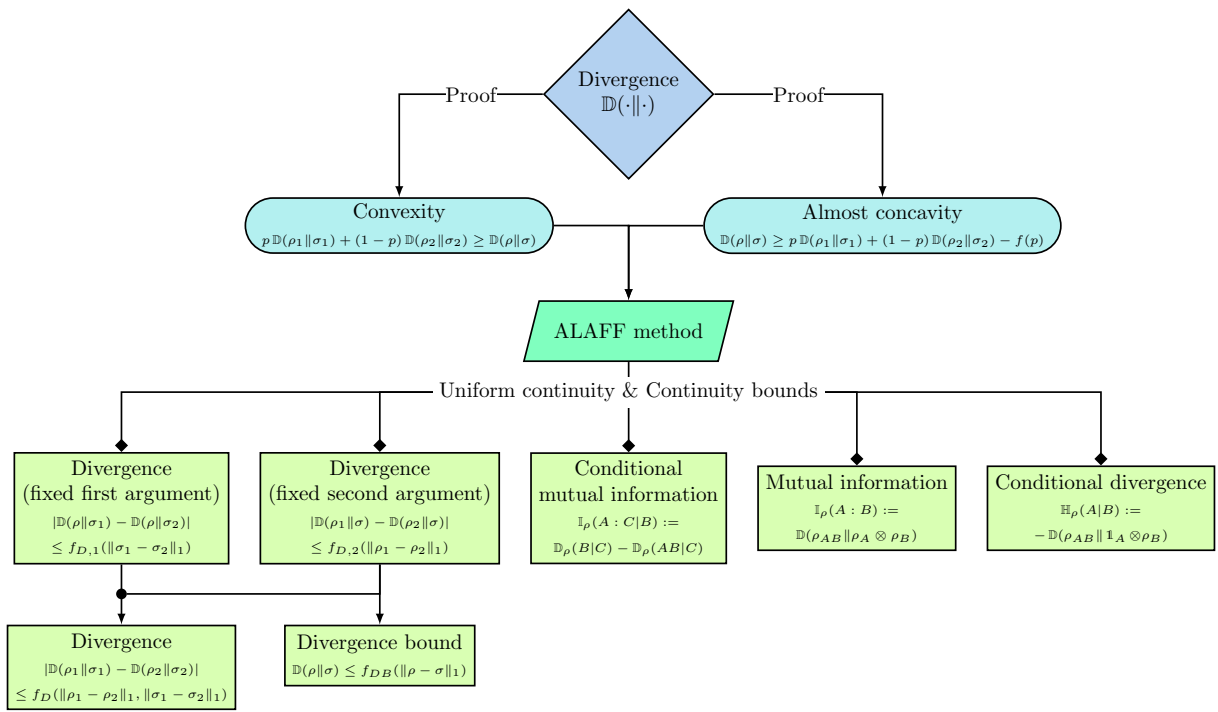

Thus, we are left with proving convexity and almost concavity for the divergences we are interested in, namely the Umegaki relative entropy (\cref{sec:umegaki}) and the Belavkin-Staszewski entropy  (\cref{sec:BS}), and deriving the precise continuity estimates. For the convexity, we can rely on well-known results from the literature both for the Umegaki relative entropy \cite{Lindblad-ConvexityRE-1974} and the BS-entropy \cite{Hiai2017, matsumoto2010reverse}. For the Umegaki relative entropy, given by
\begin{equation}
    D(\rho \Vert \sigma) :=
        \tr[\rho (\log\rho - \log \sigma)]  \;  \; \;  \text{ if } \ker \sigma \subseteq \ker \rho \, \; ,
\end{equation}
$\text{or } + \infty \; \text{ otherwise}$, we prove almost concavity in \cref{theo:theo_almost_concavity_relative_entropy} and find that it is tight. The application of the ALAFF method then allows us to recover in \cref{subsec:reduction_continuity_bounds_relative_entropy} the almost tight results for the conditional entropy \cite{Winter-AlickiFannes-2016} and the mutual and conditional mutual information (which can be derived from the conditional entropy \cite{WildeFromClassicalToQuantumInformation_2016}), but also to derive in \cref{subsec:new_continuity_bounds_relative_entropy} new versions of what we call divergence bounds \cite{AudenaertEisert_I_2005,AudenaertEisert_II_2011,BraRob81,Vershynina_2019}, i.e. bounds on $D(\rho||\sigma)$ in terms of $\frac{1}{2}\|\rho - \sigma\|_1$. Furthermore, our technique produces a new result, which is the uniform continuity of the relative entropy itself (in both arguments, on a suitable set $\mathcal S_0$), as well as an explicit continuity bound. 

For the BS-entropy, given by
\begin{equation}
    \widehat D(\rho \Vert \sigma) :=
        \tr[\rho \log(\rho^{1/2} \sigma^{-1} \rho^{1/2}  )]  \;  \; \;  \text{ if } \ker \sigma \subseteq \ker \rho \, \; ,
\end{equation}
$\text{or } + \infty \; \text{ otherwise}$, we prove the almost concavity in \cref{theo:theo_almost_concavity_bs_entropy}.

The ALAFF method yields novel explicit bounds in particular for the BS-conditional entropy, the BS-mutual and BS-conditional mutual information that we gather in \cref{subsec:subsec_continuity_bounds_bs_entropy}. We expect these new continuity bounds and those provided for quantities derived from the relative entropy to find applications in proving the continuity of various quantities in diverse fields related to quantum information theory. In particular, we provide here a number of applications of our results in the context of quantum hypothesis testing (\cref{sec:hyp-test}), to show that states that are hard to discriminate have almost the same performance in terms of hypothesis testing, as well as in quantum thermodynamics (\cref{sec:free_energy}), to show continuity of the distillable athermality. We also reprove that a state is an approximate quantum Markov chain if and only if it is close to being recovered by the Petz recovery map (\cref{sec:approxQMC}), and use our most general continuity bounds for the relative entropy to obtain bounds for the difference between the relative entropy and the BS-entropy of two quantum states (\cref{sec:difference_entropies}). Additionally, we show a new result of weak quasi-factorization for the relative entropy, i.e. with an additive error term and no multiplicative error term (\cref{sec:appl-qf-entropies}). Finally, we include continuity bounds for the relative entropy of entanglement as well as the analogously defined BS-entropy of entanglement (\cref{sec:min-dist-sep-states}), and subsequently lift these results to show continuity of the Rains information induced by the relative and the BS-entropy respectively (\cref{sec:Rains_info}).

%% file: sections/preliminaries.tex
\subsection{Notation and basic concepts}\label{sec:basic_notation}

We denote a Hilbert space by $\cH$ which, throughout this paper is assumed to be finite. The dimension of such a Hilbert space will be called $d$ and for its elements, we use $\ket{\varphi}$, $\ket{\psi}$ and $\ket{i}$ for $i \in \N$, possibly with additional indices. If we are concerned with a bipartite or tripartite system, we will always use capital letters in the index to identify objects associated with the respective subsystems. If we have, for example, the bipartite space $\cH = \cH_A \otimes \cH_B$ and consider the dimension of $\cH_A$, we write $d_A$.\par
The set of (bounded) linear operators on a Hilbert space $\cH$ is $\cB(\cH)$ and the subspace of positive semi-definite operators with trace one, i.e., the quantum states or density matrices, is denoted by $\cS(\cH)$. If we want to restrict this set even further, we indicate this with a subindex. 
 Thus, the set of positive definite quantum states becomes $\cS_+(\cH)$, or if we want to restrict moreover to the set of quantum states that have minimal eigenvalue greater than $m$, we write $\cS_{\ge m}(\cH)$. On the set of quantum states as well as on the set of self-adjoint operators, the relation $\le$ is meant to be the partial order in the Löwner sense. That is, $\rho \geq \sigma$ if and only if $\rho - \sigma$ is positive semidefinite. \par 
We use $\tr[\,\cdot\,]$ for the usual matrix trace and $\norm{\,\cdot\,}_1 = \tr[\Cnorm{\,\cdot\,}]$ and $\norm{\,\cdot\,}_\infty$ to denote the trace norm and the spectral norm on $\cB(\cH)$, respectively. Quantum states in general are denoted by lower Greek letters such as $\rho, \sigma$ and $\tau$, for example. For Hermitian operators in $\cB(\cH)$ we usually use upper Latin letters such as $X, Y$. For any such $X$, we denote by $[X]_+$ and $[X]_-$ its positive and negative parts, respectively.

As we later want to formally control the dependence on the states $\rho$ and $\sigma$ that are given as arguments to the divergences, we further introduce $\cH \times \cH$ the cartesian product of the Hilbert space $\cH$ with itself. Moreover, on a bipartite system $\cH_{AB} = \cH_A \otimes \cH_B$, we set $\rho_A$ to be the state on $\cH_A$ that $\rho \in \cS(\cH_{AB})$ is mapped to under the partial trace with respect to the subsystem $B$ which is a completely positive trace-preserving (CPTP) map. Furthermore, we denote by $\identity_A$ the identity matrix on $A$ and, in a slight abuse of notation, we denote by $\tr_A[\cdot]$ both the partial trace with respect to $A$ as well as the complemented map  on $\cH_{AB}$ by tensorizing with $\identity_A$.

\subsection{Entropies and derived quantities}

The \textit{von Neumann entropy} of $\rho \in \cS(\cH)$ is given by 
\begin{equation}
    S(\rho) := -\tr[\rho \log(\rho)] \, .
\end{equation}
For two quantum states $\rho, \sigma \in \cS(\cH)$, their \textit{(Umegaki) relative entropy} \cite{Umegaki-RelativeEntropy-1962} is defined as 
\begin{equation}\label{eq:eq_relative_entropy}
    D(\rho \Vert \sigma) := \begin{cases}
        \tr[\rho \log\rho - \rho \log \sigma] & \text{if } \ker \sigma \subseteq \ker \rho \, ,\\
        + \infty & \text{otherwise} \, ,
    \end{cases}
\end{equation}
and their \textit{Belavkin-Staszewski (BS) entropy} \cite{BelavkinStaszewski-BSentropy-1982} by 
\begin{equation}
    \widehat{D}(\rho \Vert \sigma) := \begin{cases}
        \tr[\rho\log \rho^{1/2} \sigma^{-1} \rho^{1/2}] & \text{if } \ker \sigma \subseteq \ker \rho \, , \\
        + \infty & \text{otherwise} \, .
    \end{cases}
\end{equation}
In the event of $\rho$ and $\sigma$ commuting, the two entropies coincide. Otherwise, the BS-entropy is strictly larger than the relative entropy \cite{Hiai2017}. We further note that both can also be defined in terms of positive semi-definite operators $A, B$ (without normalisation), by just replacing $\rho$ with $A$ and $\sigma$ with $B$. We make use of this alternative definition when we define the conditional entropy and the BS-conditional entropy, for example. 
 Using this notation we can write the \textit{conditional entropy} of $\rho$ as
\begin{equation}\label{eq:eq_conditional_entropy}
    H_\rho(A|B) := S(\rho_{AB}) - S(\rho_{B}) = -D(\rho_{AB} \Vert \identity_A\otimes \rho_B) \, ,
\end{equation}
with the last equality being a straightforward calculation. The subscript $AB$ in $\rho_{AB} = \rho$ just emphasises the fact that $\rho$ stems from $\cS(\cH_{A} \otimes \cH_{B})$ and to distinguish it from its partial trace $\rho_B$, for example. It is noteworthy that the conditional entropy admits the following variational expression
\begin{equation}\label{eq:variational_expression_conditional_entropy}
    H_\rho(A|B) = \underset{\sigma_B \in \cS(\cH_B)}{\max} \,  -D(\rho_{AB} \Vert \identity_A\otimes \sigma_B) .
\end{equation}
Furthermore, in a similar manner as for the conditional entropy, one obtains the representation of the \textit{mutual information} in terms of the von Neumann entropy and the conditional entropy
\begin{equation}
    I_\rho(A:B) := S(\rho_A) + S(\rho_B) - S(\rho_{AB}) = S(\rho_A) - H_\rho(A|B)  = D(\rho_{AB} \Vert \rho_A \otimes \rho_B) \, .
    \label{eq:eq_mutual_information}
\end{equation}
Finally, on a tripartite system $\cH = \cH_A \otimes \cH_B \otimes \cH_C$ the \textit{conditional mutual information} of a state $\rho \in \cS(\cH)$ is given by 
\begin{equation}
\begin{aligned}
      I_\rho(A:B|C) & := S(\rho_{AC}) + S(\rho_{BC}) - S(\rho_C) - S(\rho_{ABC})\\
      & = H_\rho(A|C) - H_\rho(A|BC)\\
      &  = I_\rho(A:BC) - I_\rho(A:C) \, .
    \label{eq:eq_conditional_mutual_information}
\end{aligned}
\end{equation}
The last equalities are again straightforward. One easily checks that both the mutual information and the conditional mutual information are symmetric under the exchange of the $A$ and $B$ system.\par 
Let us proceed now to introduce the analogous quantities from the BS instead of the relative entropy. In this framework, we cannot construct them as sums and differences of von Neumann entropies, which, for every BS-entropic quantity, gives rise to a zoo of different possible definitions. Some of them have already appeared before in \cite{BluhmCapelPerezHernandez-WeakQFBSentropy-2021,Scalet2021,ZhaiYangXi-BSEntropy-2022}. For a bipartite state $\rho \in \cS(\cH_{A} \otimes \cH_B)$, inspired by the notion of conditional entropy, we define the \textit{BS-conditional entropy} as 
\begin{equation}
    \widehat{H}_\rho(A|B) := -\widehat{D}(\rho_{AB} \Vert \identity_A \otimes \rho_B) \, ,\label{eq:eq_BS_conditional_entropy}
\end{equation}
and building on the mutual information, we define the \textit{BS-mutual information} as 
\begin{equation}
    \widehat{I}_\rho(A:B) := \widehat{D}(\rho_{AB} \Vert \rho_A \otimes \rho_B) \, .
    \label{eq:eq_BS_mutual_information}
\end{equation}
Finally, the analogue of the conditional mutual information in this setting is a more subtle matter. Indeed, two natural ways to construct such a quantity would be either as a difference of BS-conditional entropies or of BS-mutual information, as shown in \cref{eq:eq_conditional_mutual_information}, which in general do not coincide. Given $\rho_{ABC} \in \mathcal{S}(\cH_A \otimes \cH_B \otimes \cH_C)$:
\begin{itemize}
    \item  We define the \textit{(one-sided) BS-conditional mutual information} (os BS-CMI in short) by 
    \begin{equation}\label{eq:eq_BS_conditional_mutual_information}
        \widehat{I}^{\text{os}}_\rho(A:B|C) :=\widehat{H}_\rho(A|C) - \widehat{H}_\rho(A|BC)  = \widehat{D}(\rho_{ABC} \Vert \identity_A \otimes \rho_{BC})  - \widehat{D}(\rho_{AC} \Vert \identity_A \otimes \rho_C) \, .
    \end{equation}
    \item We define the \textit{(two-sided) BS-conditional mutual information} (ts BS-CMI in short) by 
      \begin{equation}\label{eq:eq_BS_conditional_mutual_information_notuseful}
        \widehat{I}^{\text{ts}}_\rho(A:B|C) := \widehat{I}_\rho(A:BC) - \widehat{I}_\rho(A:C) = \widehat{D}(\rho_{ABC} \| \rho_{A} \otimes \rho_{BC} ) - \widehat{D}(\rho_{AC} \| \rho_{A} \otimes \rho_C ) \, .
    \end{equation}
\end{itemize}
Note that both notions are clearly non-negative, as a consequence of the data processing inequality for the BS-entropy. In this project, we focus for simplicity on the first definition, i.e.\ the one-sided one. We will therefore drop the ``os'' notation, as there is no possible confusion.

Let us emphasize at this stage that the difference between the aforementioned two definitions of BS-conditional mutual information is partly related to the pathological behaviour of the BS-entropy with respect to continuity in general, and more specifically to the fact that the BS-conditional entropy is discontinuous on the set of positive semi-definite quantum states (cf.\ \cref{prop:discontinuity_conditional_BS_entropy}). We suspect that as a consequence thereof, the variational definition of the BS-conditional entropy (generalizing \cref{eq:variational_expression_conditional_entropy}) does not agree with the one we have given in \cref{eq:eq_BS_conditional_entropy}, namely
\begin{equation} \label{eq:variational-cond-BS}
       \widehat{H}_\rho(A|B) \le \underset{\sigma_B \in \cS(\cH_B)}{\sup} -\widehat{D}(\rho_{AB} \Vert \identity_A \otimes \sigma_B)  \,.
\end{equation}
We have numerical results that suggest that the inequality in the \cref{eq:variational-cond-BS} is strict, at least in some cases.  A plot of those numerics can be found in \cref{sec:sec_breakdown_numerical_formula}. Moreover, we will indeed formally show that both quantities are different in general in \cref{rem:both_conditional_BS_entropies_are_different}.

%% file: sections/from-almost-concav-to-cont-bounds.tex
As depicted in \cref{fig:fig_flow_chart}, our approach is based on the convexity and almost concavity of a divergence.
More precisely, it is based on its joint convexity and almost joint concavity, but for the sake of better readability, we will just speak of convexity and almost concavity.\par
It is immediately clear what is meant by convexity and this property is often even a defining property of a divergence \cite{HiaiMosonyi_2011} or a direct consequence thereof\footnote{Some authors define divergences as functions on two density operators fulfilling a data processing inequality; however, note that convexity for a divergence implies a data processing inequality and follows from it together with additional properties, as shown in \cite[Corollary 4.7]{HiaiMosonyi_2011}.} \cite[Proposition 4.2]{Tomamichel_Book_2016}. The almost (joint) concavity, however, 
needs yet to be defined. 

\begin{definition}[Almost (joint) concavity of a divergence]\label{def:almost_concavity_divergence}
    A divergence $\D(\cdot\Vert\cdot)$ is called \textit{almost (jointly) concave} on a convex set $\cS_0 \subseteq \cS(\cH) \times \cS(\cH)$ if, for $(\rho_1, \sigma_1),$ $ (\rho_2, \sigma_2) \in \cS_0$, there exists a continuous function $f:[0, 1] \to \R$ with $f(0) = f(1) = 0$  such that, for all $p\in [0,1]$,
    \begin{equation}
        \D(\rho \Vert \sigma) \ge p \D(\rho_1 \Vert \sigma_1) + (1 - p) \D(\rho_2 \Vert \sigma_2) - f(p)\label{eq:eq_almost_concavity}
    \end{equation}
    holds. Here, $\rho = p \rho_1 + (1 - p) \rho_2$ and $\sigma = p \sigma_1 + (1 - p) \sigma_2$. It is important to emphasise that $f$ in general depends on the states involved.
\end{definition}

\begin{remark}
    We note that the definition of almost concavity presented above is not itself a very strong property. For example, one could just choose $f$ to be the remainders that give equality in \cref{eq:eq_almost_concavity}. It is the behaviour of the remainder functions that is pivotal, i.e., it becomes independent of $\rho_i, \sigma_i$, $i = 1, 2$ under certain restrictions on the states, e.g.\ requiring that $\sigma_i$ is a marginal of $\rho_i$.
\end{remark}

Our approach, therefore, does not only need joint convexity but a well-behaved remainder function. If we find such a function and combine it with the boundedness of the divergence (or underlying entropic quantity), ALAFF directly gives uniform continuity through explicit continuity bounds.

As we already discussed in the introduction, the predecessor of ALAFF was developed and used by Alicki and Fannes \cite{AlickiFannes-2004}, as well as Winter \cite{Winter-AlickiFannes-2016}, to prove uniform continuity and give an explicit continuity bound for the conditional entropy. Shirokov then noticed the potential beyond this specific application and moulded a method that can be applied to functions defined on convex and $\Delta$-invariant subsets of $\cS(\cH)$ \cite{Shirokov-AFWmethod-2020, Shirokov-ContinuityReview-2022}. Independently, similar techniques were already used in \cite{mosonyi2011quantum}. In short, $\Delta$-invariance means that for two elements their normalised positive and negative part again lies in the set (see also \Cref{def:def_perturbed_delta_invariant_subset}). This definition of $\Delta$-invariance will, however, turn out to be a limitation when trying to prove the uniform continuity of the relative entropy, while in the case of the BS-entropy, it is unfitting even from the beginning, i.e., even for the BS-conditional entropy. The problem is due to  $\Delta$-invariance being a rather strong property that sets like $\cS_{\ge m}(\cH)$ or $\{(\rho, \sigma) \;:\; \ker \sigma \subseteq \ker \rho\}$ do not have. Yet, those sets, or modified versions thereof, are the relevant sets for the relative and, in particular, the BS-entropy.

In light of those problems and in an effort to make our approach as general as possible, we propose the almost locally affine (ALAFF) method, a generalisation of the Alicki-Fannes-Winter-Shirokov method that reduces to one implication of the former in a special case. First of all, we define a perturbed version of the $\Delta$-invariant subset, with the perturbation controlled by a parameter $s$.

\begin{definition}[Perturbed $\Delta$-invariant subset]\label{def:def_perturbed_delta_invariant_subset}
    Let $s \in [0, 1)$. A subset $\cS_0 \subseteq \cS(\cH)$ is called \textit{$s$-perturbed $\Delta$-invariant}, if for $\rho, \sigma \in \cS_0$ with $\rho \ne \sigma$ there exists $\tau \in \cS(\cH)$ such that the two states
    \begin{equation}
        \Delta^\pm(\rho, \sigma, \tau) = s \tau + (1 - s)\varepsilon^{-1}[\rho - \sigma]_\pm
        \label{eq:eq_delta_states}
    \end{equation}
    lie again in $\cS_0$. Here $\varepsilon := \frac{1}{2} \norm{\rho - \sigma}_1$ and $[A]_\pm$ denotes the negative and positive part of a self-adjoint operator, respectively. For $s=0$, we recover the definition of $\Delta$-invariant subset used in \cite{Shirokov-ContinuityReview-2022}.
\end{definition}

We want to give the reader some intuition about those $s$-perturbed $\Delta$-invariant sets.

\begin{remark}\label{rem:rem_s_perturbed_delta_invariance}
    \begin{enumerate}
        \item Let $\cS_{0} \subseteq \cS(\mathcal{H})$ be $s$-perturbed $\Delta$-invariant convex set. Then for $t \in [s, 1)$ it is $t$-perturbed $\Delta$-invariant as well. In particular, being $0$-perturbed is the strongest condition.
        \item If $\cS_0 \subseteq \cS(\cH)$ has non-empty interior with respect to the $1$-norm, then it is $s$-perturbed for some $s \in [0, 1)$.
        \item If $\cS_0 \subseteq \cS(\cH)$ is $s$-perturbed $\Delta$-invariant containing more than one state, then there exist $\rho, \sigma \in \cS_0$ with $\frac{1}{2}\norm{\rho - \sigma}_1 = 1 - s$. This follows directly from the definition.
    \end{enumerate}
\end{remark}

In order to get well-behaved remainder functions, we define a stronger property that we call ``almost local affinity".

\begin{definition}[Almost locally affine (ALAFF) function] \label{def:ALAFF-function}
    Let $f$ be a real-valued function on the convex set $\cS_0 \subseteq \cS(\cH)$, fulfilling
    \begin{equation}
        -a_f(p) \le f(p \rho + (1 - p) \sigma) - p f(\rho) - (1 - p) f(\sigma) \le b_f(p) \label{eq:eq_alaff_property}
    \end{equation}
    for all $p \in [0, 1]$ and $\rho, \sigma \in \cS_0$. The functions $a_f:[0, 1] \to \R$ and $b_f:[0, 1] \to \R$ are required to vanish as $p \to 0^+$, to be non-decreasing on $[0, \frac{1}{2}]$, continuous in $p$ and independent of $\rho, \sigma \in \cS_0$. We then call $f$ an \textit{almost locally affine (ALAFF)} function.
\end{definition}

The notion of almost locally affine functions has appeared previously in the literature, also under the name ``approximate affinity" (see e.g.\ \cite{brandao2013resource}). We can now formulate the following theorem, whose proof is inspired by Shirokov \cite{Shirokov-ContinuityReview-2022}.

\begin{theorem}[Almost locally affine (ALAFF) method]\label{theo:theo_alaff_method}~\\
    Let $s \in [0, 1)$ and $\cS_0 \subseteq \cS(\cH)$ be an $s$-perturbed $\Delta$-invariant convex subset of $\cS(\cH)$ containing more than one element. Let further $f$ be an ALAFF function. We then find that $f$ is uniformly continuous if
    \begin{equation}
        C_f^s := \sup\limits_{\twoline{\rho, \sigma \in \cS_0}{\frac{1}{2}\norm{\rho - \sigma}_1 = 1 - s}} |f(\rho) - f(\sigma)| < + \infty.
    \end{equation}
    In this case, we have for $\varepsilon \in (0, 1]$
    \begin{equation}
        \begin{aligned}
            \sup\limits_{\twoline{\rho, \sigma \in \cS_0}{\frac{1}{2}\norm{\rho - \sigma}_1 \le \varepsilon}} |f(\rho) - f(\sigma)| \le C_f^s \frac{\varepsilon}{1 - s} + \frac{1 - s + \varepsilon}{1 - s} E_f^{\max}\Big(\frac{\varepsilon}{1 - s + \varepsilon}\Big) \, , \label{eq:eq_continuity_bound}
        \end{aligned}
    \end{equation} 
    with 
    \begin{equation}
        E_f^{\max}:[0, 1) \to \R, \qquad p \mapsto E_f^{\max}(p) = (1 - p) \max\left\{\frac{E_f(t)}{1 - t} \;:\; 0 \le t \le p\right\} \, ,
    \end{equation}
    where $E_f = a_f + b_f$. Note that on $\varepsilon \in (0, 1 - s]$ $E_f$ and $E_f^{\max}$ coincide.
\end{theorem}
\begin{proof}[Proof]
    Let $s \in [0,1)$ and $\varepsilon \in (0, 1]$. Let further $\rho, \sigma \in \cS_0$ with $\frac{1}{2}\norm{\rho - \sigma}_1 = \varepsilon$. Then by the property of $s$-perturbed $\Delta$-invariance there exists $\tau \in \cS(\cH)$ such that $\gamma_\pm := \Delta^\pm(\rho, \sigma, \tau) \in \cS_0$ defined as in \cref{eq:eq_delta_states}.  For every such $\gamma_\pm$ with a representation in terms of $\rho$, $\sigma$ $\in \cS_0$ and a $\tau \in \cS(\cH)$ we have that
    \begin{equation}
        \frac{1 - s}{1 - s + \varepsilon}\rho + \frac{\varepsilon}{1 - s + \varepsilon} \gamma_- = \omega^* = \frac{1 - s}{1 - s + \varepsilon} \sigma + \frac{\varepsilon}{1 - s + \varepsilon} \gamma_+ \, ,
    \end{equation}
    which can be easily checked by inserting the explicit form of $\gamma_\pm$ and using that $[\rho - \sigma]_+ - [\rho - \sigma]_- = \rho - \sigma$. Now $\omega^* \in \cS_0$ as $\cS_0$ is convex, which allows us to evaluate $f$ at $\omega^*$ and use \cref{eq:eq_alaff_property} for both of the representations we have for the state in question. This gives us
    \begin{equation}
        \begin{aligned}
            -a_f(p) &\le f(\omega^*) - (1 - p) f(\rho) - p f(\gamma_-) \le b_f(p) \, ,\\
            -a_f(p) &\le f(\omega^*) - (1 - p) f(\sigma) - p f(\gamma_+) \le b_f(p)\, ,
        \end{aligned}
    \end{equation}
    where we set $p = p(\varepsilon) =\frac{\varepsilon}{1 - s + \varepsilon}$ for better readability. Note that $p \in (0, \frac{1}{2 - s}] \subseteq [0, 1)$ as $\varepsilon \in (0, 1]$ and $s \in [0, 1)$ and further that $p(\varepsilon)$ is monotone with respect to $\varepsilon$. We recombine the above to get
    \begin{equation}
        \begin{aligned}
            (1 - p) (f(\rho) - f(\sigma)) &\le p (f(\gamma_+) - f(\gamma_-)) + a_f(p) + b_f(p) \, ,\\
            (1 - p) (f(\sigma) - f(\rho)) &\le p (f(\gamma_-) - f(\gamma_+)) + a_f(p) + b_f(p) \, .
        \end{aligned}
    \end{equation}
    Those two inequalities immediately give us
    \begin{equation}
        (1 - p) |f(\rho) - f(\sigma)| \le p |f(\gamma_+) - f(\gamma_-)| + (a_f + b_f)(p) \, .
    \end{equation}
    If we now insert $E_f = a_f + b_f$, we obtain
    \begin{equation}
        |f(\rho) - f(\sigma)| \le \frac{p}{1 - p} |f(\gamma_+) - f(\gamma_-)| + \frac{1}{1 - p} E_f(p)\, . \label{eq:eq_base_inequality}
    \end{equation}
    
    In the case that $C_f^s$ is finite, we can take the supremum over all $\rho$, $\sigma \in \cS_0$ with $\frac{1}{2}\norm{\rho - \sigma}_1 = \varepsilon$ of the last equation and even extend to $\frac{1}{2}\norm{\rho - \sigma}_1 \le \varepsilon$ in two steps. The first step is upper bounding $\frac{1}{1 - p}E_f(p)$ with $\frac{1}{1 - p}E_f^{\max}(p)$ and the second one using that $\frac{1}{1 - p} E_f^{\max}(p)$ is engineered to be non-decreasing on $[0, 1)$ and thereby for the specific $p = \frac{\varepsilon}{1 - s + \varepsilon} \in [0, \frac{1}{2 - s}] \subset [0, 1)$, is non-decreasing in $\varepsilon$ as well. Since the $\gamma_+$ and $\gamma_-$ created from $\rho$ and $\sigma$ obviously fulfill $\gamma_\pm \in \cS_0$ and $\frac{1}{2}\norm{\gamma_+ - \gamma_-}_1 = 1 - s$, we immediately get the upper bound in \cref{eq:eq_continuity_bound}. The reduction of $E_f^{\max}$ to $E_f$ on $\varepsilon \in (0, 1 - s]$ follows immediately from $E_f$ being non-decreasing on $[0, \frac{1}{2}]$ meaning further that $E_f^{\max}$ inherits the vanishing property of $E_f$ as $p \to 0^+$. This directly translates into $E_f^{\max}(p(\varepsilon)) \to 0$ if $\varepsilon \to 0^+$, hence concluding the proof of uniform continuity. 
\end{proof}

The method presented in \cref{theo:theo_alaff_method} is named the ``ALAFF method" to highlight the required ALAFF property necessary, for this technique to be applicable. We will refer to this theorem by that name in subsequent sections. 

\begin{remark}\label{rem:rem_s=0_alaff_method}
    For $s = 0$, one recovers one implication of the method by Shirokov, i.e., the definitions for perturbed $\Delta$-invariance and $\Delta$-invariance coincide, $E_f^{\max}$ reduces to $E_f$ on the relevant domain $\varepsilon \in [0, 1]$, and \cref{eq:eq_continuity_bound} becomes 
    \begin{equation}
        \sup\limits_{\twoline{\rho, \sigma \in \cS_0}{\frac{1}{2}\norm{\rho - \sigma}_1 \le \varepsilon}} |f(\rho) - f(\sigma)| \le  C^\perp_f \varepsilon + (1 + \varepsilon) E_f\Big(\frac{\varepsilon}{1 + \varepsilon}\Big)
    \end{equation}
    with
    \begin{equation}
        C^0_f = \sup\limits_{\twoline{\rho, \sigma \in \cS_0}{\frac{1}{2}\norm{\rho - \sigma}_1 = 1}} |f(\rho) - f(\sigma)| = \sup\limits_{\twoline{\rho, \sigma \in \cS_0}{\tr[\rho \sigma] = 0}} |f(\rho) - f(\sigma)| =: C^\perp_f,
    \end{equation}
    as states with maximal trace distance have orthogonal support.
\end{remark}

In the next sections, we will use \cref{theo:theo_alaff_method} together with the almost concavity of the relative entropy and the BS-entropy, respectively, to derive a plethora of results of uniform continuity and continuity bounds for entropic quantities defined through them. Depending on the case, we will sometimes have to employ the whole machinery devised in \cref{theo:theo_alaff_method}, whereas at other times the simplification provided in \cref{rem:rem_s=0_alaff_method} will be enough.

%% file: sections/almost-concav-rel-ent.tex
In this section, we apply the ALAFF method introduced in \cref{sec:ALAFF_method} for the particular case of the relative entropy, as well as some other entropic quantities derived from it.

All the results provided in this section are summarized in \cref{fig:fig_flow_chart_RE}.

\begin{figure}[ht!]
    \centering
    \scalebox{0.8}{
    \begin{tikzpicture}[font=\small,thick,scale=1, every node/.style={scale=1}]
        \node[draw ,
            diamond,
            fill=midblue!30!white,
            minimum width=2.5cm,
            minimum height=1cm,
            inner sep=0.1cm,
            align=center] (block1) {Relative\\entropy\\[1mm]$D(\rho\Vert\sigma)$};
            
        \node[draw,
            rounded rectangle,
            fill=cyan!30!white,
            below left=of block1,
            minimum height=1.38cm,
            minimum width=2.5cm,
            inner sep=0.1cm,
            align=center] (block2) {Convexity\\{\tiny$p D(\rho_1\Vert \sigma_1) + (1 - p) D(\rho_2\Vert \sigma_2) \ge D(\rho\Vert \sigma)$}};
            
        \node[draw,
            rounded rectangle,
            fill=cyan!30!white,
            below right=of block1,
            minimum width=2.5cm,
            inner sep=0.1cm,
            align=center] (block3) {Almost concavity\\{\tiny$D(\rho\Vert\sigma) \ge p D(\rho_1\Vert \sigma_1) + (1 - p) D(\rho_2\Vert \sigma_2) - f(p)$}\\{\tiny with $f(p)= h(p)\frac{1}{2}\norm{\rho_1 - \rho_2}_1 + f_{c_1, c_2}(p)$}};
         
        \node[draw,
            trapezium, 
            fill=lightgreen!50!white,
            below=2cm of block1,
            trapezium left angle = 65,
            trapezium right angle = 115,
            trapezium stretches,
            minimum width=3.5cm,
            minimum height=1cm,
            inner sep=0.1cm,
            align=center] (block4) {ALAFF method\\[1mm]{\footnotesize \cref{theo:theo_alaff_method} and \cref{rem:rem_s=0_alaff_method}}};
        
        \node[coordinate,below=0.7cm of block4] (block5) {};
        
        \node[draw,
            below=of block5,
            fill=midgreen!50!white,
            minimum width=2.5cm,
            minimum height=1cm,
            align=center] (block6) {Conditional\\mutual information\\[1.2mm]{\tiny$ \mathbb  |I_\rho(A:B | C) - I_\sigma(A:B|C)|  $}\\{\tiny$ \le 2 \varepsilon \log\min\{ d_A, d_B\}$}\\{\tiny$+ 2(1 + \varepsilon) h \Big(\frac{\varepsilon}{1 + \varepsilon}\Big) $}\\[1mm]{\tiny with $\varepsilon \le \frac{1}{2}\| \rho - \sigma \|_1 $}};
         
        \node[draw,
            right=0.5cm of block6,
            fill=midgreen!50!white,
            minimum width=2.5cm,
            minimum height=1cm,
            align=center] (block7) {Mutual information\\[1.2mm]{\tiny$ | I_\rho(A:B) -  I_\sigma(A:B) |  $}\\{\tiny$ \leq  2 \varepsilon \log\min\{ d_A, d_B\}$}\\{\tiny $+ 2(1 + \varepsilon) h\Big(\frac{\varepsilon}{1 + \varepsilon}\Big) $}\\[1mm]{\tiny with $\varepsilon \le \frac{1}{2}\| \rho - \sigma \|_1 $}};
        
        \node[draw,
            right=0.5cm of block7,
            fill=midgreen!50!white,
            minimum width=2.5cm,
            minimum height=1cm,
            align=center] (block8) {Conditional entropy\\[1.2mm]{\tiny$| H_\rho(A|B) - H_\sigma(A|B) | \leq  $}\\{\tiny$  2 \varepsilon \log d_A + (1 + \varepsilon)h\Big(\frac{\varepsilon}{1 + \varepsilon}\Big) $}\\[1mm]{\tiny with $\varepsilon \le \frac{1}{2}\| \rho - \sigma \|_1 $}};
            
        \node[draw,
            left=0.5cm of block6,
            fill=midgreen!50!white,
            minimum width=2.5cm,
            minimum height=1cm,
            align=center] (block9) {Relative entropy\\(fixed second argument)\\[1.2mm]{\tiny$ \left| D (\rho_1 \| \sigma ) -  D (\rho_2 \| \sigma ) \right| $}\\{\tiny$  \leq \varepsilon \log \widetilde{m}_\sigma^{-1}  + (1 + \varepsilon) h \Big(\frac{\varepsilon}{1 + \varepsilon}\Big) $}\\[1mm]{\tiny with $\varepsilon \le \frac{1}{2}\| \rho_1 - \rho_2 \|_1 $}};
        
        \node[draw,
            left=0.5cm of block9,
            fill=midgreen!50!white,
            minimum width=2.5cm,
            minimum height=1cm,
            align=center] (block10) {Relative entropy\\(fixed first argument)\\[1.2mm]{\tiny$ \left| D (\rho \| \sigma_1 ) -  D (\rho \| \sigma_2 ) \right| $}\\{\tiny$  \leq f_{RE,1}(\| \sigma_1 - \sigma_2 \|_1)$}};

        \node[circle, draw, fill,inner sep=1.5pt,below=0.4cm of block10] (block11) {};
        
        \node[draw,
            below=1.52cm of block10,
            fill=midgreen!50!white,
            minimum width=2.5cm,
            minimum height=1cm,
            align=center] (block12) {Relative entropy\\[1.2mm]{\tiny$ \left| D (\rho_1 \| \sigma_1 ) -  D (\rho_2 \| \sigma_2 ) \right| $}\\{\tiny$  \leq f_{RE}(\| \rho_1 - \rho_2 \|_1,\| \sigma_1 - \sigma_2 \|_1)$}};

        \node[draw,
            below=1.3cm of block9,
            fill=midgreen!50!white,
            minimum width=2.5cm,
            minimum height=1cm,
            align=center] (block13) {Divergence bound\\[1.2mm]{\tiny$  D (\rho \| \sigma )$}\\{\tiny$  \leq \varepsilon \log \widetilde{m}_\sigma^{-1}  + (1 + \varepsilon) h \Big(\frac{\varepsilon}{1 + \varepsilon}\Big) $}\\[1mm]{\tiny with $\varepsilon \le \frac{1}{2}\| \rho - \sigma \|_1 $}};
        
        \draw[-latex] (block1) -| (block2)
            node[pos=0.25,fill=white,inner sep=0]{\cite{Lindblad-ConvexityRE-1974}};
         
        \draw[-latex] (block1) -| (block3)
            node[pos=0.25,fill=white,inner sep=0]{\cref{theo:theo_almost_concavity_relative_entropy}};

        \draw[-latex] (block2) -| (block4);
        \draw[-latex] (block3) -| (block4);
    
        \draw[-Turned Square] (block5) -| (block6)
        node[pos=0.77,fill=white,inner sep=0]{\cref{cor:continuity_bound_CMI}};
        \draw[-Turned Square] (block5) -| (block7)
         node[pos=0.74,fill=white,inner sep=0]{\cref{cor:continuity_bound_mutual_information}};
        \draw[-Turned Square] (block5) -| (block8)
          node[pos=0.74,fill=white,inner sep=0]{\cref{cor:continuity_bound_conditional_entropy}};
        \draw[-Turned Square] (block5) -| (block9)
         node[pos=0.73,fill=white,inner sep=0]{\cref{cor:cor_uniform_continuity_relative_entropy_first_argument}};
        \draw[-Turned Square] (block5) -| (block10)
         node[pos=0.73,fill=white,inner sep=0]{\cref{cor:cor_uniform_continuity_relative_entropy_second_argument}};
        
        \draw[-latex] (block10) -- (block12)
        node[pos=0.65,fill=white,inner sep=0]{\cref{theo:theo_uniform_continuity_relative_entropy}};
        \draw[-latex] (block9) -- (block13)
         node[pos=0.65,fill=white,inner sep=0]{\cref{cor:divergence_bound_relative_entropy}};
        \draw (block9) |- (block11);
        
        \draw (block4) -- (block5) 
            node[pos=1,fill=white,inner sep=0.1cm] {Uniform continuity \& Continuity bounds};
    \end{tikzpicture}
    }
    \caption{In this flow chart we collect the main results from this chapter, starting with the almost concavity of the relative entropy, which together with the ALAFF method outputs a collection of continuity bounds for related entropic quantities. For the convexity and almost concavity, we are setting $\rho= p \rho_1 + (1-p) \rho_2$ and $\sigma= p \sigma_1 + (1-p) \sigma_2$, with $p \in [0,1]$. We denote by $\widetilde{m}_\sigma$ the minimal non-zero eigenvalue of $\sigma$. The specific bounds obtained for the relative entropy fixing the first argument and in the general case (modifying both arguments) are omitted due to their technicality.}
    \label{fig:fig_flow_chart_RE}
\end{figure}

\subsection{Almost concavity for the relative entropy}\label{subsec:almost_concavity_relative_entropy}

The (joint) convexity of the relative entropy is a well-established result with proofs found for example in \cite{WildeFromClassicalToQuantumInformation_2016}. In this section, we complement this result with almost concavity and further prove that the bound we obtain is tight.

\begin{theorem}[Almost concavity of the relative entropy]\label{theo:theo_almost_concavity_relative_entropy}~\\
    Let $(\rho_1, \sigma_1), (\rho_2, \sigma_2) \in \cS_{\ker}$ with 
    \begin{equation}
        \cS_{\ker} := \{(\rho, \sigma) \in \cS(\cH) \times \cS(\cH) \;:\; \ker \sigma \subseteq \ker \rho\}
    \end{equation} and $p \in [0, 1]$. Then, for $\rho = p \rho_1 + (1 - p) \rho_2$ and $\sigma = p \sigma_1 + (1 - p) \sigma_2$,
    \begin{equation}\label{eq:almost_concavity_relative_entropy}
        D(\rho \Vert \sigma) \ge p D(\rho_1 \Vert \sigma_1) + (1 - p) D(\rho_2 \Vert \sigma_2) - h(p)\frac{1}{2}\norm{\rho_1 - \rho_2}_1 - f_{c_1, c_2}(p) \, . 
    \end{equation}
    Here, 
    \begin{equation}
        \begin{aligned}
            h(p) &= - p \log(p) - (1 - p) \log(1 - p) \, ,\\
            f_{c_1, c_2}(p) &= p \log(p + (1 - p) c_1) + (1 - p) \log((1 - p) + p c_2) \, ,
        \end{aligned}
    \end{equation}
    with the first one being the binary entropy. The constants in $f_{c_1, c_2}$ are non-negative real numbers and are given by
    \begin{equation}
        \begin{aligned}
            c_1 &:= \int\limits_{-\infty}^\infty dt \beta_0(t) \tr\left[\rho_1 \sigma_1^{\frac{it - 1}{2}} \sigma_2  \sigma_1^{\frac{-it-1}{2}}\right] < \infty \, , \\
            c_2 &:=  \int\limits_{-\infty}^\infty dt \beta_0(t) \tr\left[\rho_2 \sigma_2^{\frac{it - 1}{2}} \sigma_1 \sigma_2^{\frac{-it-1}{2}}\right] < \infty \, .
        \end{aligned}
    \end{equation}
    Here, $\beta_0$ is a probability density on $\mathbb R$ (see  \cref{eq:beta0} for a concrete expression). It is noteworthy that $f_{1, 1}(\cdot) = 0$ and $f_{c_1, c_2}(0) = f_{c_1, c_2}(1) = 0$. 
\end{theorem}
\begin{proof}[Proof.]
    It is clear that $\cS_{\ker}$ is a convex set and that the bound holds trivially for $p = 0$ and $p = 1$. Hence let $p \in (0, 1)$ and $(\rho_1, \sigma_1), (\rho_2, \sigma_2) \in \cS_{\ker}$ in the following. We find that
    \begin{equation}
        \begin{aligned}
            p D(\rho_1 \Vert \sigma_1) + (1 - p) D(\rho_2 \Vert \sigma_2) - D(\rho \Vert \sigma) &= - p S(\rho_1) - (1 - p) S(\rho_2) + S(\rho)\\
            &\hspace{1cm} + (1 - p) \tr[\rho_2(\log\sigma -
                \log\sigma_2)]\\
            &\hspace{2cm} + p \tr[\rho_1(\log\sigma - \log\sigma_1)]\\
            &\le h(p)\frac{1}{2}\norm{\rho_1 - \rho_2}_1 + f_{c_1, c_2}(p) \, ,
        \end{aligned}
    \end{equation}
    where we split the relative entropies and used that the von Neumann entropy fulfils \cite[Theorem 14]{audenaert2014quantum}
    \begin{equation}\label{eq:eq_almost_convexity_von_neumann_entropy}
        S(\rho) \le \frac{1}{2}\norm{\rho_1 - \rho_2}_1 h(p) + p S(\rho_1) + (1 - p) S(\rho_2) \, . 
    \end{equation} 
    Furthermore, we upper bound the remaining terms by $f_{c_1, c_2}(p)$, estimating the two separately. We will only demonstrate the derivation for the second term, as it is completely analogous to the first one. We have 
    \begin{equation}
        \begin{aligned}\label{eq:eq_int_ineq}
            p\tr[\rho_1 (\log(\sigma) - \log(\sigma_1))] &= p\tr\left[\exp(\log(\rho_1))(\log(\sigma) - \log(\sigma_1))\right]\\
            &\le p \log\tr\left[\exp\left(\log(\rho_1) + \log(\sigma) - \log(\sigma_1)\right)\right]\\
            &\le p \log\int\limits_{-\infty}^\infty dt\,\beta_0(t)\,\tr\left[\rho_1 \sigma_1^{\frac{it - 1}{2}} \sigma \sigma_1^{\frac{-it -1}{2}}\right] \, .
        \end{aligned}  
    \end{equation}
    The first estimate follows immediately using the well-known Peierls-Bogoliubov inequality \cite{OhyaPetz-Entropy-1993}. The second one involves a generalisation of the Araki-Lieb-Thirring inequality \cite{Araki_LiebThirring_1990, LiebThirring_1976} by Sutter et al. \cite[Corollary 3.3]{SutterBertaTomamichel-Multivariate-2017} with
    \begin{equation}\label{eq:beta0}
        \beta_0(t) = \frac{\pi}{2}\frac{1}{\cosh(\pi t) + 1}
    \end{equation}
    a probability density on $\R$. In the above steps, i.e. \cref{eq:eq_int_ineq}, we relied on $\rho_1, \sigma_1$ and $\sigma$ to be full rank. If this is not the case one obtains the same result, however, the procedure is more involved. A thorough discussion can be found in \cref{sec:sec_supplements_to_proof_almost_concavity_relative_entropy}. Note here that in the most general case, $\cdot^{-1}$ in the RHS of \cref{eq:eq_int_ineq} is the Moore-Penrose pseudoinverse. The trace in the integral can now be estimated for each $t$ by
    \begin{equation}
        \begin{aligned}\label{eq:eq_splitting_of_trace}
            \tr[\rho_1 \sigma_1^{\frac{it - 1}{2}} \sigma \sigma_1^{\frac{-it - 1}{2}}] &= p + (1 - p)\tr[\rho_1 \sigma_1^{\frac{it - 1}{2}} \sigma_2 \sigma_1^{\frac{-it - 1}{2}}] \, . 
        \end{aligned}
    \end{equation}
    Here, we just split $\sigma$ and used that $\tr[\rho_1] = 1$. To see that $c_1 < \infty$, we upper bound $\sigma_2$ by $\identity$ and $\sigma_1^{-1}$ by $\widetilde{m}_{\sigma_1}^{-1}\identity$ where $\widetilde{m}_{\sigma_1}$ is the smallest non-zero eigenvalue of $\sigma_1$. This can be done, since $\ker \sigma_1 \subseteq \ker \rho_1$. We end up with $c_1 \le \widetilde{m}_{\sigma_1}^{-1} < \infty$. Inserting \cref{eq:eq_splitting_of_trace} into \cref{eq:eq_int_ineq}, we obtain the first part of $f_{c_1, c_2}(p)$ and repeating the steps for $(1 - p)\tr[\rho_2(\log(\sigma) - \log(\sigma_2))]$ the second one as well. This concludes the proof.
\end{proof}

We remark that \cref{eq:almost_concavity_relative_entropy} provides a result of almost concavity for the relative entropy in the sense of \cref{def:almost_concavity_divergence}. Indeed, the additive ``correction'' term obtained for such an inequality to hold behaves well enough, in the sense that it reduces to the previously known bounds for quantities derived from the relative entropy, e.g.\ the von Neumann entropy or the conditional entropy, and it is almost tight in general. To illustrate that, we provide now two propositions that put the almost concavity of the relative entropy into perspective.

\begin{proposition}[Almost concavity estimate of the relative entropy is well behaved]\label{prop:prop_almost_concave_estimate_relative_entropy_well_behaved}~\\
    The function $f_{c_1, c_2} + h\frac{1}{2}\norm{\rho_1 - \rho_2}_1$ obtained in \cref{theo:theo_almost_concavity_relative_entropy} is well behaved in the following sense: For $j = 1, 2$ and $(\rho_j, \sigma_j) \in \cS_{\ker}$, we have the following:
    \begin{enumerate}
        \item If $\sigma_1 = \sigma_2$, then $c_1=c_2 = 1$, resulting in $f_{c_1, c_2} + \frac{1}{2} \norm{\rho_1 - \rho_2}_1 h \le h$.
        \item If for $\widetilde m > 0$ we have $\widetilde m \rho_j \le \sigma_j$, then $f_{c_1, c_2} + h \frac{1}{2} \norm{\rho_1 - \rho_2}_1 \le f_{\widetilde{m}^{-1}, \widetilde{m}^{-1}} + h$.
        \item If $\cH = \cH_A\otimes \cH_B$ is a bipartite space and furthermore $\sigma_j =d_A^{-1} \identity_A \otimes \rho_{j,B}$, then $f_{c_1, c_2} + h\frac{1}{2}\norm{\rho_1 - \rho_2}_1 \le h$.
        \item For $m_1, m_2 \ge 1$ we find that both $p \mapsto \frac{1}{1 - p} f_{m_1, m_2}(p)$ and $p \mapsto \frac{1}{1 - p} h(p)$ are non-decreasing on $[0, 1)$. 
    \end{enumerate}
\end{proposition}

We hence find that in the cases listed above the bound becomes independent of the states and further that the remainder functions have a desirable non-decreasing property. The proof is straightforward and can be found in \cref{sec:sec_proof_almost_concave_estimate_relative_entropy_well_behaved}.

\begin{remark}\label{rem:reduction_previous_continuity_bounds_relative_entropy}
    The different cases discussed in \cref{prop:prop_almost_concave_estimate_relative_entropy_well_behaved}
    are used in the following to find almost concavity results with a function that does not depend on the specifics of the states involved, as necessary for applying the ALAFF method.
    \begin{itemize}
        \item If $\sigma_1 = \sigma_2$, we are reducing \cref{eq:almost_concavity_relative_entropy} to a result of almost concavity only in the first input. This case was addressed in \cite{brandao2013resource}, where they obtained $h(p)$ as a correction for almost concavity, a bound we are tightening here. Moreover, this case will yield a continuity bound for the relative entropy with fixed second  input as shown in \cref{cor:cor_uniform_continuity_relative_entropy_first_argument}.
        \item Point 3 of \cref{prop:prop_almost_concave_estimate_relative_entropy_well_behaved} can be interpreted as a result of almost convexity for the conditional entropy. Moreover, it will yield a continuity bound for the conditional entropy in \cref{cor:continuity_bound_conditional_entropy}. Since the latter result is almost tight, this shows the good behaviour of the bound obtained in \cref{theo:theo_almost_concavity_relative_entropy}.
        \item Point 2  of \cref{prop:prop_almost_concave_estimate_relative_entropy_well_behaved} is the most general setting for full-rank states $\sigma_j$, with $j=1,2$, and will be essential for deriving the most general continuity bounds for the relative entropy in \cref{theo:theo_uniform_continuity_relative_entropy}.
    \end{itemize}
\end{remark}

Finally,  we conclude this subsection with some discussion of our almost concave bound. 

\begin{proposition}[Almost concavity estimate of the relative entropy is tight]~\\
    The bound presented in \cref{theo:theo_almost_concavity_relative_entropy} is tight. More specifically, there are some density operators $\rho_1, \rho_2, \sigma_1, \sigma_2$ on $\mathcal{S}(\cH)$ which saturate the inequality in \cref{eq:almost_concavity_relative_entropy}. 
\end{proposition}
\begin{proof}[Proof.]
     We can assume that the dimension of the underlying Hilbert space is $d_{\cH} \ge 2$. We then find two orthonormal states $\ket{0}, \ket{1} \in \cH$ that we use to create
    \begin{equation}
        \begin{aligned}
            \rho_1 &:= \dyad{0},\\
            \rho_2 &:= \dyad{1},\\
            \sigma_1 &:= t \dyad{0} + (1 - t) \dyad{1},\\
            \sigma_2 &:= (1 - t) \dyad{0} + t \dyad{1},
        \end{aligned}
    \end{equation}
    for $t \in (0, 1)$. We find, as of the orthonormality, that for $p \in [0, 1]$ and
    \begin{equation}
        \begin{aligned}
            \rho &:= p \rho_1 + (1 - p) \rho_2 \, ,\\
            \sigma &:= p \sigma_1 + (1 - p) \sigma_2 \, ,
        \end{aligned}
    \end{equation}
    the relative entropy between the states given by the convex combinations takes the value
    \begin{equation}
        \begin{aligned}
             D(\rho \Vert \sigma) &= \tr[\rho \log(\rho) - \rho \log(\sigma)]\\ 
             &= - h(p) - p \log(p t + (1 - p)(1 - t)) - (1 - p) \log((1 - p) t + p (1 - t)) \, ,
        \end{aligned}
    \end{equation}
    and 
    \begin{equation}
        \begin{aligned}
            D(\rho_1 \Vert \sigma_1) &= - \log(t) \, ,\\
            D(\rho_2 \Vert \sigma_2) &= - \log(t) \, .
        \end{aligned}
    \end{equation}
    This gives us 
    \begin{equation}
        \begin{aligned}\label{eq:eq_tight_bound}
             p D(\rho_1 \Vert \sigma_1) &+ (1 - p) D(\rho_2 \Vert \sigma_2) - D(\rho \Vert \sigma) \\
            &= h(p) + p \log\Big(p + (1 - p) \frac{1 - t}{t}\Big) + (1 - p) \log\Big((1 - p) + p \frac{1 - t}{t}\Big) \, .
        \end{aligned}
    \end{equation}
    As $[\rho_{i}, \sigma_{j}] = 0$ for $i, j = 1,2$ and further $[\rho_i \sigma_j, \sigma_i] = 0$ we find that the constants in \cref{theo:theo_almost_concavity_relative_entropy} are given by
    \begin{equation}
        c_{i} = \tr[\rho_{i} \sigma_{i}^{-1} \sigma_{j}] = \frac{1 - t}{t} \, ,
    \end{equation}
    for $i, j = 1, 2$, $i \ne j$. Since $\rho_1$ and $\rho_2$ orthogonal we get $\frac{1}{2}\norm{\rho_1 - \rho_2}_1 = 1$. We hence obtain the RHS of \cref{eq:eq_tight_bound} from the almost concavity estimate in \cref{eq:almost_concavity_relative_entropy}. This concludes the claim.
\end{proof}

\subsection{Reduction to almost tight previously-known continuity bounds for the relative entropy}\label{subsec:reduction_continuity_bounds_relative_entropy}

In this section, we will show that a number of almost tight previously-known continuity bounds for quantities derived from the relative entropy can be obtained as corollaries of the results of almost concavity in \cref{theo:theo_almost_concavity_relative_entropy} and \cref{prop:prop_almost_concave_estimate_relative_entropy_well_behaved} in combination with the results concerning the ALAFF method, i.e. \cref{theo:theo_alaff_method} and \cref{rem:rem_s=0_alaff_method}.

\subsubsection{Uniform continuity for the conditional entropy}\label{subsec:subsec_uniform_continuity_conditional_entropy}

Let us first consider a bipartite space and the conditional entropy of a density matrix with respect to one of the subsystems. Note that, in this case, we are able to prove a result of uniform continuity for any positive semidefinite matrix (with trace one), but we do not require positive definiteness. The following coincides with the result of Winter \cite{Winter-AlickiFannes-2016}, which he proved to be almost tight.

\begin{corollary}[Uniform continuity of the conditional entropy]\label{cor:continuity_bound_conditional_entropy}~\\
    The conditional entropy over the bipartite Hilbert space $\cH= \cH_A \otimes \cH_B$ is uniformly continuous on $\cS_0 = \cS(\cH)$ and for $\rho, \sigma \in \cS_0$ with $\frac{1}{2}\norm{\rho - \sigma}_1 \le \varepsilon \le 1$, it holds that 
    \begin{equation}
        |H_\rho(A|B) - H_\sigma(A|B)| \le 2 \varepsilon \log d_A + (1 + \varepsilon)h\Big(\frac{\varepsilon}{1 + \varepsilon}\Big) \, .
    \end{equation}
\end{corollary}
\begin{proof}[Proof.]
    First of all, $\cS_0$ is clearly $0$-perturbed $\Delta$-invariant. Setting $f(\cdot) = H_\cdot(A|B)$, we find that it is ALAFF with $a_{H_\cdot(A|B)} = 0$ as $H_\cdot(A|B)$ is concave, and $b_{H_\cdot(A|B)} = h$ since the result in \cref{theo:theo_almost_concavity_relative_entropy} becomes independent of the states as we go to $H_\cdot(A|B)$ using point 3 of \cref{prop:prop_almost_concave_estimate_relative_entropy_well_behaved}. Finally, we find that 
    \begin{equation}
        C^\perp_f = \sup\limits_{\twoline{\rho, \sigma \in \cS_0}{\tr[\rho \sigma] = 0}} |H_\rho(A|B) - H_\sigma(A|B)| \le 2 \log d_A \, ,
    \end{equation}
    where we used $-\log d_X \le H_\cdot(X|Y) \le \log d_X$ shown, for example, in \cite{WildeFromClassicalToQuantumInformation_2016}. Using \cref{theo:theo_alaff_method} in the form of \cref{rem:rem_s=0_alaff_method}, we can infer the claimed continuity bound.
\end{proof}

\subsubsection{Uniform continuity for the mutual information}

For the mutual information, it is straightforward to derive a continuity bound for such a quantity just by combining the bounds of \cite{Audenaert-ContinuityEstimateEntropy-2007} and \cite{Winter-AlickiFannes-2016} for the von Neumann and conditional entropy, respectively:
   \begin{equation}\label{eq:continuity_bound_mutual_information}
        |I_\rho(A:B) - I_\sigma(A:B)| \le 3 \varepsilon \log\min\{ d_A,  d_B\} + 2(1 + \varepsilon) h\Big(\frac{\varepsilon}{1 + \varepsilon}\Big) \, ,
    \end{equation}
where $\varepsilon:= \frac{1}{2}\norm{\rho-\sigma}_1$.  For an early version, see \cite[Exercise 5.40]{hayashi2006quantum}. The multiplicative factor in the first term of the right-hand side was subsequently improved  to $2\sqrt{2}$ in \cite{Shirokov-AdaptationAFWold-2018} and to $2$ in \cite{Shirokov-AFWmethod-2020}. Moreover, we can adapt \cref{cor:continuity_bound_conditional_entropy} to obtain the following bound on the mutual information, which coincides with the tightest previously-known continuity bound for the mutual information (see e.g.\ \cite{Shirokov-AFWmethod-2020}).

\begin{corollary}[Continuity bound for the mutual information]\label{cor:continuity_bound_mutual_information}~\\
    The mutual information on a bipartite Hilbert space $\cH = \cH_A \otimes \cH_B$ is uniformly continuous on $\cS_0 = \cS(\cH)$ and for $\rho, \sigma \in \cS_0$ with $\frac{1}{2}\norm{\rho - \sigma}_1 \le \varepsilon \le 1$, we find that
    \begin{equation}
        |I_\rho(A:B) - I_\sigma(A:B)| \le 2 \varepsilon \log\min\{d_A,  d_B\} + 2(1 + \varepsilon) h\Big(\frac{\varepsilon}{1 + \varepsilon}\Big) \, .
    \end{equation}
\end{corollary}
\begin{proof}[Proof.]
    First of all, $\cS_0$ is clearly $0$-perturbed $\Delta$-invariant.
    With $f(\cdot) = I_\cdot(A:B) = S(\cdot_A) - H_\cdot(A|B)$ one can immediately conclude almost local affinity of $I_\cdot(A:B)$ as $S(\cdot_A)$ is concave and fulfills \cref{eq:eq_almost_convexity_von_neumann_entropy} and $-H_\cdot(A|B)$ is almost locally affine with $a_{-H_\cdot(A|B)} = 0$ and $b_{-H_\cdot(A|B)} = h$. Combined we get $a_f = h$ and $b_f = h$. We further have that 
    \begin{equation}
        C^\perp_f = \sup\limits_{\twoline{\rho, \sigma \in \cS_0}{\tr[\rho \sigma] = 0}} |I_\rho(A:B) - I_\sigma(A:B)| \le \sup\limits_{\rho \in \cS_0} I_\rho(A:B) \le 2 \log\min\{d_A, d_B\} \, ,
    \end{equation}
    where we used that $0 \le I_\cdot(A:B)$ and $I_\cdot(A:B) \le 2 \log\min\{d_A, d_B\}$ \cite{WildeFromClassicalToQuantumInformation_2016}. Applying \cref{theo:theo_alaff_method} in the form of \cref{rem:rem_s=0_alaff_method}, we can conclude the claim and obtain the given continuity bound.
\end{proof}

\subsubsection{Uniform continuity for the conditional mutual information}

We can also provide a continuity bound for the conditional mutual information of two tripartite states $\rho_{ABC},\sigma_{ABC} \in \cS(\cH_A \otimes \cH_B \otimes \cH_C)$ from \cref{cor:continuity_bound_conditional_entropy}, by viewing it as  the difference between two conditional entropies. The following result coincides with the best previously-known bound for the named quantity and appeared explicitly in \cite[Lemma 4]{Shirokov-ContinuityBounds-2019}, and with a worsening of a factor $2$ previously in  \cite{SutterRenner-ApproximateRecoverability-2018} and \cite[Exercise 5.41]{hayashi2006quantum}.

\begin{corollary}[Uniform continuity of the conditional mutual information]\label{cor:continuity_bound_CMI}~\\
    The conditional mutual information with respect to $\cH = \cH_A \otimes \cH_B \otimes \cH_C$ is uniformly continuous on $\cS_0 = \cS(\cH)$ and for $\rho, \sigma \in \cS_0$ with $\frac{1}{2}\norm{\rho - \sigma}_1 \le \varepsilon \le 1$, we find that 
    \begin{equation}
        |I_\rho(A:B | C) - I_\sigma(A:B|C)| \le 2 \varepsilon \log\min\{ d_A, d_B\} + 2(1 + \varepsilon) h \Big(\frac{\varepsilon}{1 + \varepsilon}\Big) \, .
    \end{equation}
\end{corollary}
\begin{proof}[Proof.]
    The procedure is now familiar. We first note that $\cS_0$ is 0-perturbed $\Delta$-invariant. Without loss of generality, we can assume that $d_A \le d_B$ and rewrite $f(\cdot) = I_\cdot(A:B|C) = H_\cdot(A|C) - H_\cdot(A|BC)$. With this representation, we can immediately conclude that $I_\cdot(A:B|C)$ is ALAFF with $a_f = h$ and $b_f = h$. Finally, we have that 
    \begin{equation}
        \begin{aligned}
            C^\perp_f &= \sup\limits_{\twoline{\rho, \sigma \in \cS_0}{\tr[\rho \sigma] = 0}} |I_\rho(A:B|C) - I_\sigma(A:B|C)|\\
            &\le \sup\limits_{\rho \in \cS_0} I_\rho(A:B|C)\\
            &= \sup\limits_{\rho \in \cS_0} H_\rho(A|BC) - H_\rho(A|C)\\
            &\le 2 \log d_A = 2\log\min\{d_A, d_B\} \, ,
        \end{aligned}
    \end{equation}
    as the conditional mutual information is non-negative and again $-\log d_X \le H_\cdot(X|Y) \le \log d_X$. Using \cref{theo:theo_alaff_method} in the form of \cref{rem:rem_s=0_alaff_method}, we can conclude the claim and obtain the given continuity bound.
\end{proof}

\subsection{New continuity bounds for the relative entropy}\label{subsec:new_continuity_bounds_relative_entropy}

Now, we prove some new continuity bounds for further quantities derived from the relative entropy as a consequence of the results of almost concavity in \cref{theo:theo_almost_concavity_relative_entropy} and \cref{prop:prop_almost_concave_estimate_relative_entropy_well_behaved} in combination with the results concerning the ALAFF method, i.e. \cref{theo:theo_alaff_method} and \cref{rem:rem_s=0_alaff_method}. All bounds in this section can be simplified using the following lemma:

\begin{lemma}\label{lemma:nice_bounds}
    Using the notations introduced in \cref{theo:theo_almost_concavity_relative_entropy} and \cref{rem:reduction_previous_continuity_bounds_relative_entropy}, we have the following estimates for the error bounds obtained in all results of this section:
    \begin{equation}\label{eq:bound_binary_entropy_CB}
        (1 + \varepsilon)h\Big(\frac{\varepsilon}{1 + \varepsilon}\Big) \leq \sqrt{2\varepsilon} \, ,
    \end{equation}
    \begin{equation}\label{eq:bound_fmtilde_entropy_CB}
        \frac{l_{\widetilde{m}} + \varepsilon}{l_{\widetilde{m}}} f_{\widetilde{m}^{-1} \, , \widetilde{m}^{-1}}\Big(\frac{\varepsilon}{l_{\widetilde{m}} + \varepsilon} \Big) \leq \frac{\varepsilon}{l_{\widetilde m}} \log\widetilde m^{-1} + \log\left(1 + \frac{\varepsilon}{l_{\widetilde m} + \varepsilon} \frac{1}{\widetilde m}\right) \, .
    \end{equation}
\end{lemma}

\begin{proof}
    The first inequality appeared in \cite{Sutter-ApproximateQMC-2018} before and its proof follows from some elementary calculus. For the second inequality note that $\varepsilon \in [0, 1]$ and $\widetilde m \in (0, 1)$, it holds that $l_{\widetilde m} = 1 - \widetilde m \in (0, 1)$, allowing us to estimate $1 \le \frac{1}{\widetilde m}$ and $\frac{l_{\widetilde m}}{l_{\widetilde m} + \varepsilon} \le 1$. This results in:
    \begin{equation}
        \begin{aligned}
            \frac{l_{\widetilde{m}} + \varepsilon}{l_{\widetilde{m}}} f_{\widetilde{m}^{-1} \, , \widetilde{m}^{-1}} \Big(\frac{\varepsilon}{l_{\widetilde{m}} + \varepsilon} \Big) &= \frac{\varepsilon}{l_{\widetilde m}}\log\left(\frac{\varepsilon}{l_{\widetilde m} + \varepsilon} + \frac{l_{\widetilde m}}{l_{\widetilde m} + \varepsilon} \frac{1}{\widetilde m}\right) + \log\left(\frac{\varepsilon}{l_{\widetilde m} + \varepsilon} \frac{1}{\widetilde m} + \frac{l_{\widetilde m}}{l_{\widetilde m} + \varepsilon}\right)\\
            &\le \frac{\varepsilon}{l_{\widetilde m}} \log\widetilde m^{-1} + \log\left(1 + \frac{\varepsilon}{l_{\widetilde m} + \varepsilon} \frac{1}{\widetilde m}\right) \, .
        \end{aligned}
    \end{equation}
\end{proof}

\subsubsection{Divergence bounds for the relative entropy}

In this section, we prove an upper bound on the relative entropy $D(\rho \Vert \sigma)$ which involves the 
trace norm distance of $\rho$ and $\sigma$. The literature calls these bounds upper continuity bounds \cite{AudenaertDatta-Renyi-entropies-2015, Rastegin_2011, Vershynina_2019}, for which we would expect an upper bound of $|D(\rho_1\Vert \sigma_1) - D(\rho_2 \Vert \sigma_2)|$ in terms of the norm distance of $\rho_1$ and $\rho_2$, and $\sigma_1$ and $\sigma_2$, respectively. We hence propose the name ``divergence bound" for this kind of bound, a fitting name, since we are relating the strength of divergence (between $\rho$ and $\sigma$) to a fixed distance measure (the trace norm).\par
We now give the divergence bound we obtain when using the convexity and almost concavity of $D(\rho \Vert \sigma)$ together with \cref{theo:theo_alaff_method} by going through uniform continuity of the relative entropy in its first argument.

\begin{corollary}[Uniform continuity of the relative entropy in the first argument]\label{cor:cor_uniform_continuity_relative_entropy_first_argument}~\\
    Let $\sigma \in \cS(\cH)$ be fixed. Then $D(\cdot \Vert \sigma)$ is uniformly continuous on $\cS_0 = \{\rho \in \cS(\cH) \;:\; \ker \sigma \subseteq \ker \rho\}$ and, for $\rho_1, \rho_2 \in \cS_0$ with $\frac{1}{2}\norm{\rho_1 - \rho_2}_1 \le \varepsilon \le 1$, it holds that
    \begin{equation}\label{eq:uniform_continuity_rel_ent_first_argument}
        |D(\rho_1 \Vert \sigma) - D(\rho_2 \Vert \sigma)| \le \varepsilon \log \widetilde{m}_\sigma^{-1} + (1 + \varepsilon) h \Big(\frac{\varepsilon}{1 + \varepsilon}\Big) \, ,
    \end{equation}
    with $\widetilde{m}_\sigma$ the minimal non-zero eigenvalue of $\sigma$.
\end{corollary}
\begin{proof}
    $\cS_0$ is clearly convex and $0$-perturbed $\Delta$-invariant as for two operators $A, B$, $\ker A \cap \ker B  \subseteq \ker (A - B)$ and $[A - B]_{\pm}$ are orthogonal. We set $f(\cdot) = D(\cdot \Vert \sigma)$. Using \cref{theo:theo_almost_concavity_relative_entropy} and point 1 of \cref{prop:prop_almost_concave_estimate_relative_entropy_well_behaved}, we find that $D(\cdot \Vert \sigma)$ is ALAFF with $a_f = h$ and $b_f = 0$. At last, we have that 
    \begin{equation}
        C^\perp_f = \sup\limits_{\twoline{\rho_1, \rho_2 \in \cS_0}{\frac{1}{2}\|\rho_1- \rho_2\| = 1}} |D(\rho_1 \Vert \sigma) - D(\rho_2 \Vert \sigma)| \le \sup\limits_{\rho \in \cS(\cH)} D(\rho\Vert \sigma) \le \log \widetilde{m}_\sigma^{-1}.
    \end{equation}
    In the first inequality, we used that $D(\rho\Vert \sigma) \ge 0$, and in the second one that $\widetilde{m}_\sigma \rho \le \sigma$ hence $D(\rho \Vert \sigma) \le \log \widetilde{m}_\sigma^{-1}$. Using \cref{theo:theo_alaff_method} in the form of \cref{rem:rem_s=0_alaff_method} concludes the claim.
\end{proof}

We can compare \cref{eq:uniform_continuity_rel_ent_first_argument} with the findings of \cite[Eq. (43) and (44)]{GourTomamichel-ResourceMeasures-2020}, based on the previous \cite{GourTomamichel-RelativeEntropy-2021}, where it was shown that
\begin{equation}\label{eq:uniform_continuity_rel_ent_first_input_tomamichel_gour}
    |D(\rho_1 \Vert \sigma) - D(\rho_2 \Vert \sigma)| \le \max\limits_{i = 1, 2}\log \left( 1 + \frac{\norm{\rho_1 - \rho_2}_\infty}{m_{\rho_i}m_{\sigma}} \right) \, ,
\end{equation}
whenever $\rho_i > 0$ and $\min\limits_{i = 1, 2} m_{\rho_i} > \norm{\rho_1 - \rho_2}_\infty$. Here $m_{\rho_i}$ is the minimal eigenvalue of $\rho_i$ for $i=1, 2$ and correspondingly $m_\sigma$ the one of $\sigma$. This expression presents the advantage with respect to ours of depending on the operator norm of the difference of $\rho_1$ and $\rho_2$, instead of the trace norm. However, when $\rho_1 \approx \rho_2$, the upper bound in \cref{eq:uniform_continuity_rel_ent_first_input_tomamichel_gour} can be approximated by $\frac{\norm{\rho_1 - \rho_2}_\infty}{m_{\rho_i}m_{\sigma}}$, and thus the dependence with $m_{\sigma}^{-1}$ is linear, instead of logarithmic as in \cref{eq:uniform_continuity_rel_ent_first_argument}. Further in \cref{eq:uniform_continuity_rel_ent_first_input_tomamichel_gour} one needs $\rho_1$ and $\rho_2$ to be full rank and has a condition on their minimal eigenvalues.

We can subsequently use the \cref{cor:cor_uniform_continuity_relative_entropy_first_argument} to prove a divergence bound for the relative entropy.

\begin{corollary}[Divergence bound for the relative entropy]\label{cor:divergence_bound_relative_entropy}~\\
    Let $\rho, \sigma \in \cS(\cH)$ with $\ker \sigma \subseteq \ker \rho$ and $\frac{1}{2}\norm{\rho - \sigma}_1 \le \varepsilon \le 1$, we have 
    \begin{equation}
        D(\rho \Vert \sigma) \le \varepsilon \log \widetilde{m}_\sigma^{-1} + (1 + \varepsilon)h\Big(\frac{\varepsilon}{1 + \varepsilon}\Big)  \leq \Big(1 + \frac{\log \widetilde{m}_\sigma^{-1}}{\sqrt{2}} \Big) \varepsilon^{1/2} \, .
    \end{equation}
    with $\widetilde{m}_\sigma$ the minimal non-zero eigenvalue of $\sigma$. The second inequality follows from \eqref{eq:bound_binary_entropy_CB} and the fact that $\varepsilon \leq \sqrt{\varepsilon}$ for any $\varepsilon \in [0,1]$.
\end{corollary}
\begin{proof}[Proof.]
    In the context of \cref{cor:cor_uniform_continuity_relative_entropy_first_argument}, we just set $\rho_1 = \rho$ and $\rho_2 = \sigma$, giving us that $\frac{1}{2}\norm{\rho_1 - \rho_2}_1 = \frac{1}{2}\norm{\rho - \sigma}_1 \le \varepsilon \le 1$. Furthermore, $D(\rho_2\Vert \sigma) = D(\sigma \Vert \sigma) = 0$ and $|D(\rho_1\Vert \sigma)|$ loses the absolute value, as $D(\cdot \Vert \cdot) \ge 0$. The bound follows immediately.
\end{proof}

\begin{figure}[ht!]
    \centering
    \begin{subfigure}[t]{0.49\textwidth}
        \centering
        \includegraphics[width=\textwidth]{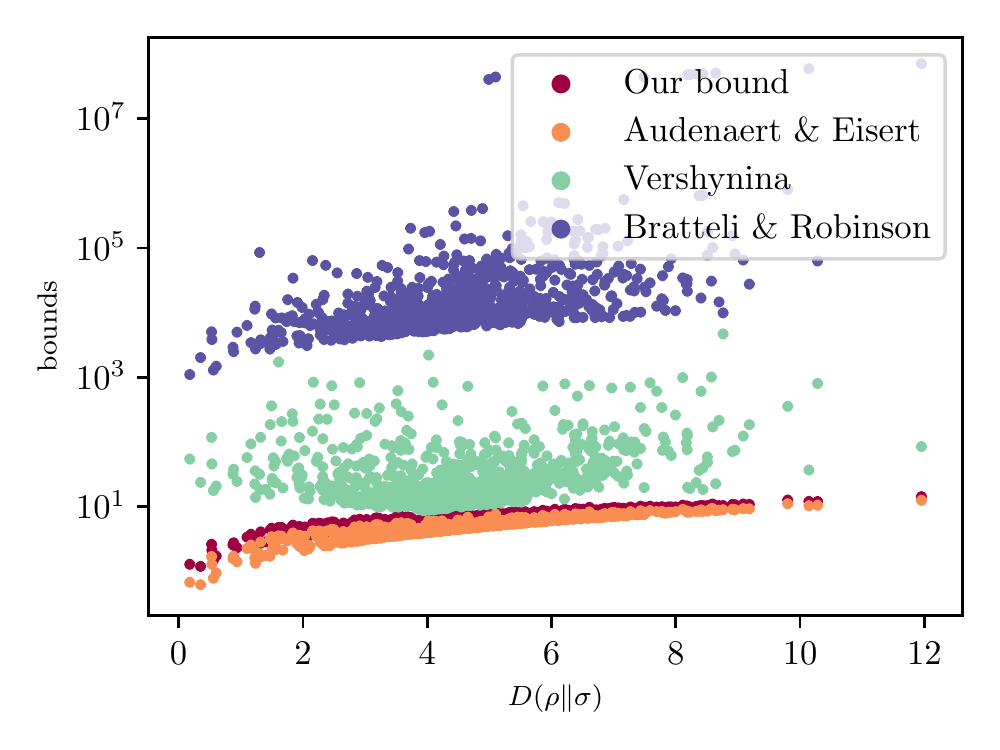}
        \caption{The magnitude of the different bounds plotted over the relative entropy. We sampled thousand different pairs of qubits and controlled the minimal eigenvalue of $\sigma$ in a range from $10^{-4}$ to $10^{-8}$. The explicit bounds can be found in \cref{tab:divergence_bounds_comparison}.}
        \label{subfig:subfig_pointcloud}
    \end{subfigure}
    \begin{subfigure}[t]{0.49\textwidth}
        \centering
        \includegraphics[width=\textwidth]{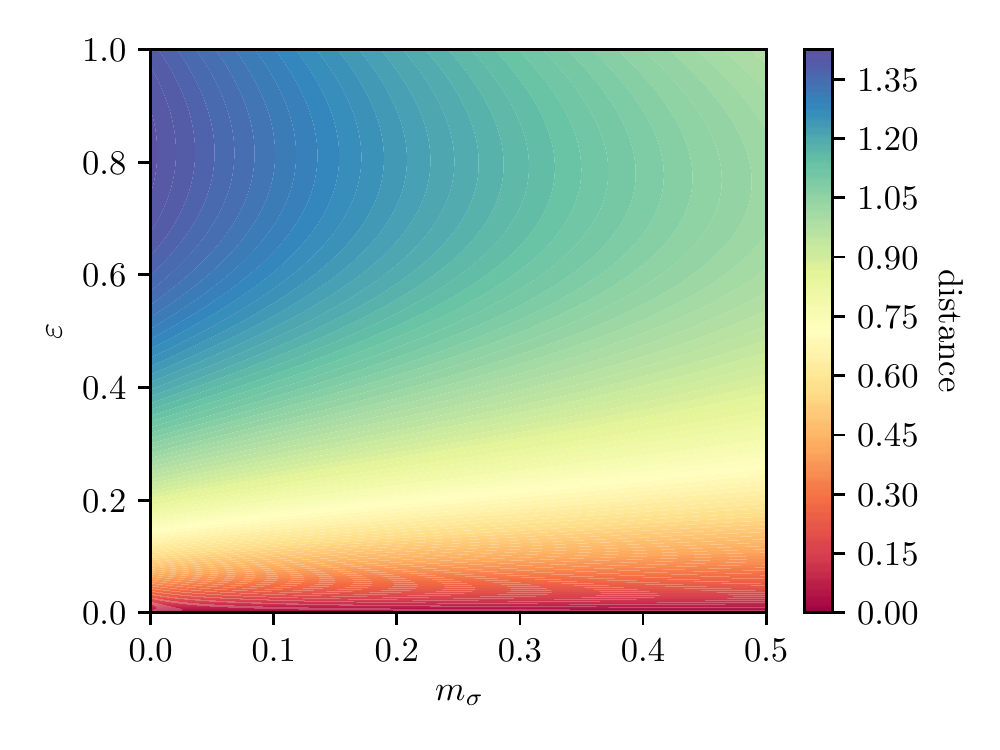}
        \caption{The difference between the bound from \cref{cor:divergence_bound_relative_entropy} and the one of Audenaert \& Eisert \cite[Theorem 1]{AudenaertEisert_II_2011}. On the x-axis we plot the minimal eigenvalue of $\sigma$ and on the y-axis $\varepsilon = \frac{1}{2} \norm{\rho - \sigma}_1$. The minimal eigenvalue of $\rho$ is set to the minimal eigenvalue of $\sigma$, thereby strengthening the bound of Audenaert \& Eisert. Both were varied between $10^{-20}$ and $\frac{1}{2}$.}
        \label{subfig:subfig_heatmap}
    \end{subfigure}
    \caption{Two plots comparing the divergence bounds from \cref{tab:divergence_bounds_comparison}.}
    \label{fig:comparison_of_divergence_bounds}
\end{figure}

\begin{table}[ht!]
    \centering
    \scalebox{0.97}{
    \begin{tabular}{c|ccc}
        Bound by & not  & not & Bound on $D(\rho \Vert \sigma)$\\
        & full rank $\rho$ &  full rank $\sigma$ & \\
        \hline\hline
        
        \cref{cor:divergence_bound_relative_entropy} & $\checkmark$ & $\checkmark$ & 
         \footnotesize $\varepsilon \log \widetilde{m}_\sigma^{-1} + (1 + \varepsilon) h\Big(\frac{\varepsilon}{1 + \varepsilon}\Big)$ \\
        
        \begin{tabular}{c}
            Audenaert \& Eisert\\
            \cite[Theorem 1]{AudenaertEisert_II_2011}
         \end{tabular} & $\checkmark$ & x & 
         \footnotesize $(m_\sigma + \varepsilon) \log\Big(\frac{m_\sigma + \varepsilon}{m_\sigma}\Big)- m_\rho \log\Big(\frac{m_\rho + \varepsilon}{m_\rho}\Big)$ \\
         
         \begin{tabular}{c}
            Vershynina\\
            \cite{Vershynina_2019}
         \end{tabular} & x & x & 
          \footnotesize $2 \varepsilon \lambda_\rho \frac{\log m_\rho - \log m_\sigma}{m_\rho - m_\sigma}$ \\

         \begin{tabular}{c}
            Bratteli \&\\
            Robinson\\
            \cite{BraRob81}
         \end{tabular} & x & x & 
         \footnotesize $m_\sigma^{-1} \norm{\rho - \sigma}_\infty$
         
    \end{tabular}
    }
    \caption{A comparison of different divergence bounds. Here $\varepsilon = \frac{1}{2} \norm{\rho - \sigma}_1$ and $m_\cdot$ and $\widetilde{m}_\cdot$ are the minimal and the minimal non-zero eigenvalue of the quantum state in the index, respectively. Further $\lambda_\rho$ is the maximal eigenvalue of $\rho$. The bound of Audenaert $\&$ Eisert in the case $m_\rho = 0$ has to be understood as the limit $m_\rho \to +0$.}
    \label{tab:divergence_bounds_comparison}
\end{table}

There exist results on divergence bounds in the literature which predate ours. In \cite{AudenaertEisert_I_2005, Vershynina_2019}, the authors present some linear bounds for the relative entropy in terms of the trace norm difference between those states, with some multiplicative factors depending on the eigenvalues of the states involved, whereas in \cite{BraRob81} a similar bound is provided in terms of the operator norm of the difference between the states. One of the bounds in \cite{AudenaertEisert_I_2005} is further generalised in \cite{AudenaertEisert_II_2011} and is closely related to our bound as both of them are non-linear in the trace norm (resp. operator norm) difference between the involved states, and show a dependence on the inverse of the minimal eigenvalue of $\sigma$ only logarithmically. This is partly an advantage over the bounds in \cite{BraRob81, Vershynina_2019}. There further exists a bound in \cite[Proposition 5.81]{Hanson-ThesisEntropyBounds-2020} which has an explicit dependence on the dimension in addition to the dependence on the minimal eigenvalue of $\sigma$ and therefore was not investigated. The bound might have an advantage in low-dimensional settings. In \cref{tab:divergence_bounds_comparison} and \cref{fig:comparison_of_divergence_bounds} we compare the aforementioned bounds from \cite{AudenaertEisert_II_2011, BraRob81, Vershynina_2019}. From \cref{subfig:subfig_pointcloud} it is clear that our bound, in the majority of the cases, outperforms the bound by Vershynina and the one by Bratteli \& Robinson. This is because of the logarithmic scaling with the inverse minimal eigenvalue of $\sigma$ of our bound versus the linear scaling with the inverse minimal eigenvalue of $\sigma$ of theirs. We hence reduce the discussion to a comparison between Audenaert \& Eisert's and our bound. From the first \cref{subfig:subfig_pointcloud} and second plot \cref{subfig:subfig_heatmap} we conclude a slight advantage of theirs. The numerical experiments suggest, however, that the difference between both bounds is bounded by two, hence as the minimal eigenvalue decreases both bounds should converge asymptotically. Furthermore, our bound has the advantage that it does not need $\sigma$ nor $\rho$ to be full rank. This fact and its simple representation might give some advantages in applications.

\subsubsection{Continuity bounds for the relative entropy}\label{subsec:subsec_continuity_bounds_relative_entropy}

We conclude our section on continuity bounds with the most involved continuity bound until now. It concerns the relative entropy and regards it in all its power as a function of two variables, i.e., it constitutes a continuity bound both for the first and the second input simultaneously. This presents some challenges that need to be dealt with, as the relative entropy exhibits discontinuity whenever the kernel of the second input is not contained in that of the first one. To overcome these issues, we need to employ the ALAFF method in its full generality.

In the first step, we fix the first input of the relative entropy and provide a continuity bound for the relative entropy in the second argument. 

\begin{corollary}[Uniform continuity of the relative entropy in the second argument]\label{cor:cor_uniform_continuity_relative_entropy_second_argument}~\\
    Let $\rho \in \cS(\cH)$ be fixed and $1 > \widetilde{m} > 0$. Then, $D(\rho \Vert \cdot)$ is uniformly continuous on 
    \begin{equation}
        \cS_0 := \{\sigma \in \cS(\cH) \;:\;  \widetilde m \rho \le \sigma \} \, .
    \end{equation}
    We further get that, for $\sigma_1, \sigma_2 \in \cS_0$ with $\frac{1}{2}\norm{\sigma_1 - \sigma_2}_1 \le \varepsilon$,
    \begin{equation}\label{eq:continuity_bound_relative_entropy_second_argument}
        \begin{aligned}
            |D(\rho\Vert \sigma_1) - D(\rho \Vert \sigma_2)| &\le \frac{\varepsilon}{l_{\widetilde{m}}} \log(\widetilde{m}^{-1}) + \frac{l_{\widetilde{m}} + \varepsilon}{l_{\widetilde{m}}} f_{\widetilde{m}^{-1} \, , \widetilde{m}^{-1}}\Big(\frac{\varepsilon}{l_{\widetilde{m}} + \varepsilon} \Big) \\
            &\le 2 \frac{\varepsilon}{l_{\widetilde m}} \log \widetilde m^{-1} + \log\left(1 + \frac{\varepsilon}{l_{\widetilde m}+ \varepsilon} \frac{1}{\widetilde m}\right)\, ,
        \end{aligned}
    \end{equation}
    where $l_{\widetilde{m}} = 1 - \widetilde{m}$. The second inequality follows from \eqref{eq:bound_fmtilde_entropy_CB}.
\end{corollary}
\begin{proof}
    We have that $\cS_0$ is clearly convex as, for $\sigma_1, \sigma_2 \in \cS_0$ and $\lambda \in [0, 1]$, 
    \begin{equation}
        \lambda \sigma_1 + (1 - \lambda) \sigma_2 \ge \lambda\widetilde{m}\rho + (1 - \lambda) \widetilde{m} \rho  = \widetilde{m} \rho  \, ,
    \end{equation}
    giving the kernel inclusion as well as the condition for the smallest eigenvalue on the support of $\rho$. Furthermore, $\cS_0$ is $s$-perturbed $\Delta$-invariant with $s = \widetilde{m}$. This is because one can perturb with $\tau = \rho$ and get subminorization by $\widetilde{m} \rho$. Employing point 2 of \cref{prop:prop_almost_concave_estimate_relative_entropy_well_behaved} we further find that $f(\cdot) = D(\rho\Vert \cdot)$ satisfies \cref{eq:eq_alaff_property} with $b_f = 0$ and $a_f = f_{\widetilde{m}^{-1}, \widetilde{m}^{-1}}$, hence $E_f =  f_{\widetilde{m}^{-1}, \widetilde{m}^{-1}}$. Using again \cref{prop:prop_almost_concave_estimate_relative_entropy_well_behaved} (point 4, since $\tilde m \leq 1$) we conclude $E_f^{\max} = f_{\widetilde{m}^{-1}, \widetilde{m}^{-1}}$. At last, we have that 
    \begin{equation}
        \begin{aligned}
            C^{\widetilde{m}}_f &= \sup\limits_{\twoline{\sigma_1, \sigma_2 \in \cS_0}{\frac{1}{2}\norm{\sigma_1 - \sigma_2}_1 = 1 - \widetilde{m}}}|D(\rho \Vert \sigma_1) - D(\rho \Vert \sigma_2)|\\
            &\le \sup\limits_{\sigma \in \cS_0} D(\rho \Vert \sigma)\\
            &\le \log(\widetilde{m}^{-1}) \, ,
        \end{aligned}
    \end{equation}
    where we used that $D(\rho\Vert \cdot) \ge 0$ and for the last inequality that $\widetilde{m}\rho \le \sigma$ for all $\sigma \in \cS_0$. Employing now \cref{theo:theo_alaff_method} we obtain uniform continuity and the claimed continuity bound.
\end{proof}
As in the case of the continuity bound for the relative entropy in the first input, we can compare \cref{eq:continuity_bound_relative_entropy_second_argument} with \cite[Eq. (39) and (40)]{GourTomamichel-ResourceMeasures-2020}, as an extension of the previous \cite{GourTomamichel-RelativeEntropy-2021}, in which it was shown that
\begin{equation}\label{eq:uniform_continuity_rel_ent_second_input_tomamichel_gour}
    |D(\rho \Vert \sigma_1) - D(\rho \Vert \sigma_2)| \le \max\limits_{i = 1, 2} - \log \left( 1 - \frac{\norm{\sigma_1 - \sigma_2}_\infty}{m_{\sigma_i}} \right) \, ,
\end{equation}
whenever $\min\limits_{i = 1, 2}m_{\sigma_i}> \norm{\sigma_1 - \sigma_2}_\infty$, where $m_{\sigma_i}$ is the minimal eigenvalue of $\sigma_i$ for $i = 1, 2$. In the low $\varepsilon$ regime the bound in \cref{eq:uniform_continuity_rel_ent_second_input_tomamichel_gour} as well as the bound in \cref{eq:continuity_bound_relative_entropy_second_argument} scale linearly in $m_{\sigma}^{-1}$.

In the above corollary, two choices need some more justification. The first choice is $1 > \widetilde{m}$ and the second one is $s = \widetilde{m}$. We want to put them into context by the following lemma, demonstrating that these assumptions are necessary to obtain a non-trivial $\mathcal S_0$.

\begin{lemma}\label{lem:delta_invariant_subset}~\\
    Let $\rho \in \cS(\cH)$ and $s \in [0, 1)$ with $\rank \rho \ge 2$, further $\widetilde{m} \in (0, \infty)$ and
    \begin{equation}
        \cS_0 := \{\sigma \in \cS(\cH) \;:\; \ker \sigma \subseteq \ker \rho, \,\widetilde{m} \rho \le \sigma \} \, .
    \end{equation}
    Then, the following is true:
    \begin{enumerate}
        \item If $1 > \widetilde{m}$, then $\cS_0$ is $s$-perturbed $\Delta$-invariant if and only if $s \ge \widetilde{m}$.
        \item If $1 = \widetilde{m}$, then $\cS_0 = \{\rho\}$.
        \item If $1 < \widetilde{m}$, $\cS_0 = \emptyset$.
    \end{enumerate}
\end{lemma}

We will only give proof for the first one in \cref{sec:sec_proof_delta_invariant_subset} and leave the last two for the reader. Next, we proceed to state and prove the main result of this subsection on continuity bounds, namely the uniform continuity bound for the relative entropy in both arguments on a suitable subspace. Since we have already explored the cases in which we either fix the second (\cref{cor:cor_uniform_continuity_relative_entropy_first_argument}) or first (\cref{cor:cor_uniform_continuity_relative_entropy_second_argument}) density operator, we now combine both results in the proof of the next theorem.

\begin{theorem}[Uniform continuity of the relative entropy]\label{theo:theo_uniform_continuity_relative_entropy}~\\
    Let $1 > 2\widetilde{m} > 0$ and 
    \begin{equation}
        \cS_0 = \{(\rho, \sigma) \;:\; \rho, \sigma \in \cS(\cH), \, \ker \sigma \subseteq \ker \rho, \, 2\widetilde{m}  \le \widetilde{m}_\sigma \} \, ,
    \end{equation}
    with $\widetilde{m}_\sigma$ the minimal non-zero eigenvalue of $\sigma$. Then, $D(\cdot \Vert \cdot)$ is uniformly continuous on $\cS_0$ and we find that for $(\rho_1, \sigma_1), (\rho_2, \sigma_2) \in \cS_0$ with $\frac{1}{2}\norm{\rho_1 - \rho_2} \le \varepsilon\le 1$ and $\frac{1}{2}\norm{\sigma_1 - \sigma_2}_1 \le \delta \le 1$
    \begin{equation}
        \begin{aligned}
              |D(\rho_1 \Vert \sigma_1) - D(\rho_2 \Vert \sigma_2)| & \le \Big(\varepsilon + \frac{\delta}{l_{\widetilde{m}}}\Big) \log(\widetilde{m}^{-1}) + (1 + \varepsilon)h\Big(\frac{\varepsilon}{1 + \varepsilon}\Big) + 2\frac{l_{\widetilde{m}} + \delta}{l_{\widetilde{m}}} f_{\widetilde{m}^{-1}, \widetilde{m}^{-1}}\Big(\frac{\delta}{l_{\widetilde{m}} + \delta} \Big)\\
              & \leq \Big(\sqrt{2} + \log \widetilde{m}^{-1}\Big) \varepsilon^{1/2} + 3 \frac{\delta}{l_{\widetilde m}} \log \widetilde m^{-1} + 2 \log\left(1 + \frac{\delta}{l_{\widetilde m} + \delta} \frac{1}{\widetilde m}\right)  \,  , \label{eq:eq_continuity_bound_relative_entropy}  
        \end{aligned}
    \end{equation}
    with $l_{\widetilde{m}} = 1 - \widetilde{m}$. The second inequality follows from \eqref{eq:bound_fmtilde_entropy_CB} and $\varepsilon\leq \sqrt{\varepsilon}$ for $\varepsilon \in [0,1]$.
\end{theorem}
\begin{proof}
    We will prove the uniform continuity by proving that the bound \cref{eq:eq_continuity_bound_relative_entropy} holds. Therefore, let $(\rho_1, \sigma_1), (\rho_2, \sigma_2) \in \cS_0$ with $\frac{1}{2}\norm{\rho_1 - \rho_2} \le \varepsilon \le 1$ and $\frac{1}{2}\norm{\sigma_1 - \sigma_2}\le \delta \le 1$. We define
    \begin{equation}\label{eq:interpolation_sigmas_RE}
        \overline{\sigma} = \frac{1}{2}\sigma_1 + \frac{1}{2}\sigma_2 \, ,
    \end{equation}
    and obtain
    \begin{equation}
        \begin{aligned}
            \frac{1}{2}\norm{\overline{\sigma} - \sigma_1}_1 &= \frac{1}{4}\norm{\sigma_1 - \sigma_2}_1 \le \frac{\delta}{2} \le 1 \, ,\\
            \frac{1}{2}\norm{\overline{\sigma} - \sigma_2}_1 &= \frac{1}{4}\norm{\sigma_1 - \sigma_2}_1 \le  \frac{\delta}{2} \le 1 \, .
        \end{aligned}
    \end{equation}
    Using this, we get 
    \begin{equation}
        \begin{aligned}
            |D(\rho_1 \Vert \sigma_1) - D(\rho_2 \Vert \sigma_2)| &\le |D(\rho_1\Vert \sigma_1) - D(\rho_1 \Vert \overline{\sigma})| + |D(\rho_1 \Vert \overline{\sigma}) - D(\rho_2 \Vert \overline{\sigma})| + |D(\rho_2 \Vert \overline{\sigma}) - D(\rho_2 \Vert \sigma_2)| \, .
        \end{aligned}
    \end{equation}
    The middle term can be bounded using \cref{cor:cor_uniform_continuity_relative_entropy_first_argument} and the fact that
    \begin{equation}
        \log\widetilde{m}_{\overline{\sigma}}^{-1} \le \log (2 \max\{\widetilde{m}_{\sigma_1}^{-1}, \widetilde{m}_{\sigma_2}^{-1}\}) \le \log \widetilde{m}^{-1}.
    \end{equation}
    One obtains
    \begin{equation}
        \begin{aligned}
            |D(\rho_1 \Vert \overline{\sigma}) - D(\rho_2 \Vert \overline{\sigma})| \le \varepsilon \log\widetilde{m}^{-1} + (1 + \varepsilon) h\Big(\frac{\varepsilon}{1 + \varepsilon}\Big).
        \end{aligned}
    \end{equation}
    The other two terms are bound using \cref{cor:cor_uniform_continuity_relative_entropy_second_argument} and the fact that $\widetilde{m} \rho_1 \le \frac{1}{2} \sigma_1 \le \overline{\sigma}$ and $\widetilde{m} \rho_2 \le \frac{1}{2} \sigma_2 \le \overline{\sigma}$ by construction of $\cS_0$ and the definition of $\overline{\sigma}$, respectively. We therefore obtain
    \begin{equation}
        \begin{aligned}
             |D(\rho_1\Vert \sigma_1) - D(\rho_1 \Vert \overline{\sigma})| &\le \frac{\delta}{2l_{\widetilde{m}}} \log(\widetilde{m}^{-1}) + \frac{l_{\widetilde{m}} + 2^{-1} \delta}{l_{\widetilde{m}}}f_{\widetilde{m}^{-1}, \widetilde{m}^{-1}}\Big(\frac{2^{-1}\delta}{l_{\widetilde{m}} + 2^{-1} \delta}\Big) \, ,\\
             |D(\rho_2 \Vert \overline{\sigma}) - D(\rho_2 \Vert \sigma_2)| &\le \frac{\delta}{2l_{\widetilde{m}}} \log(\widetilde{m}^{-1}) + \frac{l_{\widetilde{m}} + 2^{-1} \delta}{l_{\widetilde{m}}}f_{\widetilde{m}^{-1}, \widetilde{m}^{-1}}\Big(\frac{2^{-1}\delta}{l_{\widetilde{m}} + 2^{-1} \delta}\Big) \, .
        \end{aligned}
    \end{equation}
    Combining the bounds and further using that 
    \begin{equation}
        \frac{l_{\widetilde{m}} + 2^{-1} \delta}{l_{\widetilde{m}}}f_{\widetilde{m}^{-1}, \widetilde{m}^{-1}}\Big(\frac{2^{-1}\delta}{l_{\widetilde{m}} + 2^{-1}\delta}\Big) \le  \frac{l_{\widetilde{m}} + \delta}{l_{\widetilde{m}}}f_{\widetilde{m}^{-1}, \widetilde{m}^{-1}}\Big(\frac{\delta}{l_{\widetilde{m}} + \delta}\Big) \, ,
    \end{equation}
    we obtain the claimed bound, and thereby also uniform continuity.
\end{proof}

Let us conclude this section by emphasizing that there might be some room for improvement in the previous result. For instance, it should be possible to improve the interpolation between $\sigma_1$ and $\sigma_2$ considered in \cref{eq:interpolation_sigmas_RE} by optimizing over the interpolation parameter instead of setting it to $1/2$. However, we believe this would not change the appearance of the bound drastically and thus the reason for not performing this optimization.

%% file: sections/almost-concav-bs-ent.tex
Following the same lines as in the previous section, now we apply the ALAFF method introduced in  \cref{sec:ALAFF_method} for the particular case of the BS-entropy. For that, we need to prove a result of almost concavity for the BS-entropy, which we do in \cref{subsec:almost_concavity_BS_entropy}. However, in contrast to the case of the relative entropy, our result for the BS-entropy is not tight. We leave the discussion on the almost concavity bound and the difficulties that appear in the BS-entropy case to the next subsection.

Subsequently, we combine our result of almost concavity for the BS-entropy with the ALAFF method to provide certain results of uniform continuity and explicit continuity bounds for entropic quantities derived from the BS-entropy in \cref{subsec:subsec_continuity_bounds_bs_entropy}. All the results provided in this section are summarized in \cref{fig:fig_flow_chart_BS}.

\begin{figure}[ht!]
    \centering
    \scalebox{0.8}{
    \begin{tikzpicture}[font=\small,thick,scale=1, every node/.style={scale=1}]
        \node[draw ,
            diamond,
            fill=midblue!30!white,
            minimum width=2.5cm,
            minimum height=1cm,
            inner sep=0.1cm,
            align=center] (block1) {BS\\entropy\\[1mm]$\widehat D(\rho\Vert\sigma)$};
            
        \node[draw,
            rounded rectangle,
            fill=cyan!30!white,
            below left=of block1,
            minimum height=1.38cm,
            minimum width=2.5cm,
            inner sep=0.1cm,
            align=center] (block2) {Convexity\\{\tiny$p \widehat D(\rho_1\Vert \sigma_1) + (1 - p) \widehat D(\rho_2\Vert \sigma_2) \ge \widehat D(\rho\Vert \sigma)$}};
            
        \node[draw,
            rounded rectangle,
            fill=cyan!30!white,
            below right=of block1,
            minimum width=2.5cm,
            inner sep=0.1cm,
            align=center] (block3) {Almost concavity\\{\tiny$\widehat D(\rho\Vert\sigma) \ge p \widehat D(\rho_1\Vert \sigma_1) + (1 - p) \widehat D(\rho_2\Vert \sigma_2) - f(p)$}\\{\tiny with $f(p)= \hat{c}_0(1 - \delta_{\rho_1 \rho_2})h(p) + f_{\hat c_1, \hat c_2}(p)$}};
         
        \node[draw,
            trapezium, 
            fill=lightgreen!50!white,
            below=2cm of block1,
            trapezium left angle = 65,
            trapezium right angle = 115,
            trapezium stretches,
            minimum width=3.5cm,
            minimum height=1cm,
            inner sep=0.1cm,
            align=center] (block4) {ALAFF method\\[1mm]{\footnotesize \cref{theo:theo_alaff_method} and \cref{rem:rem_s=0_alaff_method}}};
        
        \node[coordinate,below=0.7cm of block4] (block5) {};
        
        \node[coordinate,below=2.6cm of block5] (block10) {};
        
        \node[draw,
            left=0.7cm of block10,
            fill=midgreen!50!white,
            minimum width=2.5cm,
            minimum height=1cm,
            align=center] (block6) {BS-conditional\\mutual information\\[1.2mm]{\tiny$ \mathbb  |\widehat I_\rho(A:B | C) - \widehat I_\sigma(A:B|C)|  $}\\{\tiny$ \le  2 \, \varepsilon\, l_m^{-1} \log\min\{d_A, \sqrt{d_{ABC}}\} $}\\{\tiny$+ 2\frac{l_m + \varepsilon}{l_m}(f_{m^{-1}, m^{-1}} + m^{-1}h)\Big(\frac{\varepsilon}{l_m + \varepsilon}\Big)$}\\[1.1mm]{\tiny with $l_m = 1 - d_{\cH} m$, \, $\varepsilon \le \frac{1}{2}\| \rho - \sigma \|_1 $}};
        
        \node[draw,
            right=0.7cm of block6,
            fill=midgreen!50!white,
            minimum width=2.5cm,
            minimum height=1cm,
            align=center] (block7) {BS-mutual information\\[1.2mm]{\tiny$ | \widehat I_\rho(A:B) -  \widehat I_\sigma(A:B) |  $}\\{\tiny$ \leq  2 l_m^{-1} \varepsilon( \log\min\{d_A, d_B\} + \log m^{-1}) $}\\{\tiny $+ 2 \frac{l_m + \varepsilon}{l_m} (f_{m^{-1}, m^{-1}} + (m^{-1} + 1) h )\Big(\frac{\varepsilon}{l_m + \varepsilon}\Big) $}\\[1.1mm]{\tiny with $l_m = 1 - d_{\cH} m$, \, $\varepsilon \le \frac{1}{2}\| \rho - \sigma \|_1 $}};
         
        \node[draw,
            right=0.7cm of block7,
            fill=midgreen!50!white,
            minimum width=2.5cm,
            minimum height=1cm,
            align=center] (block8) { BS-conditional entropy\\[1.2mm]{\tiny$| \widehat H_\rho(A|B) - \widehat H_\sigma(A|B) | \leq  $}\\{\tiny$  2 l_{m}^{-1} \varepsilon \log d_A + \frac{l_m + \varepsilon}{l_m}(f_{m^{-1}, m^{-1}}$}\\{\tiny$  + m^{-1}h)\Big(\frac{\varepsilon}{l_m + \varepsilon}\Big) $}\\[1.1mm]{\tiny with  $l_m = 1 - d_{\cH} m$, \, $\varepsilon \le \frac{1}{2}\| \rho - \sigma \|_1 $}};
            
        \node[draw,
            left=0.7cm of block6,
            fill=midgreen!50!white,
            minimum width=2.5cm,
            minimum height=1cm,
            align=center] (block9) {BS-entropy\\(fixed second argument)\\[1.2mm]{\tiny$ \left| \widehat D (\rho_1 \| \sigma ) - \widehat D (\rho_2 \| \sigma ) \right| $}\\{\tiny$  \leq \varepsilon \log(m_\sigma^{-1}) + (1 + \varepsilon) m_\sigma^{-1} h\Big(\frac{\varepsilon}{1 + \varepsilon}\Big) $}\\[1.1mm]{\tiny with  $\varepsilon \le \frac{1}{2}\| \rho_1 - \rho_2 \|_1 $}};
        
        \node[draw,
            below=1.1cm of block9,
            fill=midgreen!50!white,
            minimum width=2.5cm,
            minimum height=1cm,
            align=center] (block13) {Divergence bound\\[1.2mm]{\tiny$ \widehat D (\rho \| \sigma )$}\\{\tiny$  \leq \varepsilon \log(m_\sigma^{-1}) + (1 + \varepsilon) m_\sigma^{-1} h\Big(\frac{\varepsilon}{1 + \varepsilon}\Big) $}\\[1.1mm]{\tiny with $\varepsilon \le \frac{1}{2}\| \rho - \sigma \|_1 $}};
        
        \draw[-latex] (block1) -| (block2)
            node[pos=0.25,fill=white,inner sep=0]{\cite{matsumoto2010reverse, HiaiMosonyi_2011}};
         
        \draw[-latex] (block1) -| (block3)
            node[pos=0.25,fill=white,inner sep=0]{\cref{theo:theo_almost_concavity_bs_entropy}};

        \draw[-latex] (block2) -| (block4);
        \draw[-latex] (block3) -| (block4);
    
        \draw[-Turned Square] (block5) -| (block6)
        node[pos=0.77,fill=white,inner sep=0]{\cref{cor:cor_uniform_continuity_BS_conditional_mutual_information}};
        \draw[-Turned Square] (block5) -| (block7)
        node[pos=0.74,fill=white,inner sep=0]{\cref{cor:cor_uniform_continuity_BS_mutual_information}};
        \draw[-Turned Square] (block5) -| (block8)
        node[pos=0.74,fill=white,inner sep=0]{\cref{cor:cor_uniform_continuity_BS_conditional_entropy}};
        \draw[-Turned Square] (block5) -| (block9)
        node[pos=0.73,fill=white,inner sep=0]{\cref{cor:cor_uniform_continuity_BS_entropy_first_argument}};
        \draw[-latex] (block9) -- (block13)
        node[pos=0.6,fill=white,inner sep=0]{\cref{cor:cor_divergence_bound_bs_entropy}};
        
        \draw (block4) -- (block5) 
            node[pos=1,fill=white,inner sep=0.1cm] {Uniform continuity \& Continuity bounds};
    \end{tikzpicture}
    }
    \caption{In this flow chart we collect the main results from this section, starting with the almost concavity for the BS-entropy, which together with the ALAFF method outputs a plethora of continuity bounds for related entropic quantities. For the convexity and almost concavity of the BS-entropy we are setting $\rho = p \rho_1 + (1-p) \rho_2$ and $\sigma = p \sigma_1 + (1-p) \sigma_2$, with $p \in [0,1]$.  We denote by $m_\sigma$ the minimal eigenvalue of $\sigma$.  In the almost concavity bound, $\hat{c}_0$ is the maximum of $\norm{\sigma_1^{-1}}_\infty$ and $\norm{\sigma_2^{-1}}_\infty$. Additionally, we assume in all the continuity bounds that $m \leq \norm{\eta^{-1}}_\infty$, for $\eta= \sigma, \rho$.}
    \label{fig:fig_flow_chart_BS}
\end{figure}

\subsection{Almost concavity for the BS-entropy}\label{subsec:almost_concavity_BS_entropy}

In this section we prove the almost concavity of the BS-entropy and thereby complement the established result of convexity \cite[Theorem 4.4]{matsumoto2010reverse}, \cite[Corollary 4.7]{HiaiMosonyi_2011}. We first want to give some auxiliary results that will be needed later. The first of those concerns an operator inequality for the term inside the trace in the definition of the BS-entropy. 

\begin{lemma}\label{lem:almost_concavity_xlogx}~\\
    Let $A_1, A_2 \in \cB(\cH)$ positive semi-definite, $p \in [0, 1]$ and 
    \begin{equation}
        A = p A_1 + (1 - p) A_2.
    \end{equation}
    Then
    \begin{equation}
        -A\log(A) \le -p A_1 \log(A_1) - (1 - p) A_2 \log(A_2) + h_{A_1, A_2}(p)\identity
    \end{equation}
    with $h_{A_1, A_2}(p) = -p \log(p)\tr[A_1] - (1 - p) \log(1 - p)\tr[A_2]$ a distorted binary entropy.
\end{lemma}
\begin{proof}
    It holds that 
    \begin{equation}
        \begin{aligned}
             -A\log(A) &  + p A_1 \log(A_1) +  (1 - p) A_2 \log(A_2) \\
             &\le \norm{ -A\log(A)  + p A_1 \log(A_1) +  (1 - p) A_2 \log(A_2) }_1 \identity \, . \label{eq:eq_norm_upper_bound}
        \end{aligned}
    \end{equation}
    Now, since $x \mapsto -x\log(x)$ is operator concave \cite[Theorem 2.6]{Carlen-TraceInequalities-2009}, we have 
    \begin{equation}
        -A\log(A) \ge -p A_1 \log(A_1) - (1 - p) A_2 \log(A_2) \, ,
    \end{equation}
    giving us that 
    \begin{equation}
        -A\log(A) + p A_1 \log(A_1) +  (1 - p) A_2 \log(A_2) \ge 0 \, ,
    \end{equation}
    and hence
    \begin{equation}
        \begin{aligned}
            &\norm{-A\log(A) + p A_1 \log(A_1) +  (1 - p) A_2 \log(A_2)}_1\\
            &\hspace{1cm}= \tr[-A\log(A) + p A_1 \log(A_1) +  (1 - p) A_2 \log(A_2)] \, .
        \end{aligned}
        \label{eq:eq_almost_convexity_xlogx}
    \end{equation}
    We now use the operator monotonicity of the logarithm to find
    \begin{equation}
        \begin{aligned}
             -\tr[A \log(A)] &= -p\tr[A_1 \log(A)] - (1 - p) \tr[A_2 \log(A)]\\
             &\le -p \tr[A_1 \log(p A_1)] - (1 - p) \tr[A_2 \log((1 - p) A_2)]\\
             &= -p \tr[A_1 \log(A_1)] - (1 - p) \tr[A_2 \log(A_2)] + h_{A_1, A_2}(p)
        \end{aligned}
    \end{equation}
    Inserting this into \cref{eq:eq_almost_convexity_xlogx} and then into \cref{eq:eq_norm_upper_bound} yields the claimed result.
\end{proof}

The next auxiliary result concerns an equivalent formulation for the BS-entropy constructed from the function $x \mapsto x \log x$ and has already appeared in the literature (see e.g.\ \cite[Eq. (7.35)]{OhyaPetz-Entropy-1993}). We include here a short proof of this result for completeness. 

\begin{lemma}\label{lem:alternative_representation_bs_entropy}~\\
    Let $\rho \in \cS(\cH)$ and $\sigma \in \cS_+(\cH)$, then
    \begin{equation}
        \widehat{D}(\rho \Vert \sigma) = \tr[\sigma(\sigma^{-1/2} \rho \sigma^{-1/2}) \log(\sigma^{-1/2} \rho \sigma^{-1/2})] \, . 
    \end{equation}
\end{lemma}
\begin{proof}[Proof.] Slightly misusing notation, we can replace the regular $\log$ with a $\log$ that evaluates to $0$ at $0$ and thereby artificially add $0$ to the domain. This changes neither the RHS nor the LHS but allows us to derive
    \begin{equation}
        \begin{aligned}
             \widehat{D}(\rho \Vert \sigma) &= \tr[\rho \log(\rho^{1/2} \sigma^{-1} \rho^{1/2})] = \tr[\log(\rho^{1/2} \sigma^{-1/2} \sigma^{-1/2} \rho^{1/2}) \rho^{1/2} \sigma^{-1/2} \sigma^{1/2} \rho^{1/2}]\\
             &=\tr[\rho^{1/2} \sigma^{-1/2} \log(\sigma^{-1/2} \rho \sigma^{-1/2}) \sigma^{1/2} \rho^{1/2}]\\
             &= \tr[\sigma(\sigma^{-1/2} \rho \sigma^{-1/2}) \log(\sigma^{-1/2} \rho \sigma^{-1/2})] \, .
        \end{aligned}
    \end{equation}
    We used the cyclicity of the trace several times, and the well-known fact that we have $f(L^* L)L^* = L^* f(L L^*)$ in case the spectrum of $L^*L$ and $LL^*$ lie in the domain of $f$ \cite[Lemma 61.]{Wilde_Katariya}.
\end{proof}

Building on the previous results from this section, we proceed to prove now the main result, namely the almost concavity for the BS-entropy. This falls in the line of results of almost concavity discussed in \cref{def:almost_concavity_divergence}. 

\begin{theorem}[Almost concavity of the BS-entropy] \label{theo:theo_almost_concavity_bs_entropy}~\\
    Let $(\rho_1, \sigma_1), (\rho_2, \sigma_2) \in \cS_{\ker, +}$ with 
    \begin{equation}
        \cS_{\ker, +} := \{(\rho, \sigma) \in \cS(\cH) \times \cS(\cH) \;:\; \sigma \in \cS_+(\cH)\}
    \end{equation} 
    and $p \in [0, 1]$. Then, for $\rho = p \rho_1 + (1 - p) \rho_2$, $\sigma = p \sigma_1 + (1 - p) \sigma_2$, we have
    \begin{equation}
        \widehat{D}(\rho \Vert \sigma) \ge p \widehat{D}(\rho_1 \Vert \sigma_1) + (1 - p) \widehat{D}(\rho_2 \Vert \sigma_2) - \hat{c}_0 (1 - \delta_{\rho_1 \rho_2}) h(p) - f_{\hat{c}_1, \hat{c}_2}(p) \, ,
    \end{equation}
    with
    \begin{equation}
        \begin{aligned}
            h(p) &= - p \log(p) - (1 - p) \log(1 - p) \, , \\[1mm]
            f_{\hat{c}_1, \hat{c}_2}(p) &= p \log(p + \hat{c}_1 (1 - p)) + (1 - p) \log((1 - p) + \hat{c}_2p)\, , \\[1mm]
            \delta_{\rho_1 \rho_2} &= \begin{cases}
                1 & \text{if } \rho_1 = \rho_2 \\
                0 & \text{otherwise}
            \end{cases} \, ,
        \end{aligned}
    \end{equation}
    and the constants
    \begin{equation}
        \begin{aligned}
            \hat{c}_0 &:= \max\{\norm{\sigma_1^{-1}}_\infty, \norm{\sigma_2^{-1}}_\infty\} \, , \\
            \hat{c}_1 &:= \int\limits_{-\infty}^\infty dt \beta_0(t) \tr[\rho_1(\rho_1^{1/2} \sigma_1^{-1} \rho_1^{1/2})^{\frac{it + 1}{2}} \rho_1^{-1/2}\sigma_2 \rho_1^{-1/2}(\rho_1^{1/2} \sigma_1^{-1} \rho_1^{1/2})^{\frac{-it + 1}{2}}] \, ,\\
            \hat{c}_2 &:= \int\limits_{-\infty}^\infty dt \beta_0(t) \tr[\rho_2(\rho_2^{1/2} \sigma_2^{-1} \rho_2^{1/2})^{\frac{it + 1}{2}} \rho_2^{-1/2}\sigma_1 \rho_2^{-1/2}(\rho_{2}^{1/2} \sigma_2^{-1} \rho_2^{1/2})^{\frac{-it + 1}{2}}] \, , \label{eq:eq_bs_almost_concave_bound_constants}
        \end{aligned}
    \end{equation}
    with the probability density $\beta_0$ defined as in \cref{eq:beta0}.
\end{theorem}
\begin{proof}[Proof.]
    The formula for $p = 0, 1$ is trivial, hence let $p \in (0, 1)$. We find that
    \begin{equation}
        \begin{aligned}
            p \widehat{D}(\rho_1 \Vert \sigma_1)& + (1 - p) \widehat{D}(\rho_2 \Vert \sigma_2) - \widehat{D}(\rho \Vert \sigma)\\ &\le p (\widehat{D}(\rho_1 \Vert \sigma_1) - \widehat{D}(\rho_1 \Vert \sigma)) + (1 - p)(\widehat{D}(\rho_2 \Vert \sigma_2) - \widehat{D}(\rho_2 \Vert \sigma)) + \hat{c}_0 h(p) \, .
        \end{aligned}
    \end{equation}
    Indeed, as of \cref{lem:alternative_representation_bs_entropy} and then \cref{lem:almost_concavity_xlogx} with $A_1 = \sigma^{-1/2} \rho_1 \sigma^{-1/2}$, $A_2 = \sigma^{-1/2} \rho_2 \sigma^{-1/2}$ respectively, we can prove
    \begin{equation}
        \begin{aligned}
            -\widehat{D}(\rho \Vert \sigma) &= \tr[\sigma \left(-\sigma^{-1/2} \rho \sigma^{-1/2} \log(\sigma^{-1/2} \rho \sigma^{-1/2}) \right)]\\
            &\le p \tr[\sigma \left(-\sigma^{-1/2} \rho_1 \sigma^{-1/2} \log(\sigma^{-1/2} \rho_1 \sigma^{-1/2}) \right)]\\
            &\hspace{1cm} + (1 - p) \tr[\sigma \left(-\sigma^{-1/2} \rho_2 \sigma^{-1/2} \log(\sigma^{-1/2} \rho_2 \sigma^{-1/2})\right)] + h_{A_1, A_2}(p)\\
            &= -p \widehat{D}(\rho_1 \Vert \sigma) - (1 - p) \widehat{D}(\rho_2 \Vert \sigma) + h_{A_1, A_2}(p) \, .
        \end{aligned}
    \end{equation}
    At last we can estimate $\tr[A_j] = \tr[\sigma^{-1}\rho_j] \le \norm{\sigma^{-1}}_\infty \le \hat{c}_0$ for $j = 1, 2$ using Hölder's inequality, giving us $h_{A_1, A_2}(p) \le \hat{c}_0 h(p)$.
    
    We now have to estimate terms of the form $\widehat{D}(\rho_j \Vert \sigma_j) - \widehat{D}(\rho_j \Vert \sigma)$ for $j = 1, 2$. This is done using the Peierls-Bogoliubov inequality \cite{OhyaPetz-Entropy-1993} and the multivariate trace inequalities of Sutter et al. \cite{SutterBertaTomamichel-Multivariate-2017}:
    \begin{equation}
        \begin{aligned}
             \widehat{D}(\rho_j \Vert \sigma_j) - \widehat{D} (\rho_j \Vert \sigma) &= \tr[\rho_j\left(\log(\rho_j^{1/2} \sigma_j^{-1} \rho_j^{1/2}) - \log(\rho_j^{1/2} \sigma^{-1} \rho_j^{1/2})\right)]\\[1.3mm]
             &\le \tr[\exp\left(\log(\rho_j) + \log(\rho_j^{1/2}\sigma_j^{-1}\rho_j^{1/2}) - \log(\rho_j^{1/2}\sigma^{-1}\rho_j^{1/2})\right)]\\[1.3mm]
             &\le \tr[\exp\left(\log(\rho_j) + \log(\rho_j^{1/2}\sigma_j^{-1}\rho_j^{1/2}) + \log(\rho_j^{-1/2}\sigma\rho_j^{-1/2})\right)]\\
             & \le \log\Big(\int\limits_{-\infty}^\infty  dt \beta_0(t) \tr[\rho_j(\rho_j^{1/2} \sigma_j^{-1} \rho_j^{1/2})^{\frac{it + 1}{2}} \rho_j^{-1/2}\sigma \rho_j^{-1/2}(\rho_j^{1/2} \sigma_j^{-1} \rho_j^{1/2})^{\frac{-it + 1}{2}}]\Big)\\
             &= \begin{cases}
                \log(p + (1 - p) \hat{c}_1) & j = 1\\
                \log((1 - p) + p \hat{c}_2) & j = 2
             \end{cases}.
        \end{aligned}\label{eq:eq_int_ineq_bs_entropy}
    \end{equation}
    In the third line, we use that 
    \begin{equation}
        -\log(\rho_j^{1/2} \sigma^{-1} \rho_j^{1/2}) \le \log(\rho_j^{-1/2} \sigma \rho_j^{-1/2})
    \end{equation}
    which is true since for $P_{\rho_{j}}$ the projection on the support of $\rho_{j}$, we have 
    \begin{equation}
       P_{\rho_{j}} (P_{ \rho} \sigma P_{\rho_{j}})^{-1}P_{\rho_{j}} \le P_{\rho_{j}} \sigma^{-1}P_{\rho_{j}} \,,
    \end{equation}
    as $x \to x^{-1}$ is operator convex and hence fulfills the Sherman-Davis inequality \cite[Theorem 4.19]{Carlen-TraceInequalities-2009}. Note that $\sigma$ is invertible and that by $(P_{\rho_{j}} \sigma P_{\rho_{j}})^{-1}$ we mean the Moore-Penrose pseudoinverse. We find
    \begin{equation}
        \begin{aligned}
             -\log(\rho_j^{1/2} \sigma^{-1} \rho_j^{1/2}) &= -\log(\rho_j^{1/2}P_{\rho_{j}} \sigma^{-1}P_{\rho_{j}} \rho_j^{1/2}) \\
             &\le -\log(\rho_j^{1/2}P_{\rho_{j}} (P_{\rho_{j}} \sigma P_{\rho_{j}})^{-1}P_{\rho_{j}} \rho_j^{1/2})\\
             &= \log(\rho_j^{-1/2}P_{\rho_{j}}\sigma P_{\rho_{j}}\rho_j^{-1/2})\\
             &= \log(\rho_j^{-1/2} \sigma \rho_j^{-1/2}) \, .
        \end{aligned} 
    \end{equation}
    The argument why the inequalities in \cref{eq:eq_int_ineq_bs_entropy} hold in the case of $\rho_j$ not being full rank is simpler than in the case of the corresponding inequality for the relative entropy (cf.\ \cref{theo:theo_almost_concavity_relative_entropy} and \cref{sec:sec_supplements_to_proof_almost_concavity_relative_entropy}). For the BS-entropy, we can already restrict \cref{eq:eq_int_ineq_bs_entropy} to the support of $\rho_j$ as all operators involved, $\rho_j$, $\rho_j^{1/2} \sigma_j^{-1} \rho_j^{1/2}$ and $\rho_j^{1/2} \sigma^{-1} \rho_j^{1/2}$, commute with the projection onto this support. In the last step we split $\sigma$ and evaluated the first term to $p$ in case $j = 1$ or the second term in case $j = 2$ to $(1 - p)$ and left the other one untouched, respectively. This concludes the proof.
\end{proof}

\begin{remark}
    We strongly suspect that \cref{theo:theo_almost_concavity_bs_entropy} can be improved because of two reasons. The first one is that we would expect the results of almost concavity of the relative and the BS-entropy to coincide in the case that the involved states commute. The reason is that in this case, both quantities reduce to the classical relative entropy. A straightforward calculation shows that then $\hat{c}_1 = c_1$ and $\hat{c}_2 = c_2$, hence $f_{c_1, c_2} = f_{\hat{c}_1, \hat{c}_2}$ , but $h \le  \hat{c}_0 h$ with equality if, and only if, $\sigma_1$ and $\sigma_2$ are pure, which in addition to $\sigma_1$ and $\sigma_2$ being full rank means $\cH = \C$.\par
    The other reason is given by the continuity bound we obtain for the BS-conditional entropy in \cref{cor:cor_uniform_continuity_BS_conditional_entropy}. Numerics suggest an almost convex bound that is independent of the minimal eigenvalue (cf. \cref{fig:bs_conditional_entropy_remainder}) if the inputs are full rank\footnote{The full rank requirement is necessary, as we will show in \cref{prop:discontinuity_conditional_BS_entropy} that the BS-conditional entropy is discontinuous in the presence of vanishing eigenvalues}. Hence we would also suspect that an optimal almost concave remainder of the BS-entropy reduces to an eigenvalue independent bound in the case of the BS-conditional entropy.
\end{remark}

Analogous to the case of the relative entropy we provide an additional proposition to give context to the above result, i.e. to provide simpler expressions if the involved states satisfy specific conditions.

\begin{proposition}[Almost concavity estimate of the BS-entropy is well behaved]
    \label{prop:prop_almost_concave_estimate_bs_entropy_well_behaved}~\\
    The function $\hat{c}_0h + f_{\hat{c}_1, \hat{c}_2}$ obtained in \cref{theo:theo_almost_concavity_bs_entropy} is well behaved in the following sense: Let $j = 1, 2$ and $(\rho_j, \sigma_j) \in \cS_{\ker, +}$. We have the following:
    \begin{enumerate}
        \item If $\sigma_1 = \sigma_2$, then $\hat{c}_j = 1$, resulting in $f_{\hat{c}_1, \hat{c}_2} + \hat{c}_0h = \hat{c}_0h$.
        \item If the $\sigma_j$ have a minimal eigenvalue that is bounded from below by $m > 0$ respectively, then $f_{\hat{c}_1, \hat{c}_2} + \hat{c}_0h \le f_{m^{-1}, m^{-1}} + m^{-1}h$.
        \item If $\cH = \cH_A\otimes \cH_B$ is a bipartite space, $\rho_j$ has a minimal eigenvalue bounded from below by $m > 0$, and further $\sigma_j =d_A^{-1} \identity_A \otimes \rho_{j,B}$, then $f_{\hat{c}_1, \hat{c}_2} + \hat{c}_0h \le f_{m^{-1}, m^{-1}} + m^{-1}h$.
        \item We find for $m_1, m_2 \ge 1$, $p \mapsto \frac{1}{1 - p}f_{m_1, m_2}(p)$ and $p \mapsto \frac{1}{1 - p}\hat{c}_0 h(p)$ are non-decreasing on $[0, 1)$.
    \end{enumerate}
\end{proposition}

This result should be compared to \cref{prop:prop_almost_concave_estimate_relative_entropy_well_behaved}, its analogue for the relative entropy. The proof can be found in \cref{sec:sec_proof_almost_concave_estimate_BS_entropy_well_behaved}. We will use the reductions from \cref{prop:prop_almost_concave_estimate_bs_entropy_well_behaved} to simplify the terms in \cref{theo:theo_almost_concavity_bs_entropy} for the various applications presented in the subsequent section. 

\subsection{Continuity bounds for the BS-entropy}\label{subsec:subsec_continuity_bounds_bs_entropy}

In this section, we will use the almost concavity for the BS-entropy from \cref{theo:theo_almost_concavity_bs_entropy} together with the ALAFF method in its full generality.

Before we dive into the continuity and divergence bounds, we want to collect some lower and upper estimates of entropic quantities derived from the BS-entropy (see \cref{sec:basic_notation} for the specific definitions).

\begin{proposition}[Bounds on BS-entropic quantities]\label{prop:prop_bound_BS_entropic_quantities}~\\
    For $\rho \in \cS(\cH_A \otimes \cH_B)$, we find:
    \begin{enumerate}
        \item For the BS-conditional entropy:
    \begin{equation} \label{eq:bounds-BS-cond-ent}
        -\log\min\{d_A, d_B\} \le \widehat{H}_\rho(A|B) \le \log d_A \, .
    \end{equation}
    \item For the BS-mutual information:
    \begin{equation}
        0 \le \widehat{I}_\rho(A:B) \le \log\min\{d_A, d_B\} + \log\min\{\norm{\rho_A^{-1}}_\infty, \norm{\rho_B^{-1}}_\infty\} \, , \label{eq:eq_bound_BS_mutual_information}
    \end{equation}
    with $\cdot^{-1}$ the Moore-Penrose pseudoinverse. 
    \item For $\rho \in \cS(\cH_A \otimes \cH_B \otimes \cH_C)$, we find that the BS-conditional mutual information satisfies:
       \begin{equation}
        0 \le \widehat{I}_\rho(A:B |C) \le \log\min\{ d_A^2, d_{ABC}\}.
    \end{equation}
    \end{enumerate}
    The first two bounds are shown to be tight. For the third one, we expect that similar reasoning should also show its tightness.
\end{proposition}
The proof can be found in \cref{sec:sec_proof_bound_BS_entropic_quantities}. We further want to remark that the scaling of the bound with respect to $\norm{\rho_A^{-1}}_{\infty}$ or $\norm{\rho_B^{-1}}_\infty$ is justified. The reasoning can be found in \cref{sec:sec_proof_bound_BS_entropic_quantities} as well.

\subsubsection{Uniform continuity for the BS-conditional entropy}

We encounter a slight complication when it comes to the uniform continuity of the BS-conditional entropy compared to the uniform continuity of the conditional entropy that we have covered in \cref{cor:continuity_bound_conditional_entropy}. This is because the almost concave bound of the BS-entropy depends on the minimal eigenvalue of the second argument (see \cref{eq:eq_bs_almost_concave_bound_constants}), i.e. it has to be full rank. This means the input to the BS-conditional entropy has to be full rank as well. Although we think that the result of almost concavity for the BS-entropy can be improved, we know that there is no extension of uniform continuity nor continuity for the BS-conditional entropy to positive semi-definite states, as this quantity is not continuous on those. This is the content of the next proposition. We also refer the reader to \cite[Remark 3.3]{FawziFawzi-SDPhierarchies-2021} for a similar behaviour of the sharp quantum Rényi divergences.

\begin{proposition}[Discontinuity of the BS-conditional entropy]\label{prop:discontinuity_conditional_BS_entropy}~\\
    The BS-conditional entropy is discontinuous on the set of positive semi-definite operators over $\cH_A \otimes \cH_B$ if $d_A, d_B \ge 2$.
\end{proposition}
\begin{proof}
    Since $d_A \ge 2$ as well as $d_B \ge 2$, we find orthogonal $\ket{i_A} \in \cH_A$, $\ket{i_B} \in \cH_B$, $i = 0, 1$. For $\varepsilon \in (0, 1)$ we then define
    \begin{equation}
        \ket{\varepsilon_B} = \sqrt{1 - \varepsilon}\ket{0_B} + \sqrt{\varepsilon}\ket{1_B} \, ,
    \end{equation}
    which is clearly normalised. Furthermore,
    \begin{equation}
        \begin{aligned}
            \rho_0 &:= \frac{1}{2} (\dyad{0_A} + \dyad{1_A}) \otimes \dyad{0_B} \, , \\
            \rho_{\varepsilon} &:=\frac{1}{2} \dyad{0_A} \otimes \dyad{0_B}  + \frac{1}{2} \dyad{1_A} \otimes \dyad{\varepsilon_B} \, ,
        \end{aligned}
    \end{equation}
    The above are states and fulfil
    \begin{equation}
        \begin{aligned}
            \norm{\rho_0 - \rho_{\varepsilon}}_1 &= \frac{1}{2}\norm{\dyad{1_A} \otimes (\dyad{0_B} - \dyad{\varepsilon_B})}_1 \\
            &= \frac{1}{2}\norm{\dyad{0_B} - \dyad{\varepsilon_B}}_1 = \sqrt{\varepsilon} \, .
        \end{aligned} \label{eq:eq_discontinuity_BS_conditional_entropy_norm_distance}
    \end{equation}
    To see the last equality, we can identify the subspace spanned by $\ket{0_B}$ and $\ket{1_B}$ with $\C^2$ and then get that 
    \begin{equation}
        \dyad{0_B} \to \begin{pmatrix}
            1 & 0 \\
            0 & 0
        \end{pmatrix} \quad \text{and} \quad 
        \dyad{\varepsilon_B} \to \begin{pmatrix}
            1 - \varepsilon & \sqrt{\varepsilon}\sqrt{1 - \varepsilon}\\
            \sqrt{\varepsilon}\sqrt{1 - \varepsilon} & \varepsilon
        \end{pmatrix}.
        \label{eq:eq_matrix_representation_discontinuity_BS_conditional_entropy}
    \end{equation}
    Calculating the eigenvalues of the difference and taking the sum of their absolute value gives $2 \sqrt{\varepsilon}$ and thereby \cref{eq:eq_discontinuity_BS_conditional_entropy_norm_distance}. Since clearly $[\rho_{0}, \identity \otimes \tr_A[\rho_{0}]] = 0$, the BS and conditional entropy coincide and we find 
    \begin{equation}
        \begin{aligned}
            \hat{H}_{\rho_0}(A|B) &= \tr[\dyad{0_B} \log \dyad{0_B}]\\
            &\hspace{1cm}- \tr[\frac{1}{2} (\dyad{0_A} + \dyad{1_A}) \otimes \dyad{0_B} \log\frac{1}{2} (\dyad{0_A} + \dyad{1_A}) \otimes \dyad{0_B}]\\
            &= 0 - \log\frac{1}{2} = \log 2 \, .
        \end{aligned}
    \end{equation}
    The result for $\rho_{\varepsilon}$ cannot be calculated so easily. We have that 
    \begin{equation}
        \begin{aligned}
            \widehat{H}_{\rho_{\varepsilon}}(A|B) &= -\frac{1}{2}\tr[\dyad{0_B} \log(\dyad{0_B}^{1/2} (\dyad{\varepsilon_B} + \dyad{0_B})^{-1} \dyad{0_B}^{1/2})]\\
            &\hspace{1cm} -\frac{1}{2}\tr[\dyad{\varepsilon_B} \log(\dyad{\varepsilon_B}^{1/2} (\dyad{\varepsilon_B} + \dyad{0_B})^{-1} \dyad{\varepsilon_B}^{1/2})]\\
            &= -\frac{1}{2} \log \tr[\dyad{0_B}(\dyad{\varepsilon_B} + \dyad{0_B})^{-1}]\\
            &\hspace{1cm} -\frac{1}{2}\log\tr[\dyad{\varepsilon_B}(\dyad{\varepsilon_B} + \dyad{0_B})^{-1}] \, ,
        \end{aligned}
        \label{eq:eq_BS_conditional_entropy_malicious_state}
    \end{equation}
    where in the first equality we used that $\dyad{0_B} \dyad{1_B} = \dyad{1_B} \dyad{0_B} = 0$ and in the second equality that $\dyad{\varepsilon_B}$ and $\dyad{0_B}$ are rank-one projections. We find, using again the matrix representation in \cref{eq:eq_matrix_representation_discontinuity_BS_conditional_entropy}, that
    \begin{equation}
        (\dyad{\varepsilon_B} + \dyad{0_B})^{-1} \to \begin{pmatrix}
            1 & \frac{\varepsilon - 1}{\sqrt{\varepsilon}\sqrt{1 - \varepsilon}}\\
            \frac{\varepsilon - 1}{\sqrt{\varepsilon}\sqrt{1 - \varepsilon}} & \frac{2}{\varepsilon} - 1
        \end{pmatrix} \, .
    \end{equation}
    By forming matrix products and calculating the trace, we can immediately conclude that
    \begin{equation}
        \begin{aligned}
            \tr[\dyad{\varepsilon_B}(\dyad{\varepsilon_B} + \dyad{0_B})^{-1}] &= 1 \, , \\
            \tr[\dyad{0_B}(\dyad{\varepsilon_B} + \dyad{0_B})^{-1}] &= 1 \, .
        \end{aligned}
    \end{equation}
    If we insert this into \cref{eq:eq_BS_conditional_entropy_malicious_state}, we get $\widehat{H}_{\rho_{\varepsilon}}(A|B) = 0$.
\end{proof}

This previous result shows in particular that we could only expect continuity and uniform continuity for the BS-conditional entropy on the set of full-rank states. The presence of the minimal eigenvalue of the states in the continuity bound provided below is thus not surprising.  

\begin{corollary}[Uniform continuity of the BS-conditional entropy]\label{cor:cor_uniform_continuity_BS_conditional_entropy}~\\
    The BS-conditional entropy over the bipartite Hilbert space $\cH= \cH_A \otimes \cH_B$ is for $ d_{\cH}^{-1} > m > 0$ uniformly continuous on $\cS_0= \cS_{\ge m}(\cH)$ and for $\rho, \sigma \in \cS_0$ with $\frac{1}{2}\norm{\rho - \sigma}_1 \le \varepsilon \le 1$ it holds that 
    \begin{equation}
        |\widehat{H}_\rho(A|B) - \widehat{H}_\sigma(A|B)| \le 2 l_{m}^{-1} \varepsilon \log d_A + \frac{l_m + \varepsilon}{l_m}(f_{m^{-1}, m^{-1}} + m^{-1}h)\Big(\frac{\varepsilon}{l_m + \varepsilon}\Big) \, ,
    \end{equation}
    with $l_m = 1 - d_{\cH} m$.
\end{corollary}
\begin{proof}
    We find that $\cS_0$ is $s$-perturbed $\Delta$-invariant with $s = m d_{\cH}$. The justification of this choice is completely analogous to the reasoning in \cref{lem:delta_invariant_subset} with $\rho = d_{\cH}^{-1} \identity$, i.e. the maximally mixed state. Furthermore, $f(\cdot) = \widehat{H}_\cdot(A|B)$ is ALAFF with $a_f = 0$ as $\widehat{H}_\cdot(A|B)$ is concave, and $b_f = m^{-1}h + f_{m^{-1}, m^{-1}}$ since the result in \cref{subsec:almost_concavity_BS_entropy} becomes independent of the states as we restrict to $\widehat{H}_\cdot(A|B)$ using point 3 of \cref{prop:prop_almost_concave_estimate_bs_entropy_well_behaved}. We further find that 
    \begin{equation}
        C^s_f \le \sup\limits_{\rho_1, \rho_2 \in \cS(\cH)} |\widehat{H}_{\rho_1}(A|B) - \widehat{H}_{\rho_2}(A|B)| \le 2 \log d_A \, ,
    \end{equation}
    using \cref{prop:prop_bound_BS_entropic_quantities}. This allows us to apply \cref{theo:theo_alaff_method} where $E_f^{\max}$ coincides with $E_f$ as of point 4 in \cref{prop:prop_almost_concave_estimate_bs_entropy_well_behaved}. This concludes the claim. 
\end{proof}

Even though a continuity bound for the BS-conditional entropy can only be proven for positive definite states, numerical simulations show us that we could expect a tighter bound on the previous proposition coinciding with that of \cref{cor:continuity_bound_conditional_entropy}, i.e., without the dependence on the minimal eigenvalues of the states involved. One can find a visualisation of those numeric simulations that underlie the conjecture in \cref{fig:bs_conditional_entropy_remainder}. The possibility of obtaining such a tighter bound is left for future work. 

\begin{figure}[ht!]
    \centering\includegraphics{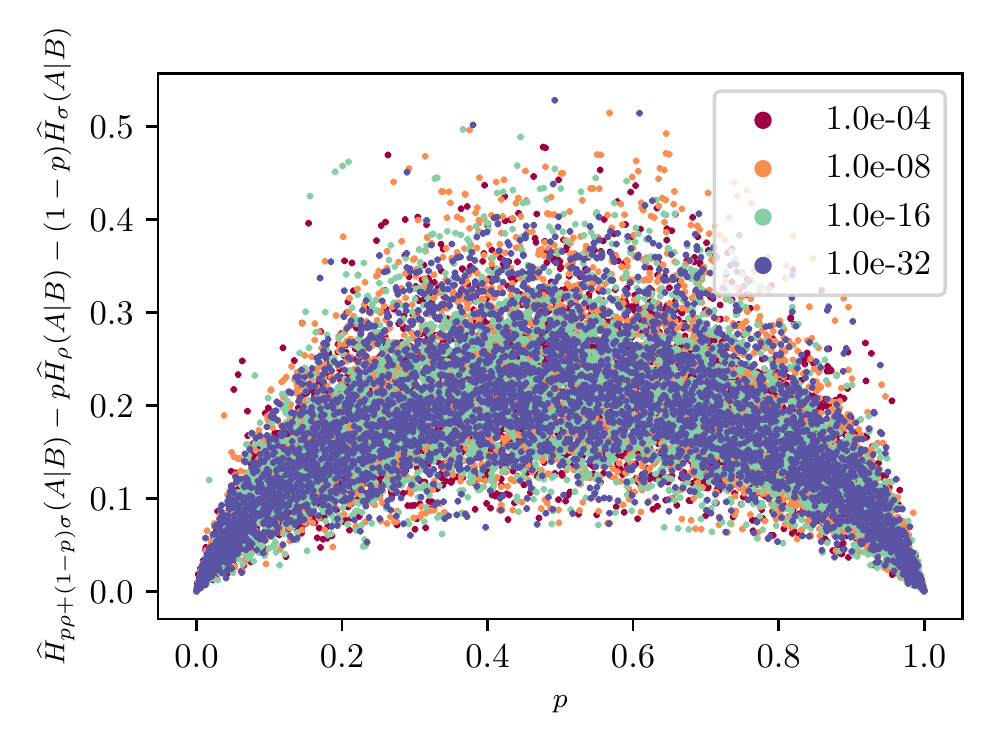}
    \caption{We investigate the dependence of the almost convex remainder term of the BS-conditional entropy on the minimal eigenvalue of the involved states. For the minimal eigenvalues $10^{-4}, 10^{-8}, 10^{-16}, 10^{-32}$ we sampled five hundred pairs of qubits $(\rho, \sigma)$ both of them with controlled eigenvalues. We then sampled for every state pair ten values of $p$, the convex interpolation parameter, and plotted the remainder. As can be seen from the plot, the remainder appears to be independent of the minimal eigenvalue and the shape suggests a binary entropy or Gini impurity. The result shows a similar pattern if the dimension is increased.}
    \label{fig:bs_conditional_entropy_remainder}
\end{figure}

\subsubsection{Uniform continuity for the BS-mutual information}

Let us address now the case of the BS-mutual information. Since the BS-conditional entropy is a particular case of the latter (by assuming that one of the reduced states of $\rho_{AB}$ is maximally mixed), the discontinuity issues presented in the previous subsection are expected to arise in the current one as well. More specifically, the example of discontinuity of the BS-conditional entropy presented in \cref{prop:discontinuity_conditional_BS_entropy} also constitutes an example of discontinuity of the BS-mutual information. Thus, we can only expect to prove uniform continuity for the BS-mutual information for full-rank states

However, there is a subtle difference between the settings of the BS-conditional entropy and the BS-mutual information. As shown in \cref{prop:prop_bound_BS_entropic_quantities}, the former is bounded between the same values as the (usual) conditional entropy, whereas the latter presents some pathological behaviour. Pathological in the sense that its (tight) upper bound depends on the minimal eigenvalues of the reduced state, as shown in \cref{eq:eq_bound_BS_mutual_information}. For this reason, a continuity bound for the BS-mutual information will necessarily depend on the minimal eigenvalues of the states involved.

\begin{corollary}[Uniform continuity for the BS-mutual information]\label{cor:cor_uniform_continuity_BS_mutual_information}~\\
    The BS-mutual information on a bipartite Hilbert space $\cH = \cH_A \otimes \cH_B$ is for $d_{\cH}^{-1} > m > 0$ uniformly continuous on $\cS_0 = \cS_{\ge m}$ and for $\rho, \sigma \in \cS_0$ with $\frac{1}{2}\norm{\rho - \sigma}_1 \le \varepsilon \le 1$ we find that 
    \begin{equation}
    \begin{aligned}
        |\widehat{I}_\rho(A:B) - \widehat{I}_\sigma(A:B)| & \le 2 l_m^{-1} \varepsilon( \log\min\{d_A, d_B\} + \log m^{-1}) + \frac{l_m + \varepsilon}{l_m}z_m \Big(\frac{\varepsilon}{l_m + \varepsilon}\Big) \\
        & \leq \frac{ 2\log\min\{d_A, d_B \} + 4\log m^{-1} + (\sqrt{2} + 2)m^{-1} + \sqrt{2}}{l_{m}} \sqrt{\varepsilon}\, ,
        \end{aligned}
    \end{equation}
    with $l_m = 1 - m d_{\cH}$ and 
    \begin{equation}
        z_m(p) = 2 f_{m^{-1}, m^{-1}}(p) + (m^{-1} + 1) h(p) \, .
    \end{equation}
    For the second inequality, we used \cref{lemma:nice_bounds}, $\log(1 + x) \le x$ for $0 \le x$, $\varepsilon \leq \sqrt{\varepsilon}$ for $\varepsilon \in [0, 1]$ and $l_{m} \le l_{m} + \varepsilon$.
\end{corollary}
\begin{proof}
    As in the case of the BS-conditional entropy, we find that $\cS_0$ is $s$-perturbed $\Delta$-invariant with $s = m d_{\cH}$. To conclude that $\widehat{I}_\cdot(A:B)$ is ALAFF we first note that because of the convexity of $\widehat{D}(\cdot \Vert \cdot)$,
    \begin{equation}
        \begin{aligned}
            \widehat{I}_{p \rho_1 + (1 - p) \rho_2}(A:B) &\le p \widehat{D}(\rho_1 \Vert \rho_{1, A} \otimes(p \rho_{1, B}+ (1 - p) \rho_{2, B}))\\
            &\hspace{1cm} + (1 - p) \widehat{D}(\rho_2 \Vert \rho_{2, A} \otimes(p \rho_{1, B} + (1 - p)  \rho_{2, B}))\\
            &\le p \widehat{I}_{\rho_1}(A:B) + (1 - p) \widehat{I}_{\rho_2}(A:B) + h(p) \, .
        \end{aligned}
    \end{equation}
    In the last step, we used that $\widehat{D}(\cdot \Vert \cdot)$ is monotone decreasing in its second argument, and $p \rho_{1, B} \le p \rho_{1, B} + (1 - p) \rho_{2, B}$, $(1 - p) \rho_{2, B} \le p \rho_{1, B} + (1 - p) \rho_{2, B}$, respectively. Hence $b_f = h$. We follow similar lines to obtain $a_f$. Starting with \cref{theo:theo_almost_concavity_bs_entropy} and point 2 in \cref{prop:prop_almost_concave_estimate_bs_entropy_well_behaved} using that $\norm{\rho_A^{-1}}_\infty \le \norm{\rho_{AB}^{-1}}_\infty$, and analogously for $\rho_{B}$, we find 
    \begin{equation}
        \begin{aligned}
            \widehat{I}_{p \rho_1 + (1 - p) \rho_2}(A:B) &\ge p \widehat{D}(\rho_1 \Vert \rho_{1, A} \otimes(p \rho_{1, B}+ (1 - p) \rho_{2, B}))\\
            &\hspace{1cm} + (1 - p) \widehat{D}(\rho_2 \Vert \rho_{2, A} \otimes(p \rho_{1, B} + (1 - p)  \rho_{2, B})) - m^{-1}h(p) - f_{m^{-1}, m^{-1}}(p)\\
            &\ge p\widehat{I}_{\rho_1}(A:B) + (1 - p) \widehat{I}_{\rho_2}(A:B) - m^{-1}h(p) - 2 f_{m^{-1}, m^{-1}}(p) \, .
        \end{aligned}
    \end{equation}
    In the last step we used again that $\widehat{D}(\cdot \Vert \cdot)$ is monotone decreasing in its second argument and that $p \rho_{1, AB} + (1 - p) \rho_{2, AB} \le (p + (1 - p) m^{-1}) \rho_{1, AB}$ and $p \rho_{1, AB} + (1 - p) \rho_{2, AB} \le (m^{-1}p + (1 - p)) \rho_{2, AB}$, giving us another $f_{m^{-1}, m^{-1}}(p)$. Hence $a_f = m^{-1} h + 2 f_{m^{-1}, m^{-1}}$. We conclude the proof by noticing again that $\norm{\rho_A^{-1}}_\infty \le \norm{\rho_{AB}^{-1}}_\infty \le m^{-1}$, yielding the upper bound
    \begin{equation}
        C^s_f \le \sup\limits_{\rho \in \cS_0} \widehat{I}_\rho(A:B) \le \log\min\{ d_A, d_B\} + \log m^{-1} \, .
    \end{equation}
    Finally, we apply \cref{theo:theo_alaff_method} and get the claimed bounds as $E_f$ coincides with $E_f^{\max}$, due to point 4 in \cref{prop:prop_almost_concave_estimate_bs_entropy_well_behaved}.
\end{proof}

\subsubsection{Uniform continuity for the BS-conditional mutual information}

Next, we provide a result of uniform continuity for the BS-conditional mutual information, defined in \cref{eq:eq_BS_conditional_mutual_information}. As a difference between two BS-conditional entropies, it will not present the pathological behaviour from the BS-mutual information, as the BS-conditional entropies are bounded between the same limits as the (usual) conditional entropies. See \cref{prop:prop_bound_BS_entropic_quantities} for the specific bounds on all these BS-entropic quantities.

Nevertheless, the continuity bound we obtain below for the BS-conditional mutual information also depends on the minimal eigenvalues of the states involved, as happened in the case of the BS-conditional entropies.

\begin{corollary}[Uniform continuity of the BS-conditional mutual information]\label{cor:cor_uniform_continuity_BS_conditional_mutual_information}~\\
    The BS-conditional mutual information over $\cH = \cH_A \otimes \cH_B \otimes \cH_C$ is for $d_{\cH}^{-1} > m > 0$ uniformly continuous on $\cS_0 = \cS_{\ge m}(\cH)$ and for $\rho, \sigma \in \cS_0$ with $\frac{1}{2}\norm{\rho - \sigma}_1 \le \varepsilon \le 1$  we find that 
    \begin{equation}
    \begin{aligned}
        |\widehat{I}_\rho(A:B|C) - \widehat{I}_\sigma(A:B|C)| &\le 2 \, \varepsilon\, l_m^{-1} \log\min\{ d_A, \sqrt{d_{ABC}}\} + 2g_{m}(\varepsilon)  \\ 
        & \le \frac{ 2\log\min\{d_A, \sqrt{d_{ABC}} \} + 2 \log m^{-1} + 2 (\sqrt{2} + 1) m^{-1}}{l_m} \sqrt{\varepsilon} \, ,
        \end{aligned}
    \end{equation}
    with $l_m = 1 - m d_{\cH}$ and 
    \begin{equation}
        g_{m}(\varepsilon) = \frac{l_m + \varepsilon}{l_m}(f_{m^{-1}, m^{-1}} + m^{-1}h)\Big(\frac{\varepsilon}{l_m + \varepsilon}\Big) \, .
    \end{equation}
     For the second inequality, we used \cref{lemma:nice_bounds}, $\log(1 + x) \le x$ for $0 \le x$, $\varepsilon \leq \sqrt{\varepsilon}$ for $\varepsilon \in [0, 1]$ and $l_{m} \le l_{m} + \varepsilon$.
\end{corollary}
\begin{proof}
    We have that $\cS_0$ is $s$-perturbed $\Delta$-invariant using the same reasoning as in the proof of \cref{cor:cor_uniform_continuity_BS_conditional_entropy}. Because of the representation $\widehat{I}_\cdot(A:B|C) = \widehat{H}_\cdot(A|C) - \widehat{H}_\cdot(A|BC)$ we can immediately conclude that $\widehat{I}_\cdot(A:B|C)$ is ALAFF with $a_f = f_{m^{-1}, m^{-1}} + m^{-1}h$ and $b_f = f_{m^{-1}, m^{-1}} + m^{-1}h$ arguing along the same lines as in \cref{cor:cor_uniform_continuity_BS_conditional_entropy}. Using \cref{prop:prop_bound_BS_entropic_quantities} we can conclude 
    \begin{equation}
        C^s_f \le \sup\limits_{\rho \in \cS(\cH)} \widehat{I}_\rho(A:B|C) \le 2 \log\min\{ d_A, \sqrt{d_{ABC}}\} \, .
    \end{equation}
    Applying \cref{theo:theo_alaff_method} and using point 4 of \cref{prop:prop_almost_concave_estimate_bs_entropy_well_behaved} we get that $E_f = E_f^{\max}$ and thereby conclude the assertion.
\end{proof}

\subsubsection{Divergence bound for the BS-entropy}

We conclude this section by following the same lines as in the case of the relative entropy to provide a divergence bound for the BS-entropy. Firstly, we will prove the uniform continuity of the BS-entropy in the first argument and subsequently derive from that result the divergence bound. These results should be compared to their relative entropy analogues, namely \cref{cor:cor_uniform_continuity_relative_entropy_first_argument} and \cref{cor:divergence_bound_relative_entropy}, respectively.

\begin{corollary}[Uniform continuity of the BS-entropy in the first argument]\label{cor:cor_uniform_continuity_BS_entropy_first_argument}~\\
    Let $\sigma \in \cS_+(\cH)$ be fixed. Then $\widehat{D}(\cdot \Vert \sigma)$ is uniformly continuous on $\cS_0 = \cS(\cH)$, and for $\rho_1, \rho_2 \in \cS_0$ with $\frac{1}{2}\norm{\rho_1 - \rho_2} \le \varepsilon \le 1$ we find that
    \begin{equation}
        |\widehat{D}(\rho_1 \Vert \sigma) - \widehat{D}(\rho_2 \Vert \sigma)| \le \varepsilon \log(m_\sigma^{-1}) + (1 + \varepsilon) m_\sigma^{-1} h\Big(\frac{\varepsilon}{1 + \varepsilon}\Big) \, ,
    \end{equation}
    with $m_\sigma$ the minimal eigenvalue of $\sigma$.
\end{corollary}
\begin{proof}
    The procedure is familiar. First, $\cS_0$ is 0-perturbed $\Delta$-invariant. Second $f(\cdot) = \widehat{D}(\cdot \Vert \sigma)$ is ALAFF with $a_f = m_\sigma^{-1} h$ and $b_f = 0$ employing \cref{theo:theo_almost_concavity_bs_entropy} and point 1 of \cref{prop:prop_almost_concave_estimate_bs_entropy_well_behaved}. Further 
    \begin{equation}
        C_f^\perp \le \sup\limits_{\rho \in \cS(\cH)} \widehat{D}(\rho \Vert \sigma) \le \log m_{\sigma}^{-1}
    \end{equation}
    since $\rho^{1/2} \sigma^{-1} \rho^{1/2} \leq \identity m_\sigma^{-1}$. Applying now \cref{rem:rem_s=0_alaff_method} gives the claimed result.
\end{proof}

Utilizing the above result we obtain a divergence bound for the BS-entropy which constitutes the analogue to the one of the relative entropy in \cref{cor:divergence_bound_relative_entropy}. Note that even the divergence bounds obtained in both cases are similar, except for the presence of a factor $m_\sigma^{-1}$ in the second term of the bound.

\begin{corollary}[Divergence bound for the BS-entropy]\label{cor:cor_divergence_bound_bs_entropy}~\\
    Let $\rho \in \cS(\cH)$ and $\sigma \in \cS_+(\cH)$, then for $\frac{1}{2}\norm{\rho - \sigma}_1 \le \varepsilon \le 1$, we have
    \begin{equation}
        \widehat{D}(\rho \Vert \sigma)  \le \varepsilon \log m_{\sigma}^{-1} + (1 + \varepsilon) m_\sigma^{-1} h \Big(\frac{\varepsilon}{1 + \varepsilon}\Big) \, ,
    \end{equation}
    with $m_\sigma$ the minimal eigenvalue of $\sigma$.
\end{corollary}
\begin{proof}
    In the context of \cref{cor:cor_uniform_continuity_BS_entropy_first_argument}, we just set $\rho_1 = \rho$ and $\rho_2 = \sigma$, giving us that $\frac{1}{2}\norm{\rho_1 - \rho_2}_1 = \frac{1}{2}\norm{\rho - \sigma}_1 \le \varepsilon \le 1$. Further $\widehat{D}(\rho_2 \Vert \sigma) = \widehat{D}(\sigma \Vert \sigma) = 0$ and $|\widehat{D}(\rho_1 \Vert \sigma)|$ loses the absolute value, as $\widehat{D}(\cdot \Vert \cdot) \ge 0$. The bound follows immediately.
\end{proof}

With this, we conclude our section on continuity bounds for entropic quantities derived from the BS-entropy. We have deliberately omitted the analogues of \cref{cor:cor_uniform_continuity_relative_entropy_second_argument} and \cref{theo:theo_uniform_continuity_relative_entropy} for the BS-entropy, due to their high technicality and the complexity of the continuity bounds that we would obtain with our method. However, the same procedure as for the relative entropy would give analogous continuity bounds also in this setting.

%% file: sections/applications.tex
In this section, we use some of the previously derived bounds to provide applications in various contexts within the field of quantum information.

\subsection{Quantum hypothesis testing} \label{sec:hyp-test}

In this section, we interpret our bounds in terms of hypothesis testing. Quantum state discrimination and quantum hypothesis testing are both well-studied tasks in quantum information theory. 

In quantum state discrimination, you are given a source which prepares quantum states $\rho_1$ and $\rho_2$ with equal probability. The task is to perform a measurement in order to identify whether the state prepared by the source is $\rho_1$ or $\rho_2$. In this setting, the optimal probability of successfully identifying the state is given in terms of the trace distance as
\begin{equation} \label{eq:succ-prop-helstrom}
    p_{\mathrm{succ}} = \frac{1}{2}\left(1+\frac{1}{2}\|\rho_1-\rho_2\|_1 \right)
\end{equation}
using the Helstrom measurement (see textbooks such as \cite{Nielsen2010}).

In quantum hypothesis testing, we consider an asymmetric setting with $n$ copies and we are interested in the asymptotic performance. Again, the task is to discriminate between $\rho$ and $\sigma$, using a measurement $\{E, \identity-E\}$ where $0 \leq E \leq \identity$. Upon the first outcome, the guess is $\rho$, and upon the second $\sigma$. Therefore, we define the errors of the first and second kind as
\begin{equation}
    \alpha(E)_n = \mathrm{Tr}[\rho^{\otimes n}(\mathds{1}-E)]
\end{equation}
and 
\begin{equation}
    \beta(E)_n = \mathrm{Tr}[\sigma^{\otimes n}E] \, . 
\end{equation}
We now want to fix the error of the first kind to be at most $\varepsilon$ and define
\begin{equation}
    \beta_\varepsilon(\rho^{\otimes n} || \sigma^{\otimes n}) := \min \{\beta(E)_n: \alpha(E)_n \leq \varepsilon\},
\end{equation}
where the minimum runs over $0 \leq E \leq \mathds{1}$. Then, the quantum Stein's lemma \cite{hiai1991proper, Petz2008} states that 
\begin{equation}
    \lim_{n \to \infty} \frac{1}{n} \log[  \beta_\varepsilon(\rho^{\otimes n} || \sigma^{\otimes n})] = - D(\rho||\sigma).
\end{equation}
Therefore, we can interpret the continuity bound in the way that two states that are hard to discriminate have almost the same performance in terms of hypothesis testing. We can illustrate this with \cref{cor:cor_uniform_continuity_relative_entropy_first_argument}, just by taking $1+\varepsilon$ there to be $2 p_{\mathrm{succ}}$ following \cref{eq:succ-prop-helstrom}. 

\begin{corollary}
    Let $\sigma \in \cS(\cH)$ be fixed, $0 < \varepsilon < 1$ and let us consider a source which produces $\rho_1$, $\rho_2$ with equal probability. Moreover, let $p$ be an upper bound on the probability $p_{\mathrm{succ}}$ of successfully identifying the state. Then, the difference in the asymptotic error exponent in hypothesis testing is bounded by
    \begin{equation}
        \Big|\lim_{n \to \infty} \frac{1}{n} \log[  \beta_\varepsilon(\rho_1^{\otimes n} || \sigma^{\otimes n})] - \lim_{n \to \infty} \frac{1}{n} \log[  \beta_\varepsilon(\rho_2^{\otimes n} || \sigma^{\otimes n})] \Big| \le (2p-1) \log \widetilde{m}_\sigma^{-1} + 2p \, h \Big(\frac{2p-1}{2p}\Big) \, ,
    \end{equation}
    with $\widetilde{m}_\sigma^{-1}$ the minimal non-zero eigenvalue of $\sigma$.
\end{corollary}

\subsection{Free energy}\label{sec:free_energy}

In \cref{sec:hyp-test}, we already saw one interpretation of our results in terms of hypothesis testing. This section gives another interpretation using the language of quantum thermodynamics. 

A ubiquitous quantity in quantum thermodynamics is free energy. To define it, we need to fix a Hamiltonian $H \in \mathcal B(\mathcal H)$, $H = H^\ast$, and some inverse temperature $\beta > 0$. The Gibbs state of this system, describing a quantum system in thermal equilibrium, is 
\begin{equation}
    \rho_\beta(H) = \frac{e^{-\beta H}}{\mathrm{tr}[e^{-\beta H}]} \, .
\end{equation}
Now, we can define the free energy as
\begin{equation}
    F(\rho) = \mathrm{tr}[H\rho] - \beta^{-1} S(\rho) \, .
\end{equation}
It can be related to the relative entropy as
\begin{equation} \label{eq:rel-ent-free-energy}
    D(\rho||\rho_\beta(H)) = \beta (F(\rho) - F(\rho_\beta(H))) \, ,
\end{equation}
which can easily be verified by direct computation.

Inspired by quantum information theory, in particular entanglement theory, during the last years various descriptions of quantum thermodynamics as a resource theory have emerged. Resource theories are described in terms of free states and free operations. In quantum thermodynamics, the free state is $\rho_\beta(H)$, whereas the choices of free operations can differ. Possible choices include the thermal operations (TO), their closure (CTO), and the Gibbs preserving covariant operations (GPC). Instead of giving a formal definition here, we refer the reader to \cite[Section II.C]{gour2022role}. In entanglement theory, we are interested in the distillation of EPR pairs from other states, possibly taking many copies. In the same spirit, in quantum thermodynamics, the corresponding task is the distillation of athermality. The asymptotic distillable athermality is quantified by the free energy difference in \cref{eq:rel-ent-free-energy} \cite{brandao2013resource, gour2022role}. Theorem V.1 of \cite{gour2022role} states that for the asymptotic distillation rate of athermality
\begin{equation}
    \mathrm{Distill}_{\mathfrak F}(\rho, \rho_\beta(H)) = D(\rho||\rho_\beta(H)) \, , 
\end{equation}
where $\mathfrak F \in \{\mathrm{TO}, \mathrm{CTO}, \mathrm{GPC}\}$. Again, we refer the reader to \cite{gour2022role} for the formal definitions. Thus, we can interpret \cref{cor:cor_uniform_continuity_relative_entropy_first_argument} as quantifying the continuity of distillable athermality. 

\begin{corollary}
    Let $H$ be a fixed Hamiltonian with maximal eigenvalue $\lambda_{\mathrm{max}}$, minimal eigenvalue $\lambda_{\mathrm{min}}$, and $\beta >0$ an inverse temperature. Then, for $\rho_1$, $\rho_2 \in \mathcal S(\mathcal H)$ such that $\frac{1}{2} \|\rho_1-\rho_2\|_1 \leq \epsilon \leq 1$, it holds that
    \begin{align}
        |\mathrm{Distill}_{\mathfrak F}(\rho_1, \rho_\beta(H))- \mathrm{Distill}_{\mathfrak F}(\rho_2, \rho_\beta(H))| &\le \varepsilon \left(\beta \lambda_{\mathrm{max}}+\log(\tr[e^{-\beta H}])\right) + (1 + \varepsilon) h \Big(\frac{\varepsilon}{1 + \varepsilon}\Big) \\
        &\le \varepsilon \beta (\lambda_{\mathrm{max}}-\lambda_{\mathrm{min}})+\varepsilon\log(d) + (1 + \varepsilon) h \Big(\frac{\varepsilon}{1 + \varepsilon}\Big)\, ,
    \end{align}
    where $\mathfrak F \in \{\mathrm{TO}, \mathrm{CTO}, \mathrm{GPC}\}$.
\end{corollary}

\subsection{Approximate Quantum Markov Chains}\label{sec:approxQMC}

In this section, we consider a tripartite Hilbert space $\mathcal{H}_{ABC} = \cH_A \otimes \cH_B \otimes \cH_C$ and  $\rho_{ABC} \in \cS_{+}(\cH_{ABC})$. We further consider the conditional mutual information of $\rho_{ABC}$ between $A$ and $C$ conditioned on $B$. The well-known property of strong subadditivity of the von Neumann entropy \cite{LiebRuskai-Subadditivity-1973} is equivalent to the non-negativity of the conditional mutual information, which is furthermore known \cite{Petz-MonotonicityRelativeEntropy-2003,HaydenJozsaPetzWinter-StrongSubadditivity-2004} to vanish if, and only if, 
\begin{equation}
    \rho_{ABC} = \rho_{AB}^{1/2} \rho_B^{-1/2} \rho_{BC}  \rho_B^{-1/2}  \rho_{AB}^{1/2} \, ,
\end{equation}
i.e., whenever $\rho_{ABC}$ is a quantum Markov chain. In particular,  if we denote $\mathcal{P}_{B \rightarrow AB}(\rho_{BC}) = \rho_{AB}^{1/2} \rho_B^{-1/2} \rho_{BC}  \rho_B^{-1/2}  \rho_{AB}^{1/2} $, we have
\begin{equation}
    I_{\mathcal{P}_{B \rightarrow AB}(\rho_{BC}) }(A:C |B) = 0 \, .
\end{equation}
Moreover, by the decomposition of the CMI of $\rho_{ABC}$ in terms of a difference of conditional entropies, as well as the data processing inequality, we have
\begin{equation}\label{eq:decomposition_CMI_conditional_entropies}
    I_\rho(A:C|B) = H_\rho(C|B) - H_\rho(C|AB) \leq H_{\mathcal{P}_{B \rightarrow AB}(\rho_{BC})} (C|AB) - H_\rho(C|AB) \, .
\end{equation}
Here we w.l.o.g. assumed that $d_A \le d_C$ using the symmetry of the CMI in $A$ and $C$. Therefore, we can apply our continuity bound for the CE from \cref{cor:continuity_bound_conditional_entropy} (which provides, in this case, a tighter result than \cref{cor:continuity_bound_CMI}), cf. also \cite{Winter-AlickiFannes-2016}, to obtain an upper bound on the CMI of $\rho_{ABC}$ in terms of how far it is from being recovered with the Petz recovery map, i.e., in terms of 
\begin{equation}
    \norm{\rho_{ABC} - \rho_{AB}^{1/2} \rho_B^{-1/2} \rho_{BC}  \rho_B^{-1/2}  \rho_{AB}^{1/2}}_1 \, .
\end{equation}
A similar direction was previously explored in \cite[Eq. (26)]{SutterRenner-ApproximateRecoverability-2018}. Note that, as a direct consequence of \cref{cor:continuity_bound_conditional_entropy}, we get the following bound for any state $\rho_{ABC}\in \cS(\cH_{ABC})$:
\begin{equation}
    I_\rho(A:C | B) \le 2 \varepsilon \log\min\{ d_A, d_C\} + (1 + \varepsilon) h \Big(\frac{\varepsilon}{1 + \varepsilon}\Big) \, ,
\end{equation}
with 
\begin{equation}
    \varepsilon:= \frac{1}{2}  \norm{\rho_{ABC} - \rho_{AB}^{1/2} \rho_B^{-1/2} \rho_{BC}  \rho_B^{-1/2}  \rho_{AB}^{1/2}}_1 \, .
\end{equation}
Moreover, we can use the following inequality
\begin{equation}
    (1+x) h \left(\frac{x}{1+x} \right) \leq \sqrt{2 x} \, ,
\end{equation}
for every $x \in [0,1]$, as well as the fact that, since $\varepsilon \in [0,1]$, then $\varepsilon \leq \sqrt{\varepsilon}$, to upper bound the CMI of $\rho_{ABC}$ by 
\begin{equation}\label{eq:bound_CMI}
     I_\rho(A:C | B) \leq   \left(\sqrt{2} \log\min\{d_A, d_C\} + 1  \right)  \norm{\rho_{ABC} - \rho_{AB}^{1/2} \rho_B^{-1/2} \rho_{BC}  \rho_B^{-1/2}  \rho_{AB}^{1/2}}_1^{1/2} \, .
\end{equation}
This bound should be compared to lower bounds for the conditional mutual information. On the one hand, Fawzi and Renner proved in \cite{FawziRenner-ConditionalMutualInformation-2015} the following lower bound for such a quantity in terms of the fidelity $F(\rho , \sigma) = \norm{\sqrt{\rho}\sqrt{\sigma}}_1^2$:
\begin{equation}
      I_\rho(A:C | B) \geq - \log F(\rho_{ABC}, \mathcal{R}_{B \rightarrow AB}(\rho_{BC}) ) \, ,
\end{equation}
where $\mathcal{R}_{B \rightarrow AB}$ is another recovery map, the so-called \textit{rotated Petz recovery map}, which was explicitly constructed in \cite{JungeRennerSutterWildeWinter-UniversalRecovery-2018}. Several results have been provided in this line in the past decade. Here we specifically focus on \cite{CarlenVershynina-Stability-DPI-RE-2017}, in which Carlen and Vershynina proved:
\begin{equation}\label{eq:carlen_vershynina_lower_bound_CMI}
    I_\rho(A:C | B) \geq \left( \frac{\pi}{8} \right)^4\norm{\rho_{B}^{-1}}_\infty^{-2} \norm{\rho_{ABC}^{-1}}_\infty^{-2}  \norm{\rho_{ABC} - \rho_{AB}^{1/2} \rho_B^{-1/2} \rho_{BC}  \rho_B^{-1/2}  \rho_{AB}^{1/2}}_1^4 \, ,
\end{equation}
Therefore, by combining \cref{eq:bound_CMI} with \cref{eq:carlen_vershynina_lower_bound_CMI} we obtain the following ``sandwich'' for the conditional mutual information of a tripartite density matrix $\rho_{ABC}$ in terms of its trace distance to its Petz recovery map:
\begin{equation}
    \begin{aligned}
        &  \left( \frac{\pi}{8} \right)^4 \norm{\rho_{B}^{-1}}_\infty^{-2} \norm{\rho_{ABC}^{-1}}_\infty^{-2}  \norm{\rho_{ABC} - \rho_{AB}^{1/2} \rho_B^{-1/2} \rho_{BC}  \rho_B^{-1/2}  \rho_{AB}^{1/2}}_1^4\\ 
        & \hspace{5cm}\leq I_\rho(A:C | B) \\
        & \hspace{5cm} \leq  2 \left(\log\min\{d_A, d_C\} + 1  \right)  \norm{\rho_{ABC} - \rho_{AB}^{1/2} \rho_B^{-1/2} \rho_{BC}  \rho_B^{-1/2}  \rho_{AB}^{1/2}}_1^{1/2}\, .
    \end{aligned}
\end{equation}
In particular, this implies that a state $\rho_{ABC} \in \cS(\cH_{A} \otimes \cH_B \otimes \cH_C)$ is an \textit{approximate quantum Markov chain} \cite{Sutter-ApproximateQMC-2018} (i.e. $ I_\rho(A:C | B) < \epsilon$) if, and only if, it is close to its reconstructed state under the Petz recovery map. This idea was used in \cite{Kato2019} to prove that a Gibbs state of a one-dimensional local Hamiltonian is an approximate quantum Markov chain, and subsequently, in \cite{hanson2020EEBmarkovian} to provide an estimate on the time it takes for a Markovian evolution of a density matrix to become an approximate quantum Markov chain. Moreover, a similar inequality has recently been employed in \cite{SvetlichnyyKennedy-DecayCMI-2022} to study the decay of the CMI for purely generated finitely correlated states.

\subsection{Difference between relative entropy and BS-entropy}\label{sec:difference_entropies}

It is well-known that the BS-entropy is an upper bound on the Umegaki relative entropy \cite{Hiai2017, matsumoto2010reverse, OhyaPetz-Entropy-1993}, i.e., that
\begin{equation}
    D(\rho||\sigma) \leq \widehat D(\rho||\sigma) \, ,
\end{equation}
and they coincide if and only if $\rho$ and $\sigma$ commute (see, e.g., \cite{hiai1991proper} and \cite[Proposition 4.7]{Hiai2017}). In this section, our aim is to quantify how large the difference between the two divergences can become. We start with two upper bounds on $\widehat D(\rho||\sigma)$ in terms of $D(\rho||\sigma)$.

\begin{proposition}\label{prop:additive_ineq_entropies}
Consider two positive definite states $\rho, \sigma \in \mathcal{S}_+(\mathcal{H})$. Then, the following inequality holds:
\begin{equation}\label{eq:additive_ineq_entropies}
\widehat{D} (\rho \| \sigma ) \leq D(\rho \| \sigma) + m^{-1} \norm{\rho - \sigma}_\infty \, ,
\end{equation}
where $m$ is the minimal eigenvalue of $\sigma$.
\end{proposition}
\begin{proof}
We can upper bound the difference between the entropies by
\begin{align}
    \widehat{D} (\rho \| \sigma ) - D(\rho \| \sigma) & = \tr[\rho \Big( \log (\rho^{1/2} \sigma^{-1} \rho^{1/2} ) - \log \rho + \log \sigma  \Big)] \\
    & = - D \left( \rho \Big\| \text{exp} \left\{ \log \sigma +  \log (\rho^{1/2} \sigma^{-1} \rho^{1/2} ) \right\} \right) \\
    & \leq \log \tr[ \text{exp} \left\{ \log \sigma +  \log (\rho^{1/2} \sigma^{-1} \rho^{1/2} ) \right\} ] \\
    & \leq \log \tr[ \sigma \rho^{1/2} \sigma^{-1} \rho^{1/2} ] \, ,
\end{align}
where we have used the non-negativity for the relative entropy of density matrices and Golden-Thompson inequality \cite{Golden-GTinequality-1965,Thompson-GTInequality-1965}. Next, we can write 
\begin{equation}
    \tr[ \sigma \rho^{1/2} \sigma^{-1} \rho^{1/2} ] = \tr[ \sigma \rho^{1/2} ( \sigma^{-1} -\rho^{-1} )\rho^{1/2} ] + 1 \, .
\end{equation}
Therefore, using $\log (x+1) \leq x$, we have 
\begin{equation}
     \widehat{D} (\rho \| \sigma ) - D(\rho \| \sigma) \leq \tr[ \sigma \rho^{1/2} ( \sigma^{-1} -\rho^{-1} )\rho^{1/2} ] \, .
\end{equation}
Now, we can use the following expression for invertible matrices $X$ and $Y$: 
\begin{equation}
    X^{-1} - Y^{-1} = Y^{-1} (Y-X) X^{-1} 
\end{equation}
Then,
\begin{align}
    \tr[ \sigma \rho^{1/2} ( \sigma^{-1} -\rho^{-1} )\rho^{1/2} ] & = \tr[ \sigma \rho^{-1/2} (  \rho - \sigma ) \sigma^{-1} \rho^{1/2} ] \\
    &\leq  \|\rho^{-1/2} (  \rho - \sigma ) \sigma^{-1} \rho^{1/2} \|_\infty\\
    & \leq \norm{\sigma^{-1}}_\infty \norm{\rho - \sigma}_\infty \, , 
\end{align}
by \cite[Proposition IX.1.1]{Bhatia1997} and Hölder's inequality. 
\end{proof}

The previous proposition provides a general upper bound for the distance between both entropies in terms of the spectral norm and the minimal eigenvalue of the second input. This is valid for any pair of states but does not yield any further information on specific pairs with better conditions. Alternatively, we can prove the following bound, from which it is obvious that $D(\rho||\sigma) = \widehat D(\rho||\sigma)$ if $\rho$ and $\sigma$ commute.  

\begin{proposition}\label{prop:additive_ineq_entropies2}
    Consider two positive definite states $\rho, \sigma \in \mathcal{S}_+(\mathcal{H})$. Then, the following inequality holds:
    \begin{equation}\label{eq:additive_ineq_entropies2}
        \widehat{D} (\rho \| \sigma ) \leq D(\rho \| \sigma) + f([\rho^{1/2}, \sigma^{-1/2}])  \, 
    \end{equation}
    where the last term is given by
    \begin{equation}
        f([\rho^{1/2}, \sigma^{-1/2}]) := \norm{\left[ \rho^{1/2} , \sigma^{-1/2} \right]}_\infty^2 + 2 \norm{\left[ \rho^{1/2} , \sigma^{-1/2} \right]}_\infty \, .
    \end{equation}
    In particular, whenever $\rho$ and $\sigma$ commute, $f$ vanishes.
\end{proposition}
\begin{proof}
    The proof proceeds in the same way as for \cref{prop:additive_ineq_entropies} until
    \begin{equation}
        \widehat{D} (\rho \| \sigma ) - D(\rho \| \sigma) \leq \tr[ \sigma \rho^{1/2} ( \sigma^{-1} -\rho^{-1} )\rho^{1/2} ] \, .
    \end{equation}
    Let us define now
    \begin{equation}
        \eta := \sigma^{1/2}  \rho^{1/2} \sigma^{-1} \rho^{1/2} \sigma^{1/2} \, .
    \end{equation}
    Then, 
    \begin{equation}
        \tr[ \sigma \rho^{1/2} ( \sigma^{-1} -\rho^{-1} )\rho^{1/2} ]  = \tr[\eta - \sigma] \, .
    \end{equation}
    Introducing $\rho$ gives
    \begin{equation}
        \tr[\eta - \sigma]  = \tr[\eta - \sigma + \rho - \rho] = \tr[\eta - \rho] + \tr[\rho - \sigma] = \tr[\eta - \rho] \le \norm{\eta - \rho}_1 \, .
    \end{equation}
    Moreover, as appears in \cite[Remark 2.2]{BluhmCapelPerezHernandez-WeakQFBSentropy-2021}, the right-hand side above can be estimated by 
    \begin{equation}
        \norm{\eta - \rho}_1 \leq \norm{\left[ \rho^{1/2} , \sigma^{-1/2} \right]}_\infty^2 + 2 \norm{\left[ \rho^{1/2} , \sigma^{-1/2} \right]}_\infty \, .
    \end{equation}
    This concludes the proof of the proposition. 
\end{proof}

Finally, we want to compare our previous bounds, proven using inequalities such as Golden-Thompson or Hölder, with those we could obtain by means of our continuity bounds, as the BS-entropy can, in particular, be regarded as a relative entropy. For that, we can also apply the continuity bound we derived in \cref{theo:theo_uniform_continuity_relative_entropy}.

\begin{corollary}
    Let $\rho \in \mathcal S(\mathcal H)$, $\sigma \in \mathcal S_+(\mathcal H)$ and $\tilde m$ such that $d_{\cH}^{-1} > 2\widetilde{m} > 0$ and the minimal eigenvalue of $\sigma$ is lower bounded by $2\widetilde{m}$. Let
    \begin{equation}
    \sigma^{-\frac{1}{2}} \rho \sigma^{-\frac{1}{2}} = \sum_{i=1}^k \lambda_i P_i
    \end{equation}
    be the spectral decomposition with eigenvalues $\lambda_i$ and projections $P_i$. Define density matrices
    \begin{equation}
    p = \sum_{i=1}^k\lambda_i \mathrm{tr}[\sigma P_i] \frac{P_i}{\mathrm{tr}[P_i]}, \qquad q = \sum_{i=1}^k \mathrm{tr}[\sigma P_i]\frac{P_i}{\mathrm{tr}[P_i]}.
    \end{equation}
    Then, for $\frac{1}{2}\norm{\rho - p} \le \varepsilon\le 1$ and $\frac{1}{2}\norm{\sigma - q}_1 \le \delta \le 1$, it holds that
    \begin{equation}
        |\widehat D(\rho \Vert \sigma) - D(\rho \Vert \sigma)| \le \Big(\varepsilon + \frac{\delta}{l_{\widetilde{m}}}\Big) \log\widetilde{m}^{-1} + (1 + \varepsilon)h\Big(\frac{\varepsilon}{1 + \varepsilon}\Big) + 2\frac{l_{\widetilde{m}} + \delta}{l_{\widetilde{m}}} f_{\widetilde{m}^{-1}, \widetilde{m}^{-1}}\Big(\frac{\delta}{l_{\widetilde{m}} + \delta} \Big) \, , \label{eq:BS-vs-Umegaki}
    \end{equation}
    with $l_{\widetilde{m}} = 1 - \widetilde{m}$. In particular, if $[\rho,\sigma] =0$, $\varepsilon$ and $\delta$ can be taken as $0$ such that the RHS of \cref{eq:BS-vs-Umegaki} is zero.
    
    Moreover, we can further simplify the previous bound to 
    \begin{equation}
        |\widehat D(\rho \Vert \sigma) - D(\rho \Vert \sigma)| \le (\sqrt{2} - \log \widetilde{m}) \sqrt{\varepsilon} + 3 \frac{\log\widetilde m^{-1}}{1 - \widetilde m} \delta  + 2 \log\left(1 + \frac{\delta}{1 - \widetilde m + \delta} \frac{1}{\widetilde m}\right)\, . \label{eq:BS-vs-Umegaki_2}
    \end{equation}
\end{corollary}
\begin{proof}
    Our argument is a slight variation of Matsumoto's minimal reverse test \cite{matsumoto2010reverse} (see also \cite{Hiai2017}). We can write the BS-entropy as the relative entropy of two commuting density matrices
    \begin{equation}
        \widehat D(\rho \Vert \sigma) = D(p \Vert q),
    \end{equation}
    since we can verify with $p_i = \lambda_i \mathrm{tr}[\sigma P_i]$, $q_i = \mathrm{tr}[\sigma P_i]$ that 
    \begin{align}
        D(p \Vert q) &= \sum_{i = 1}^k \tr[\frac{P_i}{\tr[P_i]} p_i\left(\log \frac{p_i}{\tr[P_i]} - \log \frac{q_i}{\tr[P_i]}\right)] \\
        &= \sum_{i = 1}^k p_i\left(\log p_i - \log q_i\right) \\
        &= \sum_{i = 1}^k \lambda_i \mathrm{tr}[\sigma P_i] \log \lambda_i \\
        &= \tr[\sigma \sigma^{-\frac{1}{2}} \rho \sigma^{-\frac{1}{2}} \log(\sigma^{-\frac{1}{2}} \rho \sigma^{-\frac{1}{2}})] \\
        &= \tr[ \rho \log(\rho^{\frac{1}{2}} \sigma^{-1} \rho^{\frac{1}{2}})].
    \end{align}
    Obviously, if $m$ is the minimal eigenvalue of $\sigma$, then $\frac{q_i}{\tr[P_i]} \geq m$ for all $i \in \{1, \ldots, k\}$. Thus, the assertion follows from \cref{theo:theo_uniform_continuity_relative_entropy}. Moreover, it is clear that if $[\rho, \sigma] = 0$ there is a unitary $U$ which diagonalizes $\rho$ and $\sigma$ simultaneously such that $\rho = p$ and $\sigma = q$.
    
    Finally, the last simplification from \cref{eq:BS-vs-Umegaki_2} is a direct consequence of \cref{theo:theo_uniform_continuity_relative_entropy} and \cref{lemma:nice_bounds}.
\end{proof}

\subsection{Weak quasi-factorization of the relative entropy}\label{sec:appl-qf-entropies}

Results of quasi-factorization for a divergence allow us to split such a divergence in a bipartite space in terms of the sum of two ``conditional'' divergences on subsystems and a multiplicative error term that is related to the correlations between both subsystems on the second input of the divergences. A weak version of such a result presents instead an additive error term.

More specifically, it was proven in \cite{capel2018quantum} that, given a bipartite space $\cH_{AB} = \cH_A \otimes \cH_B$  and $\rho_{AB}, \sigma_{AB} \in \cS (\cH_{AB})$, the following inequality holds:
\begin{equation}\label{eq:qf-relative-entropy}
    D(\rho_{AB} \| \sigma_{AB} ) \leq \frac{1}{1-2 \norm{h(\sigma_{AB})}_\infty} \left[ D_A (\rho_{AB} \| \sigma_{AB}) + D_B (\rho_{AB} \| \sigma_{AB})  \right] \, ,
\end{equation}
with 
\begin{equation}
    h(\sigma_{AB}) := \sigma_A^{-1/2} \otimes  \sigma_B^{-1/2}  \sigma_{AB} \sigma_A^{-1/2} \otimes  \sigma_B^{-1/2}  - \identity_{AB} \, ,
\end{equation}
and 
\begin{equation}
    D_X (\rho_{AB} \| \sigma_{AB}) := D (\rho_{AB} \| \sigma_{AB}) - D (\rho_{X^c} \| \sigma_{X^c})  \, , \quad \text{ for }X = A, B \, ,
\end{equation}
whenever $\norm{h(\sigma_{AB})}_\infty < 1/2$. Note that the term $\norm{h(\sigma_{AB})}_\infty $ provides a measure of how far $\sigma_{AB}$ is from being a tensor product between $A$ and $B$. This result, and subsequent extensions with additional conditions on $\sigma_{AB}$, are expected to find applications on various tasks in quantum information theory, and in particular, have proven to be essential for some recent proofs of positivity of \textit{modified logarithmic Sobolev inequalities} (MLSIs) for quantum Markov semigroups modelling thermal dissipative evolutions on quantum spin systems \cite{bardet2021entropy,bardet2021rapid, bardet2019modified,capel2020MLSI}. It is important to remark that \cref{eq:qf-relative-entropy} is equivalent to a generalization of the property of superadditivity of the relative entropy, as shown in  \cite{capel2017superadditivity}.

In \cite{BluhmCapelPerezHernandez-WeakQFBSentropy-2021}, some authors of the current manuscript tried to extend the previous result for the Umegaki relative entropy to the BS-entropy framework. However, we showed that the BS-entropy cannot satisfy a property of superadditivity, which makes it impossible to obtain a quasi-factorization for the BS-entropy in the spirit of  \cref{eq:qf-relative-entropy} without an additive error term. Instead, we proved a result of \textit{weak quasi-factorization}, from which we recovered \cref{eq:qf-relative-entropy}  if the marginals of $\rho_{AB}$ and $\sigma_{AB}$ commute. Here, we can prove another result along these lines as a consequence of our continuity bound for the relative entropy. Indeed, as a consequence of \cref{theo:theo_uniform_continuity_relative_entropy}, we obtain the following result of quasi-factorization for the relative entropy with an additive error term.

\begin{corollary}[Weak quasi-factorization for the relative entropy]\label{cor:quasi-factorization-relative-entropy}~\\
    Given $\rho_{AB}, \sigma_{AB} \in \cS(\cH_A \otimes \cH_B)$ such that $\text{ker}(\sigma_X) \subset \text{ker}(\rho_X) $ for $X=A,B,AB$, we have:
    \begin{equation}
         D(\rho_{AB} \| \sigma_{AB})  \leq  D_A(\rho_{AB} \| \sigma_{AB})  +  D_B(\rho_{AB} \| \sigma_{AB}) + \xi_{RE}(\rho_{AB} , \sigma_{AB})  \, ,
    \end{equation}
    with
    \begin{equation}
        \xi_{RE}(\rho_{AB} , \sigma_{AB}) := \left(\sqrt{2} - \log \widetilde{m}\right) \sqrt{\varepsilon} +  3 \frac{\log \widetilde m^{-1}}{l_{\widetilde m}}\delta + 2 \log\left(1 + \frac{\delta}{l_{\widetilde m} + \delta} \frac{1}{\widetilde m} \right) \, ,
    \end{equation}
    where  $\widetilde{m} = \frac{1}{2}\min \left\{ \norm{\sigma_A^{-1} \otimes \sigma_B^{-1}}_\infty^{-1} , \norm{\sigma_{AB}^{-1} }_\infty^{-1}  \right\}$, $\varepsilon = \varepsilon(\rho_{AB}) = \frac{1}{2}\norm{\rho_{AB} - \rho_A \otimes \rho_B}_1$ and $\delta = \delta(\sigma_{AB}) = \frac{1}{2}\norm{\sigma_{AB} - \sigma_A \otimes \sigma_B}_1$.
\end{corollary}
\begin{proof}
    The difference between the relative and the two conditional entropies can be written as
    \begin{equation}
        \begin{aligned}
             D(\rho_{AB} \| \sigma_{AB}) & -  D_A(\rho_{AB} \| \sigma_{AB})  -  D_B(\rho_{AB} \| \sigma_{AB}) = - D(\rho_{AB} \| \sigma_{AB}) + D(\rho_{A} \otimes \rho_B  \| \sigma_{A} \otimes \sigma_B ) \, .
        \end{aligned}
    \end{equation}
    Therefore, we can apply \cref{theo:theo_uniform_continuity_relative_entropy} to obtain a continuity bound for the difference between the last two relative entropies, obtaining 
    \begin{equation}
        \begin{aligned}
            &  |D(\rho_{AB} \Vert \sigma_{AB}) - D(\rho_{A} \otimes \rho_B  \| \sigma_{A} \otimes \sigma_B)| \\
            & \hspace{2cm} \le \Big(\varepsilon + \frac{\delta}{l_{\widetilde{m}}}\Big) \log(\widetilde{m}^{-1}) + (1 + \varepsilon)h\Big(\frac{\varepsilon}{1 + \varepsilon}\Big) + 2\frac{l_{\widetilde{m}} + \delta}{l_{\widetilde{m}}} f_{\widetilde{m}^{-1}, \widetilde{m}^{-1}}\Big(\frac{\delta}{l_{\widetilde{m}} + \delta} \Big) \, , 
        \end{aligned}
    \end{equation}
    with
    \begin{equation}
        \varepsilon := \frac{1}{2} \norm{\rho_{AB} - \rho_{A} \otimes \rho_B}_1 \; , \quad  \delta := \frac{1}{2} \norm{\sigma_{AB} - \sigma_{A} \otimes \sigma_B}_1 \, , 
    \end{equation}
    and $l_{\widetilde{m}} = 1 - \widetilde{m}$, for $\widetilde{m} = \frac{1}{2}\text{min} \left\{ \norm{\sigma_A^{-1} \otimes \sigma_B^{-1}}_\infty^{-1} , \norm{\sigma_{AB}^{-1} }_\infty^{-1}  \right\} $. Moreover, we can apply the simplification of \cref{theo:theo_uniform_continuity_relative_entropy} using \cref{lemma:nice_bounds}. We then have
    \begin{equation}
        \begin{aligned}
            |D(\rho_{AB} \Vert \sigma_{AB}) - D(\rho_{A} \otimes \rho_B  \| \sigma_{A} \otimes \sigma_B)| \leq \left(\sqrt{2} - \log\widetilde{m}\right) \sqrt{\varepsilon} +  3 \frac{\log \widetilde m^{-1}}{l_{\widetilde m}}\delta + 2 \log\left(1 + \frac{\delta}{l_{\widetilde m} + \delta}\frac{1}{\widetilde m}\right) \, ,
        \end{aligned}
    \end{equation}
    concluding thus the proof.
\end{proof}

Note that, even though there is a caveat in this result in the form of an additive error term, which prevents it from being useful to prove the positivity of MLSIs, it presents the advantage with respect to \cref{eq:qf-relative-entropy} that there is no multiplicative error term in this case, which might be of more interest for some other contexts, such as for entropy accumulation \cite{metger2022entropyaccumulation} or in the line of the applications given by the Brascamp-Lieb dualities \cite{berta2019brascamplieb}. 

\subsection{Minimal distance to separable states} \label{sec:min-dist-sep-states}

In this section, we show how to reprove the continuity bounds for the relative entropy of entanglement in \cite{Winter-AlickiFannes-2016} from the ALAFF method and how this strategy generalizes if we quantify the minimal distance to the set of separable states in terms of the BS-entropy instead.

Let $\mathcal C \subset \mathcal S(\mathcal H)$ be a compact convex subset of the set of quantum states with at least one positive definite state. We can define the minimal distance to $\mathcal C$ in terms of the relative entropy as
\begin{equation}
    D_{\mathcal C}(\rho):= \inf_{\gamma \in \mathcal C} D(\rho \Vert \gamma).
\end{equation}
As explained in \cite{Winter-AlickiFannes-2016}, the fact that $\mathcal C$ contains a positive definite state guarantees that $D_{\mathcal C}(\rho) < \infty$ for all $\rho \in \mathcal S(\mathcal H)$. Moreover, the infimum is attained, as follows from the fact that the relative entropy is lower semi-continuous \cite{OhyaPetz-Entropy-1993} and Weierstrass’ theorem on extreme values of such functions \cite[Theorem 2.43]{aliprantisborder}. Examples of $\mathcal C$ include $\mathrm{SEP}_{AB}$, the set of separable states for systems $A$, $B$, and
\begin{equation}
    \{d_A^{-1} \mathds{1}_A \otimes \sigma_B: \sigma_B \in \mathcal S(\mathcal H_B)\},
\end{equation}
which yields $D_{\mathcal C}(\rho_{AB}) = - H_\rho(A|B) + \log{d_A}$. The quantity $D_{\mathrm{SEP}_{AB}}$ is known as the relative entropy of entanglement \cite{VedralEntanglement, VedralQuantifying}. It constitutes a tight upper bound on the distillable entanglement \cite{RainsBound, VedralQuantifying}. This is the quantity we focus on for now. 

\begin{lemma} \label{lem:rel-ent-entangle-convex}
    Let $\mathcal C \subset \mathcal S(\mathcal H)$ be a compact convex set containing at least one positive definite state. Then, $D_{\mathcal C}$ is convex on $\mathcal S(\mathcal H)$.
\end{lemma}
\begin{proof}
    This follows directly from the joint convexity of the relative entropy.
\end{proof}

In order to apply the ALAFF method, we need to prove almost concavity next.

\begin{lemma} \label{lem:rel-ent-entangle-al-concave}
    Let $\mathcal C \subset \mathcal S(\mathcal H)$ be a compact convex set containing at least one positive definite state.  Moreover, let $\rho_1, \rho_2 \in \mathcal S(\mathcal H)$ and $p \in [0,1]$. Then,
    \begin{equation}
        D_{\mathcal C}(p\rho_1+(1-p)\rho_2) \geq p D_{\mathcal C}(\rho_1) + (1-p) D_{\mathcal C}(\rho_2) - h(p).
    \end{equation}
\end{lemma}
\begin{proof}
    We can use the almost concavity of the relative entropy. Let $\tau$ the state that achieves the infimum in $D_{\mathcal C}(p\rho_1+(1-p)\rho_2)$. By \cref{theo:theo_almost_concavity_relative_entropy} and point 1 of \cref{prop:prop_almost_concave_estimate_relative_entropy_well_behaved}, we obtain that
    \begin{align}
        D_{\mathcal C}(p\rho_1+(1-p)\rho_2) &\geq p D(\rho_1 \Vert \tau) + (1-p) D(\rho_2 \Vert \tau) - h(p)\\
        &\geq p D_{\mathcal C}(\rho_1) + (1-p) D_{\mathcal C}(\rho_2) - h(p)\,,
    \end{align}
    which is the assertion.
\end{proof}

Finally, we need the following estimate:

\begin{lemma}\label{lem:rel-ent-entangle-diff-bound} 
    Let $\mathcal H = \mathcal H_A \otimes \mathcal H_B$. It holds that
    \begin{equation}
           \sup_{\substack{\rho, \sigma \in \mathcal S(\mathcal H)\\ \frac{1}{2}\|\rho- \sigma\|_1 = 1}} |D_{\mathrm{SEP}_{AB}}(\rho) - D_{\mathrm{SEP}_{AB}}(\sigma)| \leq \log \min\{d_A,d_B\}.
    \end{equation}
\end{lemma} 
\begin{proof}
    Without loss of generality, let $d_A \leq d_B$. For a pure state $\ket{\psi}$ with Schmidt decomposition $\sum_{i=1}^{d_A} \lambda_i \ket{i_A} \otimes \ket{i_B}$, let 
    \begin{equation}
        \tau_{\psi} = \frac{1}{d_A} \sum_{i = 1}^{d_A} \dyad{i_A} \otimes \dyad{i_B}. 
    \end{equation} 
    This state is manifestly separable. Then,
    \begin{align}
        \sup_{\substack{\rho, \sigma \in \mathcal S(\mathcal H)\\ \frac{1}{2}\|\rho- \sigma\|_1 = 1}} |D_{\mathrm{SEP}_{AB}}(\rho) - D_{\mathrm{SEP}_{AB}}(\sigma)| 
        & \leq \sup_{\dyad{\psi} \in \mathcal S(\mathcal H)} D(\dyad{\psi}\Vert \tau_{\psi}) \\
        &= \log{d_A} \, .
    \end{align}
    In the first inequality, we have used that $D_{\mathrm{SEP}_{AB}}$ is positive and convex.
\end{proof}

This allows us to prove via the ALAFF method a continuity bound for the relative entropy of entanglement:

\begin{theorem} \label{thm:rel-ent-entangle-cont-bound}
    For $\varepsilon \in [0,1]$ and $\mathcal H = \mathcal H_A \otimes \mathcal H_B$, it holds that for $\rho, \sigma \in \cS(\cH)$ with $\frac{1}{2} \norm{\rho - \sigma}_1 \le \varepsilon$
    \begin{equation}
             |D_{\mathrm{SEP}_{AB}}(\rho) - D_{\mathrm{SEP}_{AB}}(\sigma)| \leq \varepsilon \log \min\{d_A,d_B\} + (1+\varepsilon) h \left(\frac{\varepsilon}{1+\varepsilon} \right) \, .
    \end{equation}
\end{theorem}
\begin{proof}
    This follows from \cref{rem:rem_s=0_alaff_method}, using \cref{lem:rel-ent-entangle-convex}, \cref{lem:rel-ent-entangle-al-concave}, point 4 of \cref{prop:prop_almost_concave_estimate_relative_entropy_well_behaved}, and \cref{lem:rel-ent-entangle-diff-bound}.
\end{proof}

\cref{thm:rel-ent-entangle-cont-bound} recovers the bound \cite[Corollary 8]{Winter-AlickiFannes-2016}, proven with very similar methods, which improved over the earlier bound in \cite{DonaldHorodecki1999continuityREentanglement}. The interest of executing the proof here is that a similar strategy will give us bounds on a BS-entropy version of the relative entropy of entanglement, as we will show now. We define
\begin{equation}
    \widehat D_{\mathcal C}(\rho) = \inf_{\gamma \in \mathcal C} \widehat D(\rho \Vert \gamma)\, ,
\end{equation}
which measures how far $\rho$ is from $\mathcal C$ in terms of the BS-entropy. The infimum is attained as the BS-entropy is also lower semi-continuous \cite[Section 10]{Matsumoto2018}. Convexity follows again from the joint-convexity of the BS-entropy.

\begin{lemma} \label{lem:BS-ent-entangle-convex}
    Let $\mathcal C \subset \mathcal S(\mathcal H)$ be a compact convex set containing at least one positive definite state. Then, $\widehat D_{\mathcal C}$ is convex on $\mathcal S(\mathcal H)$.
\end{lemma}

Almost concavity requires more work in this case.

\begin{lemma} \label{lem:BS-ent-entangle-al-concave}
    Let $\mathcal C \subset \mathcal S(\mathcal H)$ be a compact convex set containing the maximally mixed state. Moreover, let $\rho_1, \rho_2 \in \mathcal S(\mathcal H)$, $p \in [0,1)$, and $d \in \mathbb N$, $d \geq 2$ the dimension of $\mathcal H$. Then,
    \begin{equation}
        \widehat D_{\mathcal C}(p\rho_1+(1-p)\rho_2 ) \geq p \widehat D_{\mathcal C}(\rho_1) + (1-p) \widehat D_{\mathcal C}(\rho_2) - g_d(p).
    \end{equation}
    Here, $g_d(p) := \frac{d}{p^{1/d}}h(p) - \log(1 - p^{1/d})$ for $p \in (0,1)$ and $g_d(0) := 0$.
\end{lemma}
\begin{proof}
    In order to apply the almost concavity of the BS-entropy, we need to control the minimal eigenvalue of $\tau$, the best approximation of $\rho=p\rho_1 + (1 - p)\rho_2$ in $\mathcal C$. To this end, we will use a strategy inspired by \cite{DonaldHorodecki1999continuityREentanglement}. Let $\mathcal \tau_s$ be the state achieving the infimum in 
    \begin{equation}
        \inf_{\tau \in \mathcal C} \widehat D \Big(\rho \,  \Big\Vert \, s \tau + (1-s) \frac{\mathds{1}}{d} \Big)
    \end{equation}
    for some $s \in (0,1)$ which we will specify later. Clearly, 
    \begin{equation}
           \widehat D_{\mathcal C}(\rho) \leq  \widehat D \Big(\rho \,  \Big\Vert \, s \tau_s + (1-s) \frac{\mathds{1}}{d} \Big).
    \end{equation}
    Furthermore, with $\hat \tau$ a state such that  $\widehat D_{\mathcal C}(\rho) = \widehat D(\rho \Vert \hat \tau)$,
    \begin{align}
        \widehat D \Big(\rho \, \Big\Vert  \, s \tau_s + (1-s) \frac{\mathds{1}}{d} \Big) &\leq \widehat D \Big(\rho \, \Big\Vert \, s \hat \tau + (1-s) \frac{\mathds{1}}{d} \Big) \\
        & \leq \widehat D_{\mathcal C}(\rho) - \log{s}\, ,
    \end{align}
    as $s\hat \tau + (1-s) \frac{\mathds{1}}{d} \geq s \hat \tau$ and the logarithm is operator monotone. Since $\widehat D_{\mathcal C}(\rho) < \infty$ we have  $\ker \hat \tau \subseteq \ker \rho$, thus, we can restrict $\hat \tau$ to the support of $\rho$, where $\hat \tau$ is positive definite. Combining this bound with  \cref{theo:theo_almost_concavity_bs_entropy}, we infer
    \begin{align}
        \widehat D_{\mathcal C}(p\rho_1 + (1-p) \rho_2) &\geq   \widehat D\Big( p\rho_1 + (1-p) \rho_2 \Big\Vert s \tau_s + (1-s) \frac{\mathds{1}}{d}\Big) + \log{s}\\
        &\geq p \widehat D_{\mathcal C}(\rho_1) + (1-p) \widehat D_{\mathcal C}(\rho_2) - \frac{d}{1-s}h(p) + \log{s}.
    \end{align}
    Here, we have used point 1 of \cref{prop:prop_almost_concave_estimate_bs_entropy_well_behaved}. Finally, we have to choose $s$ such that $\frac{d}{1-s}h(p) - \log{s}$ goes to zero for $p \to 0^+$ and is non-decreasing on $p \in [0,1/2]$. It turns out that $s=1-p^{1/d}$ is a convenient choice, see \cref{lem:gd-cont} and \cref{lem:gd-mono}.
\end{proof}

\begin{remark}
    Note that we could have substituted $g_d$ in \cref{lem:BS-ent-entangle-al-concave} by a symmetrized version
    \begin{equation}
        \tilde g_d(p) := \begin{cases} g_d(p) & p \in [0,1/2] \\ g_d(1-p) & p \in [1/2,1] \end{cases}
    \end{equation}
    in order to obtain
        \begin{equation}
        \widehat D_{\mathcal C}(p\rho_1+(1-p)\rho_2 ) \geq p \widehat D_{\mathcal C}(\rho_1) + (1-p) \widehat D_{\mathcal C}(\rho_2) - \tilde g_d(p) 
    \end{equation}
    for all $p \in [0,1]$ and $\tilde g_d(0) = \tilde g_d(1) = 0$. For the ALAFF method with $s=0$, however, it is only relevant what happens on $[0,1/2]$.
\end{remark}

The final estimate we need in order to apply the ALAFF method is proven in a very similar way as \cref{lem:rel-ent-entangle-diff-bound}.

\begin{lemma}\label{lem:BS-ent-entangle-diff-bound} 
    Let $\mathcal H = \mathcal H_A \otimes \mathcal H_B$. It holds that
    \begin{equation}
        \sup_{\substack{\rho, \sigma \in \mathcal S(\mathcal H)\\ \frac{1}{2}\|\rho- \sigma\|_1 =1}} |\widehat D_{\mathrm{SEP}_{AB}}(\rho) - \widehat D_{\mathrm{SEP}_{AB}}(\sigma)| \leq \log \min\{d_A,d_B\}.
    \end{equation}
\end{lemma} 
\begin{proof}
    Without loss of generality, let $d_A \leq d_B$. For a pure state $\ket{\psi}$ with Schmidt decomposition $\sum_{i=1}^{d_A} \lambda_i \ket{i_A} \otimes \ket{i_B}$, let again
    \begin{equation}
        \tau_{\psi} = \frac{1}{d_A} \sum_{i = 1}^{d_A} \dyad{i_A} \otimes \dyad{i_B}\, , 
    \end{equation} 
    which is a separable state. Then,
    \begin{align}
        \sup_{\substack{\rho, \sigma \in \mathcal S(\mathcal H)\\ \frac{1}{2}\|\rho- \sigma\|_1 =1}} |\widehat D_{\mathrm{SEP}_{AB}}(\rho) - \widehat D_{\mathrm{SEP}_{AB}}(\sigma)| 
        & \leq \sup_{\dyad{\psi} \in \mathcal S(\mathcal H)} \widehat D(\dyad{\psi}\Vert \tau_{\psi}) \\
        &= \log{d_A} \, .
    \end{align}
    In the above inequality, we have used that $\widehat D_{\mathrm{SEP}_{AB}}$ is positive and convex. Note that $\ket{\psi}$ is in the support of $\tau_\psi$.
\end{proof}

\begin{theorem} \label{thm:BS-ent-entangle-cont-bound}
    For $\varepsilon \in [0,1]$, $\mathcal H = \mathcal H_A \otimes \mathcal H_B$, and $d_{AB} \in \mathbb N$, $d_{AB} \geq 2$, it holds that for $\rho, \sigma \in \cS(\cH)$ with $\frac{1}{2} \norm{\rho - \sigma}_1 \le \varepsilon$
    \begin{equation}
        |\widehat D_{\mathrm{SEP}_{AB}}(\rho) - \widehat D_{\mathrm{SEP}_{AB}}(\sigma)| \leq \varepsilon \log \min\{d_A,d_B\} + (1+\varepsilon) g_{d_{AB}} \left(\frac{\varepsilon}{1+\varepsilon} \right) \, .
    \end{equation}
     Here, $g_d(p) := \frac{d}{p^{1/d}}h(p) - \log(1 - p^{1/d})$ for $p \in (0,1)$ and $g_d(0) = 0$.
\end{theorem}
\begin{proof}
    As shown in \cref{lem:gd/(1-p)-mono}, it holds that $g_d(p)/(1-p)$ is non-decreasing on $[0,1]$ for all $d \in \mathbb N$, $d \geq 2$. Thus, the assertion follows from \cref{rem:rem_s=0_alaff_method} using \cref{lem:BS-ent-entangle-convex}, \cref{lem:BS-ent-entangle-al-concave} with \cref{lem:gd-cont} and  \cref{lem:gd-mono}, and \cref{lem:BS-ent-entangle-diff-bound}.
\end{proof}

To end this section, let us investigate the choice
\begin{equation}
    \mathcal C_0 := \{d_A^{-1} \mathds{1}_A \otimes \sigma_B: \sigma_B \in \mathcal S(\mathcal H_B)\}.
\end{equation}
From the discussion after \cref{eq:variational-cond-BS}, we know that 
\begin{equation}
    \widehat{H}_\rho(A|B) \le \underset{\sigma_B \in \cS(\cH_B)}{\sup} -\widehat{D}(\rho_{AB} \Vert \identity_A \otimes \sigma_B)=: \widehat{H}^{\mathrm{var}}_\rho(A|B)   \,,
\end{equation}
but equality does not hold in general. This is different from the Umegaki relative entropy, where the conditional entropy coincides with its variational expression. Nonetheless, we obtain a continuity bound for $\widehat{H}^{\mathrm{var}}_\rho(A|B)$ from the approach in this section.

\begin{corollary}\label{cor:variational_conditional_BS_entropy_continuous}
    Let $\mathcal H = \mathcal H_A \otimes \mathcal H_B$. For $\varepsilon \in [0,1]$ and $d_{AB} \in \mathbb N$, $d_{AB} \geq 2$, it holds that for $\rho, \sigma \in \cS(\cH)$ with $\frac{1}{2} \norm{\rho - \sigma}_1 \le \varepsilon$
    \begin{equation}
        |\widehat{H}^{\mathrm{var}}_\rho(A|B) - \widehat{H}^{\mathrm{var}}_\sigma(A|B)| \leq 2\varepsilon \log d_A + (1+\varepsilon) g_{d_{AB}} \left(\frac{\varepsilon}{1+\varepsilon} \right) \, .
    \end{equation}
     Here, $g_d(p) := \frac{d}{p^{1/d}}h(p) - \log(1 - p^{1/d})$ for $p \in (0,1)$ and $g_d(0) = 0$.
\end{corollary}
\begin{proof}
    It holds that for $\rho, \sigma \in \cS(\cH)$ with $\frac{1}{2}\norm{\rho - \sigma}_1 \le \varepsilon$
    \begin{equation}
        |\widehat{H}^{\mathrm{var}}_\rho(A|B) - \widehat{H}^{\mathrm{var}}_\sigma(A|B)|  =  |\widehat D_{\mathcal C_0}(\rho) - \widehat D_{\mathcal C_0}(\sigma)|\,,
    \end{equation}
    since the normalization does not matter. Thus to apply ALAFF, we need to bound
    \begin{equation}
        \sup_{\substack{\rho, \sigma \in \mathcal S(\mathcal H)\\ \frac{1}{2}\|\rho- \sigma\|_1 =1}} |\widehat D_{\mathcal C_0}(\rho) - \widehat D_{\mathcal C_0}(\sigma)|\,.
    \end{equation}
    Using \cref{eq:bounds-BS-cond-ent} and the fact that $\widehat D_{\mathcal C_0}(\rho) \geq 0$ for all states $\rho$, we obtain
    \begin{align}
        \sup_{\substack{\rho, \sigma \in \mathcal S(\mathcal H)\\ \frac{1}{2}\|\rho- \sigma\|_1 =1}} |\widehat D_{\mathcal C_0}(\rho) - \widehat D_{\mathcal C_0}(\sigma)|
        &\leq  \sup_{\rho \in \mathcal S(\mathcal H)} -\widehat{H}^{\mathrm{var}}_\rho(A|B) + \log{d_A} \\
        &\leq 2\log{d_A} \,.
    \end{align}
    The assertion follows from combining the above with \cref{lem:BS-ent-entangle-convex}, \cref{lem:BS-ent-entangle-al-concave} with \cref{lem:gd-cont} and  \cref{lem:gd-mono}, and \cref{lem:gd/(1-p)-mono} to apply \cref{rem:rem_s=0_alaff_method}.
\end{proof}

\begin{remark}\label{rem:both_conditional_BS_entropies_are_different}
    Note that the findings of \cref{cor:variational_conditional_BS_entropy_continuous} and \cref{prop:discontinuity_conditional_BS_entropy} provide a formal proof that $\widehat{H}_\rho(A|B) $ and $\widehat{H}^{\mathrm{var}}_\rho(A|B)$ are different in general. Indeed, while we have just shown that the latter quantity is continuous on $\mathcal{S}(\mathcal{H})$ as a consequence of the results of this section, in \cref{prop:discontinuity_conditional_BS_entropy} we showed that the former quantity is in general discontinuous on $\mathcal{S}(\mathcal{H})$.
\end{remark}

\subsection{Rains information}\label{sec:Rains_info}

Inspired by the Rains bound from entanglement theory \cite{RainsSDP2001}, for any divergence $\mathbb{D}$, the \textit{generalized Rains bound} of a quantum state $\rho_{AB} \in \cS(\cH_A \otimes \cH_B)$ was defined in \cite{Tomamichel2017StrongConverses} by
\begin{equation}\label{eq:generalized_Rains_bound}
    \mathbb{R} (\rho_{AB}) := \underset{\sigma_{AB} \in \mathrm{PPT'}(A:B)}{\text{min}} \, \mathbb{D}(\rho_{AB} \| \sigma_{AB}) \, ,
\end{equation}
where the minimization is taken over the Rains set
\begin{equation}\label{eq:Rains_set}
    \mathrm{PPT'}(A:B) := \left\{ \sigma_{AB} \, : \, \sigma_{AB} \geq 0 , \, \norm{\sigma_{AB}^{T_B}}_1 \leq 1 \right\} \, .
\end{equation}
Where $\cdot^{T_B}$ denotes the partial transpose in the $B$-system. This definition can be easily extended to channels in the following way. For a quantum channel $T_{A' \rightarrow B}:\cS(\cH_{A} \otimes \cH_{A'}) \rightarrow \cS(\cH_{A} \otimes \cH_{B}) $, we define 
\begin{equation}
    \mathbb{R}( T ) := \underset{\rho_A \in \cS(\cH_A)}{\text{max}} \mathbb{R} (T_{A'\rightarrow B} (\phi_{AA'})) \, ,
\end{equation}
for $\phi_{AA'}$ a purification of $\rho_A$. In particular, for the Umegaki relative entropy, we introduce the \textit{Rains information} as
\begin{equation}
    R(T) :=  \underset{\rho_A \in \cS(\cH_A)}{\text{max}} \, \underset{\sigma_{AB} \in \mathrm{PPT'}(A:B)}{\text{min}}\, D(T_{A'\rightarrow B} (\phi_{AA'}) \| \sigma_{AB}) \, , 
\end{equation}
as well as  the BS-Rains information by
\begin{equation}
    \widehat{R}(T) :=  \underset{\rho_A \in \cS(\cH_A)}{\text{max}} \, \underset{\sigma_{AB} \in \mathrm{PPT'}(A:B)}{\text{min}}\, \widehat D(T_{A'\rightarrow B} (\phi_{AA'}) \| \sigma_{AB}) \, . 
\end{equation}
In the rest of the subsection, we will drop the subindex from the channels whenever it is clear in which systems they act. In \cite{FangFawzi-GeometricRenyiDivergences-2019}, it was proven that the latter two quantities constitute upper bounds to the quantum capacity of a quantum channel. Indeed, the following inequality holds for any channel $T$:
\begin{equation}
    Q(T) \leq R(T) \leq \widehat R(T) \, .
\end{equation}
Moreover, the BS-Rains information is a limit of Rains informations induced by $\alpha$-geometric Rényi divergences, which can be written as single-letter formulas and computed via a semidefinite program (SDP), as shown in \cite{FangFawzi-GeometricRenyiDivergences-2019}. The study of these quantities is therefore of great interest for application in the context of strong converses of quantum capacities of channels. 

Here, as a consequence of \cref{cor:cor_uniform_continuity_relative_entropy_first_argument} and \cref{cor:cor_uniform_continuity_BS_entropy_first_argument}, respectively, we can provide continuity results for both the Rains information and the BS-Rains information, respectively, following the lines of  \cref{thm:rel-ent-entangle-cont-bound}. Beforehand, we need to justify that both quantities are well-defined, i.e., that each of these quantities is attained at a certain $\rho_A \in \cS(\cH_A)$ and $\sigma_{AB}\in \mathrm{PPT'}(A:B)$, and thus the minimum and maximum in their definitions are properly written. For that, note that we are first taking an infimum on the second input over the compact set $\mathrm{PPT'}(A:B)$. Then, the infimum is attained and the expression obtained is a continuous function, as we will show below in \cref{eq:PPT'-cont}. Next, we perform an optimization problem on the first input over another compact set, namely $\cS(\cH_A)$. Thus, that supremum is also attained and both Rains informations are well defined. 

From now on, for simplicity and for similarity with the quantities introduced in the previous section, given $\rho_{AB} \in \cS(\cH_A \otimes \cH_B)$, let us define
\begin{equation}
    D_{\mathrm{PPT'}(A:B)}(\rho_{AB}) := \underset{\sigma_{AB} \in \mathrm{PPT'}(A:B)}{\text{min}} D ( \rho_{AB} \| \sigma_{AB}) \, .
\end{equation}
Then, it is clear that we can rewrite, for a quantum channel $T :\cS(\cH_{A} \otimes \cH_{A'}) \rightarrow \cS(\cH_{A} \otimes \cH_{B})$,
\begin{equation}
    R(T) := \underset{\rho_A \in \cS(\cH_A)}{\text{max}} D_{\mathrm{PPT'}(A:B)}(T(\phi_{AA'}) ) \, ,
\end{equation}
for $\phi_{AA'}$ a purification of $\rho_A$. The next step before applying the ALAFF method is bounding the difference between two Rains informations of two quantum channels. For that, we will use the $1 \rightarrow 1$ norm of the difference between channels. Let us recall that for $T :\cS(\cH_{A} \otimes \cH_{A'}) \rightarrow \cS(\cH_{A} \otimes \cH_{B})$ a quantum channel, its $1 \rightarrow 1$ norm is given by
\begin{equation}
    \norm{T}_{1 \rightarrow 1} := \underset{\eta : \norm{\eta}_1 \leq 1}{\text{max}} \norm{ T (\eta)}_1 \, .
\end{equation}
For $T_{A^\prime \to B}$, the $1 \rightarrow 1$ norm coincides with the diamond norm. Now, as a consequence of \cref{lem:rel-ent-entangle-diff-bound} and \cref{thm:rel-ent-entangle-cont-bound} from the previous section, we can derive the following continuity bound for the Rains information.

\begin{theorem}\label{theo:continuity_bound_rains_information}
    For $\varepsilon \in [0,1]$  and $ T^1_{A^\prime \to B}, T^2_{A^\prime \to B} :\cS(\cH_{A} \otimes \cH_{A'}) \rightarrow \cS(\cH_{A} \otimes \cH_{B}) $ two quantum channels with $\frac{1}{2}\|T^1_{A^\prime \to B}- T^2_{A^\prime \to B}\|_{1 \rightarrow 1} \leq \varepsilon$, we have:
    \begin{equation}\label{eq:continuity_bound_rains_information}
        |R(T_{A^\prime \to B}^1) - R(T_{A^\prime \to B}^2)| \leq  \varepsilon \log \min \{ d_A , d_B\} + (1+ \varepsilon) h \Big( \frac{\varepsilon}{1 + \varepsilon} \Big) \, . 
    \end{equation}
\end{theorem}
\begin{proof}
    Let us drop the subscripts from the channels for ease of notation. Firstly, note that $\mathrm{SEP}_{AB} \subset \mathrm{PPT'}(A:B)$. Therefore, 
    \begin{equation}
         R(T) = \underset{\rho_A \in \cS(\cH_A)}{\text{max}} D_{\mathrm{PPT'}(A:B)}(T(\phi_{AA'}) ) \leq \underset{\rho_A \in \cS(\cH_A)}{\text{max}} D_{\mathrm{SEP}_{AB}}(T(\phi_{AA'}) ) \, .
    \end{equation}
    Hence, in general
    \begin{equation}\label{eq:PPT'-cont}
       \begin{aligned}
            \max_{\substack{\rho_{AB} ,  \, \sigma_{AB} \in \mathcal S(\mathcal H_{AB}) \\ \frac{1}{2}\|\rho_{AB}- \sigma_{AB}\|_1 =1}} |D_{\mathrm{PPT'}(A:B)}(\rho_{AB}) - D_{\mathrm{PPT'}(A:B)}(\sigma_{AB})|&  \leq  \max_{\rho_{AB} \in \mathcal S(\mathcal H_{AB}) } D_{\mathrm{PPT'}(A:B)}(\rho_{AB}) \\
            & \leq  \max_{\rho_{AB} \in \mathcal S(\mathcal H_{AB}) } D_{\mathrm{SEP}_{AB}}(\rho_{AB}) \\[2mm]
            & \leq \log \min \{ d_A , \, d_B \} \, ,
       \end{aligned}
    \end{equation} 
    where in the last inequality we have used \cref{lem:rel-ent-entangle-diff-bound}. Following the lines of \cref{thm:rel-ent-entangle-cont-bound}, we have for $\rho_{AB}, \sigma_{AB} \in \cS(\cH_A \otimes \cH_B)$ with $\frac{1}{2}\|\rho_{AB}- \sigma_{AB}\|_1 \leq \varepsilon $ the following continuity bound:
    \begin{equation}
        \begin{aligned}
            |D_{\mathrm{PPT'}(A:B)}(\rho_{AB}) - D_{\mathrm{PPT'}(A:B)}(\sigma_{AB})|&  \leq \varepsilon \log \min \{ d_A , \, d_B \}  + (1+ \varepsilon) h \Big( \frac{\varepsilon}{1+ \varepsilon} \Big) \, . 
        \end{aligned}
    \end{equation}
    Note that since $\mathrm{PPT'}(A:B)$ does not only contain states, but also subnormalized states,  \cref{lem:rel-ent-entangle-convex} and \cref{lem:rel-ent-entangle-al-concave} are not directly applicable. One can however verify that the corresponding statements for $\mathrm{PPT'}(A:B)$ still hold using the same arguments. For simplicity, let us denote 
    \begin{equation}
        b(\varepsilon):= \varepsilon \log \min \{ d_A , \, d_B \}  + (1+ \varepsilon) h \Big( \frac{\varepsilon}{1+ \varepsilon} \Big)  \, .
    \end{equation}
    To estimate an upper bound on the difference that appears in \cref{eq:continuity_bound_rains_information}, first note that, given $ T^1, T^2 :\cS(\cH_{A} \otimes \cH_{A'}) \rightarrow \cS(\cH_{A} \otimes \cH_{B}) $  two quantum channels with $\frac{1}{2} \norm{T^1-T^2}_{1 \rightarrow 1} \leq \varepsilon$, and $\rho_A \in \cS(\cH_A)$ with $\phi_{AA'}$ a purification of it, we have
    \begin{equation}
        \frac{1}{2} \norm{T^1(\phi_{AA'} ) - T^2 (\phi_{AA'} ) }_1 \leq  \frac{1}{2} \norm{T^1-T^2}_{1 \rightarrow 1} \leq \varepsilon \, .
    \end{equation}
    Consider now $\rho^1 , \rho^2 \in \cS(\cH_A)$ with respective purifications $\phi^1_{AA'} , \phi^2_{AA'}$, the states in which the respective maxima of $R(T^1)$ and $R(T^2)$ are attained. Then, we clearly have, for $i,j =1,2$ and $i \neq j$,
    \begin{equation}
       | R(T^j) - D_{\mathrm{PPT'}(A:B)}(T^i(\phi^j_{AA'})) |  = |  D_{\mathrm{PPT'}(A:B)}(T^j(\phi^j_{AA'}))- D_{\mathrm{PPT'}(A:B)}(T^i(\phi^j_{AA'})) | \leq b(\varepsilon) \, ,
    \end{equation}
    and thus,
    \begin{equation}
       R(T^i) \geq D_{\mathrm{PPT'}(A:B)}(T^i(\phi^j_{AA'})) \geq R(T^j) - b(\varepsilon) \, .
    \end{equation}
    Therefore, we can conclude
    \begin{equation}
       | R(T^1 ) - R(T^2)| \leq b(\varepsilon) \, ,
    \end{equation}
    and consequently
    \begin{equation}
       |R(T^1) - R(T^2)| \leq  \varepsilon \log \min \{ d_A , d_B\} + (1+ \varepsilon) h \Big( \frac{\varepsilon}{1 + \varepsilon} \Big) \, . 
    \end{equation}
\end{proof}

In a similar way, we can also prove uniform continuity and provide explicit continuity bounds for the BS-Rains information. Analogously to what we have done above for the Rains information, we can define for $\rho_{AB} \in \cS(\cH_A \otimes \cH_B)$ the following quantity:
\begin{equation}
   \widehat D_{\mathrm{PPT'}(A:B)}(\rho_{AB}) := \underset{\sigma_{AB} \in \mathrm{PPT'}(A:B)}{\text{min}} \widehat D ( \rho_{AB} \| \sigma_{AB}) \, ,
\end{equation}
and thus, we can rewrite, for a quantum channel $T :\cS(\cH_{A} \otimes \cH_{A'}) \rightarrow \cS(\cH_{A} \otimes \cH_{B})$,
\begin{equation}
  \widehat  R(T) := \underset{\rho_A \in \cS(\cH_A)}{\text{max}} \widehat  D_{\mathrm{PPT'}(A:B)}(T(\phi_{AA'}) ) \, ,
\end{equation}
for $\phi_{AA'}$ a purification of $\rho_A$. We can finally use \cref{lem:BS-ent-entangle-diff-bound} and  \cref{thm:BS-ent-entangle-cont-bound} from the previous section, for the BS-entropy, to obtain a continuity bound for the BS-Rains information. However, the bound obtained, as well as the procedure employed to derive it, are a straightforward combination of the strategies of the continuity bound for the Rains information \cref{theo:continuity_bound_rains_information} and the continuity bound for the BS-entropy of entanglement from \cref{thm:BS-ent-entangle-cont-bound}. Therefore, we omit it, to avoid unnecessary repetitions. 

\begin{theorem}\label{theo:continuity_bound_BS_rains_information}
    For $\varepsilon \in [0,1]$  and $ T_{A' \to B}^1, T^2_{A' \to B} :\cS(\cH_{A} \otimes \cH_{A'}) \rightarrow \cS(\cH_{A} \otimes \cH_{B}) $ two quantum channels with $\frac{1}{2}\|T^1- T^2\|_{1 \rightarrow 1} \leq \varepsilon$, we have:
    \begin{equation}
        |\widehat  R(T^1) - \widehat  R(T^2)| \leq  \varepsilon \log \min \{ d_A , d_B\} + (1+ \varepsilon) g_{d_{AB}} \Big( \frac{\varepsilon}{1 + \varepsilon} \Big) \, ,
    \end{equation}
    where $g_d(t) := \frac{d}{t^{1/d}}h(t) - \log(1 - t^{1/d})$.
\end{theorem}

%% file: sections/outlook.tex
In this paper, we have introduced a generalisation of the Alicki-Fannes-Winter method by Shirokov and applied it to derive results of uniform continuity and explicit continuity bounds for divergences. We gave this generalisation the name ALAFF (cf. \cref{theo:theo_alaff_method}) after the functions to which it applies (almost locally affine functions). The method allows deriving various continuity bounds for entropic quantities, by simply proofing (joint) convexity and almost (joint) concavity of the underlying divergence.

In particular, in the current paper, we have applied our ALAFF method to the specific cases of the Umegaki and the Belavkin-Staszewski relative entropies. For both of them, we have proven results of almost concavity (for the Umegaki case, our result is shown to be tight), and these, together with the well-known results of convexity for these quantities, have yielded a plethora of results of continuity bounds for both the Umegaki and BS-entropies, as well as for many other quantities derived from them. In particular, our results recover the previously known almost tight continuity bounds for the conditional entropy and the (conditional) mutual information. 

A natural question arises from the findings of this paper: Is our method applicable to any other family of divergences? We expect this to be the case, since, as shown in \cref{sec:main-results}, our method only requires almost concavity and convexity (already known for divergences) in order to work. Therefore, a result of almost concavity with a  ``well-behaved'' correction factor would be enough for the ALAFF method and is expected to exist, for families such as the $\alpha$-sandwiched Rényi divergences or the $\alpha$-geometric Rényi divergences, as they converge to the quantities studied in this paper. This possibility will be explored in a future manuscript.

Let us conclude this section, and our paper, with some analysis of the results obtained here. For both the Umegaki and the BS-entropies, we have presented results of almost concavity in order to provide some continuity bounds. However, while for the former (cf.\ \cref{theo:theo_almost_concavity_relative_entropy}) we have shown that the result is tight, for the latter (cf.\ \cref{theo:theo_almost_concavity_bs_entropy}) we are certain that there is room for improvement. Indeed, our almost concavity bound for the BS-entropy depends on the minimal eigenvalues of some of the states involved even in the simplified case of the BS-conditional entropy. In such a case, numerical simulations, as well as analytical proof, have shown us that there is a universal bound for the BS-conditional entropy of a state which is independent of the state involved. Therefore, we would expect an almost convexity result for the BS-conditional entropy being independent of the states involved, and this is clearly not the case at the moment.
Nevertheless, there is no doubt that the BS-entropy, and quantities derived from it, are ``pathological'' in some sense. First of all, we have shown that the BS-conditional entropy exhibits discontinuities in the presence of vanishing eigenvalues (cf.\ \cref{prop:discontinuity_conditional_BS_entropy}), as opposed to the conditional entropy, which behaves well in that setting. This motivates the idea that the minimal eigenvalue of the involved states should appear in the most general bounds of almost concavity and continuity. Additionally, we can compare some upper bounds of some entropic quantities derived from the Umegaki and the BS-entropy:
\begin{itemize}
    \item For the relative entropy, we have the following 3 bounds:
    \begin{equation}
       - H_\rho(A|B) \leq  \log d_A \; , \quad I_\rho (A:B) \leq 2 \log \min\{d_A, d_B\} \; , \quad D(\rho \| \sigma) \leq \log \widetilde m_\sigma^{-1} \, . 
    \end{equation}
     \item For the BS-entropy, we
    have the following 3 bounds (cf. \cref{prop:prop_bound_BS_entropic_quantities}):
    \begin{equation}
       - \widehat{H}_\rho(A|B) \leq  \log d_A \; , \quad \widehat I_\rho (A:B) \leq  \log d_A \widetilde{m}_{(\rho_A)}^{-1} \; , \quad \widehat D(\rho \| \sigma) \leq \log m_\sigma^{-1} \, . 
    \end{equation}
\end{itemize}
In the above $m_\cdot$, $\widetilde{m}_{\cdot}$ denote the minimal respectively minimal non-zero eigenvalue of the state in the index. It is remarkable that for the conditional and BS-conditional entropy and the mutual information, there appears no dependence on the minimal eigenvalue of the argument, whilst for the BS-mutual information this is the case.

Moreover, let us recall that, from the discussion in \cref{rem:both_conditional_BS_entropies_are_different}, we know that the conditional BS-entropy and its variational counterpart are different because the latter is continuous on $\mathcal{S}(\mathcal{H})$ and the former is not. One could wonder whether the same difference appears for the BS-mutual information. Analogously to the case of the (Umegaki) mutual information, we could define four possible versions of such a notion by optimizing over one marginal, both or none. Remarkably, we find that, when optimizing over both marginals, we have, assuming w.l.o.g. $d_A \le d_B$,
\begin{equation}
    \widehat{I}_\rho^{\text{var}}(A:B) := \underset{\sigma_A , \sigma_B}{\text{inf}} \widehat{D}(\rho_{AB} \| \sigma_A \otimes \sigma_B) \leq \log d_A - \widehat{H}_\rho^{\text{var}}(A|B) \le 2 \log \min\{d_A, d_B\} \, . 
\end{equation}
Comparing this bound to the one shown above for $\widehat I_\rho (A:B)$, which we prove to be tight in \cref{prop:prop_bound_BS_entropic_quantities}, we realize that the BS-mutual information and its variational counterpart (with optimization over both marginals) are also different in general. 

To conclude, the literature concerning continuity bounds for entropic quantities is much broader than the results collected here. For Rényi and Tsallis entropies, many results concerning their continuity can be derived from other techniques, such as majorization flows, and can be found in texts such as \cite{HansonDatta-MaxMinEntropyBounds-2018, HansonDatta-ContinuityBounds-2019, Hanson-ThesisEntropyBounds-2020}. Additionally, some of these results for the von Neumann entropy, Rényi and Tsallis entropies, as well as their classical counterparts, can be extended to energy-constrained systems in infinite dimensions, as shown in \cite{BeckerDattaRouze-ContinuityBounds-2021}, \cite{Winter-AlickiFannes-2016} (see also the recent \cite{Shirokov-ContinuityInfinite-2022}). We leave for future work the possibility of extending the results presented here to a similar framework.

%% file: sections/appendix.tex
\section{Numerical investigation of the variational definition of the BS-conditional entropy}\label{sec:sec_breakdown_numerical_formula}

\begin{figure}[ht!]
    \centering
    \includegraphics{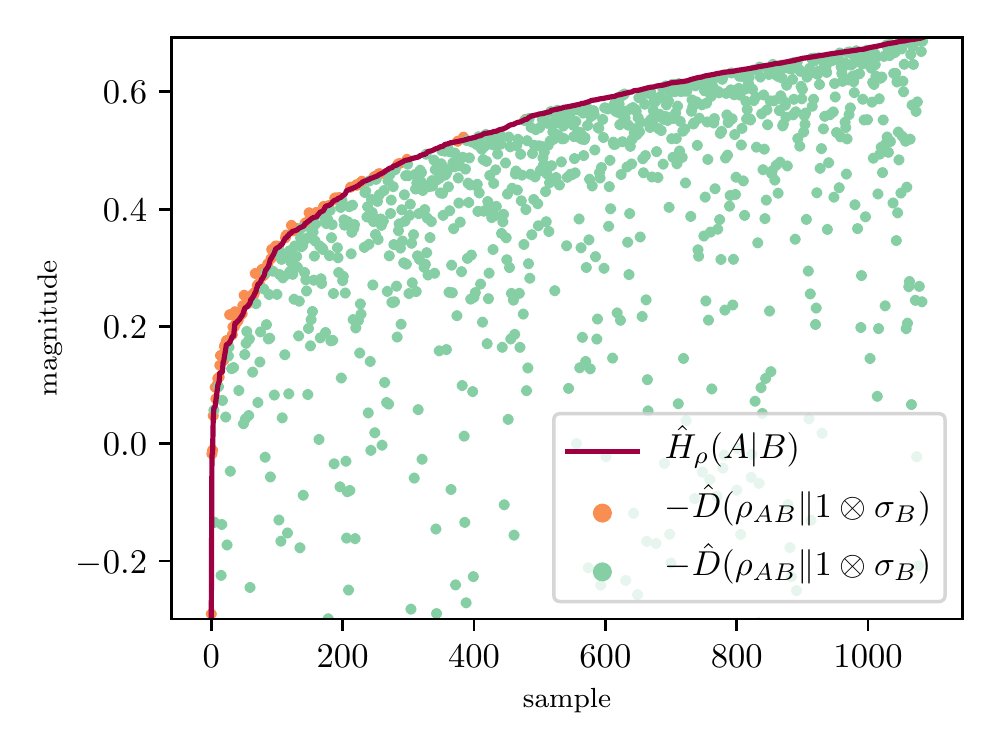}
    \caption{The red line is the BS-conditional entropy defined via the partial trace evaluated at $\rho_{AB}$. The dots are the BS-entropy between the state $\rho_{AB}$ and $\identity_A \otimes \sigma_B$ with $\sigma_B \in \cS(\cH_B)$ sampled at random. The orange dots are the cases when the $-\widehat{D}(\rho_{AB}\Vert \identity_A \otimes \sigma_B)$ exceeds $\widehat{H}(A|B)_\rho$. We sampled a total of 100.000 pairs of $\rho_{AB}$ and $\sigma_B$ and evaluated both $\widehat{H}(A|B)_\rho$ and $-\widehat{D}(\rho_{AB} \Vert \identity_A \otimes \sigma_B)$. Only a tenth of all samples were kept in addition to the ones that violated the bound. Those were then plotted in ascending order w.r.t the magnitude of their BS-conditional entropy. We further controlled the minimal eigenvalue and set $\cH_A \otimes \cH_B = \C^2\otimes\C^2$ to reduce the risk of numerical flaws.}
    \label{fig:breakdown_variational_formula_bsce}
\end{figure}

\section{Supplements to the proof of \texorpdfstring{\cref{theo:theo_almost_concavity_relative_entropy}}{Theorem 5.1}}\label{sec:sec_supplements_to_proof_almost_concavity_relative_entropy}

We will now show that the result of the inequality in \cref{eq:eq_int_ineq} is still true, even if $\rho_1, \rho_2, \sigma_1, \sigma_2$ are not full rank. We have that
\begin{equation}
    \ker \sigma \subseteq \ker \sigma_1 \subseteq \ker \rho_1. \label{eq:eq_kernel_inclusions_proof_almost_concavity_relative_entropy}
\end{equation}
If $\ker \sigma \subsetneq \ker \rho_1$ we set
\begin{equation}
    \widetilde{\Pi}_{\rho_1} := P_{\ker{\rho_{1}} \cap (\ker{\sigma})^{\perp}} \quad, \qquad \Pi_{\rho_1} := \norm{\widetilde{\Pi}_{\rho_1}}_1^{-1} \, \widetilde{\Pi}_{\rho_1}\, ,
\end{equation}
and if $\ker \sigma \subsetneq \ker \sigma_1$, 
\begin{equation}
    \widetilde{\Pi}_{\sigma_1} := P_{\ker{\sigma_{1}} \cap (\ker{\sigma})^{\perp}}\quad, \qquad \Pi_{\sigma_1} := \norm{\widetilde{\Pi}_{\sigma_1}}_1^{-1} \, \widetilde{\Pi}_{\sigma_1} \, ,
\end{equation}
normalised projections on the spaces in the index. Both of the latter are quantum states and fulfil
\begin{equation}
    \Pi_{\rho_1} \rho_1 = \rho_1 \Pi_{\rho_1} = 0, \qquad \Pi_{\sigma_1}\sigma_1 = \sigma_1 \Pi_{\sigma_1} = 0, \qquad \Pi_{\sigma_1} \rho_1 = \rho_1 \Pi_{\sigma_1}  = 0 \, . \label{eq:eq_orthongonality_proof_almost_concavity_relative_entropy}
\end{equation}
For $1 > \varepsilon > 0$ and $1 > \delta > 0$, let
\begin{equation}
    \begin{aligned}
        \rho_{1, \varepsilon} &= \begin{cases}
            \varepsilon \Pi_{\rho_1} + (1 - \varepsilon) \rho_1 & \text{if } \ker \sigma \subsetneq \ker \rho_1\\
            \rho_1 & \text{if } \ker \sigma = \ker \rho_1
        \end{cases},\\
        \sigma_{1, \delta} &= \begin{cases}
            \delta \Pi_{\sigma_1} + (1 - \delta) \sigma_1 & \text{if } \ker \sigma \subsetneq \ker \sigma_1\\
            \sigma_1 & \text{if } \ker \sigma = \ker \sigma_1
        \end{cases}.\\
    \end{aligned}
\end{equation}
We then have that $\ker \rho_{1, \varepsilon} = \ker \sigma_{1, \delta} = \ker \sigma$. This means, however, considering $\tr[\rho_{1, \varepsilon}(\log \sigma - \log \sigma_{1, \delta})]$ we can reduce to the subspace where they are all full rank. We then apply the Peierls-Bogoliubov inequality \cite{OhyaPetz-Entropy-1993} and the multivariant trace inequality by Sutter et al. \cite[Corollary 3.3]{SutterBertaTomamichel-Multivariate-2017}
\begin{equation}
    \begin{aligned}
        \tr[\rho_{1, \varepsilon}(\log\sigma - \log \sigma_{1, \delta})] &\le \log\tr\left[\exp\left(\log(\rho_{1, \varepsilon}) + \log(\sigma) - \log(\sigma_{1, \delta})\right)\right]\\
        &\le\log\int\limits_{-\infty}^\infty dt\,\beta_0(t)\,\tr\left[\rho_{1, \varepsilon} \sigma_{1, \delta}^{\frac{it - 1}{2}} \sigma \sigma_{1, \delta}^{\frac{-it -1}{2}}\right].\label{eq:eq_int_ineq_proof_almost_concavity_relative_entropy}
    \end{aligned}
\end{equation}
Both of the traces on the LHS and RHS of \cref{eq:eq_int_ineq_proof_almost_concavity_relative_entropy} can without change be extended to the full Hilbert space again. Next, we take limits on both sides of the inequality and in doing so recover the claim. We first note that the limit $\varepsilon \to 0$ requires no more argument as both sides are linear in $\varepsilon$. Hence, we get
\begin{equation}\label{eq:almost_concavity_RE_auxiliary_non-full-rank}
    \tr[\rho_1(\log\sigma - \log \sigma_{1, \delta})] \le \log\int\limits_{-\infty}^\infty dt\,\beta_0(t)\,\tr\left[\rho_1 \sigma_{1, \delta}^{\frac{it - 1}{2}} \sigma \sigma_{1, \delta}^{\frac{-it -1}{2}}\right].
\end{equation}
The limit $\delta \to 0$ on the other hand is, in the case of $\ker \sigma \subsetneq \ker \sigma_1$, a little more involved. Due to the orthogonality in \cref{eq:eq_orthongonality_proof_almost_concavity_relative_entropy} we cannot only split up the logarithm but also eliminate terms. More specifically, we have
\begin{equation}
    \log \sigma_{1, \delta} = \log(\delta \Pi_{\sigma_{1}} ) + \log((1-\delta)\sigma_{1}) \, ,
\end{equation}
where the logarithms in the RHS have to be understood as living in the support of the respective argument (and complemented with zeros in the rest). Hence, we obtain for the LHS of \cref{eq:almost_concavity_RE_auxiliary_non-full-rank}
\begin{equation}
    \begin{aligned}
        \tr[\rho_1(\log \sigma - \log \sigma_{1, \delta})] &= \tr[\rho_1(\log \sigma - \log (\delta \Pi_{\sigma_1} + (1 - \delta) \sigma_1))]\\
        &= \tr[\rho_1(\log \sigma - \log ((1 - \delta) \sigma_1)] + \tr[\rho_1 \log(\delta \Pi_{\sigma_1})]\\
        &= \tr[\rho_1(\log \sigma - \log ((1 - \delta) \sigma_1)] \\
        &= \tr[\rho_1(\log \sigma - \log\sigma_1] + \log(1 - \delta) \, .
    \end{aligned}
\end{equation}
Moreover, for the RHS of \cref{eq:almost_concavity_RE_auxiliary_non-full-rank} we use that
\begin{equation}
    \sigma_{1, \delta}^{z} = \delta^{z} \Pi_{\sigma_{1}}^{z} + (1-\delta)^{z} \sigma_{1}^{z} \, ,
\end{equation}
for any $z \in \mathbb{C}$,  where the last exponential has to be understood again in the support of the respective argument. Thus, we obtain
\begin{equation}
    \begin{aligned}
        \tr\left[\rho_1 \sigma_{1, \delta}^{\frac{it - 1}{2}} \sigma \sigma_{1, \delta}^{\frac{-it -1}{2}}\right] & = (1 - \delta)^{-1} \tr\left[\rho_1 \sigma_{1}^{\frac{it - 1}{2}} \sigma \sigma_{1}^{\frac{-it -1}{2}}\right] \\
        & + (1 - \delta)^{\frac{it-1}{2}} \delta^{\frac{-it-1}{2}}  \tr\left[\rho_1 \sigma_{1}^{\frac{it - 1}{2}} \sigma \Pi_{\sigma_1}^{\frac{-it -1}{2}}\right] \\
        & + \delta^{\frac{it-1}{2}} (1 - \delta)^{\frac{-it-1}{2}}  \tr\left[\rho_1 \Pi_{\sigma_1}^{\frac{it - 1}{2}} \sigma \sigma_{1}^{\frac{-it -1}{2}}\right] \\
        & + \delta^{- 1}   \tr\left[\rho_1 \Pi_{\sigma_1}^{\frac{it - 1}{2}} \sigma \Pi_{\sigma_1}^{\frac{-it -1}{2}}\right]\\
        &= (1 - \delta)^{-1} \tr\left[\rho_1 \sigma_{1}^{\frac{it - 1}{2}} \sigma \sigma_{1}^{\frac{-it -1}{2}}\right] \,.
    \end{aligned}
\end{equation}
Taking the limit $\delta \to 0$ now directly follows from the continuity of the logarithm. We thereby conclude 
\begin{equation}
     p\tr[\rho_1 (\log(\sigma) - \log(\sigma_1))] \le p \log\int\limits_{-\infty}^\infty dt\,\beta_0(t)\,\tr\left[\rho_1 \sigma_1^{\frac{it - 1}{2}} \sigma \sigma_1^{\frac{-it -1}{2}}\right] \, ,
\end{equation}
for $\sigma_1, \sigma_2, \rho_1$ not full rank.

\section{Proof of \texorpdfstring{\cref{prop:prop_almost_concave_estimate_relative_entropy_well_behaved}}{Proposition 5.2}}\label{sec:sec_proof_almost_concave_estimate_relative_entropy_well_behaved}

We first of all note that for all $\rho_1, \rho_2 \in \cS(\cH)$ we have $\frac{1}{2}\norm{\rho_1 - \rho_2}_1 \le 1$, hence as a direct consequence $f_{c_1, c_2} + \frac{1}{2}\norm{\rho_1 - \rho_2}_1 h \le f_{c_1, c_2} + h$. We therefore will drop the $\frac{1}{2}\norm{\rho_1 - \rho_2}$ in front of the $h$ here already. 
\begin{enumerate}
    \item If $\sigma_1 = \sigma_2 =: \sigma$, we find for $j=1,2$ that 
    \begin{equation}
        c_j =  \int\limits_{-\infty}^\infty dt \beta_0(t)  \tr[\rho_j \sigma^{\frac{it - 1}{2}} \sigma \sigma^{\frac{-it - 1}{2}}] =  \int\limits_{-\infty}^\infty dt \beta_0(t) \tr[\rho_j] = 1.
    \end{equation}
    The reduction of $f_{c_1, c_2} + h$ to $h$ then happens because $\log(p + (1 - p)) = \log(1) = 0$ gives $f_{c_1, c_2} = 0$.
    \item With $j, k = 1, 2$, $j \ne k$ and $\widetilde m \rho_j \le \sigma_j$, we find
    \begin{equation}
        \sigma_j^{\frac{it -1}{2}} \rho_j \sigma_j^{\frac{-it-1}{2}} \le \sigma_j^{\frac{it -1}{2}} \widetilde m^{-1} \sigma_j \sigma_j^{\frac{-it-1}{2}} \le \widetilde{m}^{-1} P_{\sigma_j} \le \widetilde{m}^{-1} \identity
    \end{equation}
    where $P_{\sigma_j}$ is the projection onto the support of $\sigma_j$. We therefore find
    \begin{equation}
        c_j \le \int\limits_{-\infty}^\infty dt \beta_0(t) \widetilde m^{-1} \tr[\sigma_k] = \widetilde m^{-1}  \, . 
    \end{equation}
    By the monotonicity of the logarithm, we obtain $f_{c_1, c_2} \le f_{\widetilde m^{-1}, \widetilde m^{-1}}$ and hence $f_{c_1, c_2} + h \le f_{\widetilde m^{-1}, \widetilde m^{-1}} + h$.
    \item For $j, k = 1, 2$, $j \ne k$ we have 
    \begin{equation}
        \begin{aligned}
            c_j &=  \int\limits_{-\infty}^\infty dt \beta_0(t)  \tr[\rho_{j, AB} (\identity_A \otimes \rho_{j, B})^{\frac{it - 1}{2}} \identity_A \otimes \rho_{k, B} (\identity_A \otimes \rho_{j, B})^{\frac{-it - 1}{2}}]\\
            &=  \int\limits_{-\infty}^\infty dt \beta_0(t)  \tr[\rho_{j, AB} \identity_A \otimes (\rho_{j, B}^{\frac{it - 1}{2}}\rho_{k, B}  \rho_{j, B}^{\frac{-it - 1}{2}})]\\
            &=  \int\limits_{-\infty}^\infty dt \beta_0(t)  \tr[\rho_{j, B}(\rho_{j, B}^{\frac{it - 1}{2}}\rho_{k, B}  \rho_{j, B}^{\frac{-it - 1}{2}})]\\
            &=  \int\limits_{-\infty}^\infty dt \beta_0(t)  \tr[\rho_{k, B}] = 1.
        \end{aligned}
    \end{equation}
    We used that the functional calculus has the property that $f(A \otimes B) = f(A) \otimes f(B)$ for $A$, $B$ self-adjoint, as can easily be verified by direct computation, and that the trace is cyclic. This gives us $f_{c_1, c_2} = f_{1, 1} = 0$ which concludes the claim.
    \item The derivative of $p \mapsto \frac{1}{1 - p} h(p)$ at $p \in (0, 1)$ is $-\frac{\log(p)}{(1 - p)^2} \ge 0$, which proves the second assertion.
    For $p \mapsto \frac{1}{1 - p}f_{m_1, m_2}(p) = \frac{p}{1 - p} \log(p + m_1 (1 - p)) + \log(1 - p + m_2 p)$ we use similar reasoning. First we use that $m_2 \ge 1$ hence $\log(1 - p + m_2 p) = \log(1 + (m_2 - 1)p)$ is monotone in $p$, i.e. in particular non-decreasing. Second we note that $p \mapsto \frac{p}{1 - p} \log(p + m_1 (1 - p))$ is monotone in $p$, because forming the derivative at $p \in (0, 1)$, we get 
    \begin{equation}
        \begin{aligned}
              \frac{1}{(1 - p)^2}&\Big(\frac{p}{p + (1 - p)m_1} + \log(p + m_1(1 - p)) - p\Big)\\
              &\ge \frac{1}{(1 - p)^2}\Big(\frac{p}{p + (1 - p)m_1} + \frac{p + (1 - p)m_1 - 1}{p + m_1(1 - p)} - p\Big)\\
              &= \frac{1}{(1 - p)^2}\Big(\frac{m_1(1 - p) + 2p - 1}{p + (1 - p)m_1} - p\Big)\\
              &= \frac{1}{(1 - p)^2}\Big(\frac{m_1(1 - p) + 2p - 1 - p(p + (1 - p)m_1)}{p + (1 - p)m_1}\Big)\\
              &= \frac{1}{(1 - p)^2}\Big(\frac{(m_1 - 1)(p - 1)^2}{p + (1 - p)m_1}\Big)\\
              &\ge 0 \, .
        \end{aligned}
    \end{equation}
    We used that for $x \ge 1$, $\log(x) \ge \frac{x - 1}{x}$ (this can be seen by taking the derivative and realizing that both sides coincide for $x=1$) and $m_1 \ge 1$. This concludes the claim.
\end{enumerate}

\section{Proof of \texorpdfstring{\cref{lem:delta_invariant_subset}}{Proposition 5.13}}\label{sec:sec_proof_delta_invariant_subset}

We first show that for $s \ge \widetilde{m}$, $\cS_0$ is $s$-perturbed $\Delta$-invariant. For that purpose let $\sigma_1, \sigma_2 \in \cS_0$, then we find
\begin{equation}
    \Delta^\pm(\sigma_1, \sigma_2, \rho) = s \rho + (1 - s)[\sigma_1 - \sigma_2]_\pm \ge \widetilde{m} \rho,
\end{equation}
which immediately gives the kernel inclusion as well as the condition to be lower bounded by $\widetilde{m} \rho$. Therefore, $\Delta^\pm(\sigma_1, \sigma_2, \tau) \in \cS_0$ which makes $\cS_0$ an $s$-perturbed $\Delta$-invariant set. We show the other direction by contrapositive. Let $s < \widetilde{m}$. Since $\widetilde{m} < 1$ and $\rank \rho \ge 2$ we find an $\varepsilon > 0$ and two orthonormal $\ket{0}, \ket{1} \in \supp \rho$, such that $\widetilde{m} \rho < \rho - \frac{\varepsilon}{2}\dyad{i}$ for $i = 0, 1$. We then have that 
\begin{equation}
    \begin{aligned}
        \sigma_1 &= \rho + \frac{\varepsilon}{2}\dyad{0} - \frac{\varepsilon}{2} \dyad{1}\\
        \sigma_2 &= \rho - \frac{\varepsilon}{2}\dyad{0} + \frac{\varepsilon}{2} \dyad{1}
    \end{aligned}
\end{equation}
manifestly are contained in $\cS_0$. Furthermore, $\frac{1}{2}\norm{\sigma_1 - \sigma_2}_1 = \varepsilon$ and
\begin{equation}
    \begin{aligned}
        \varepsilon^{-1}[\sigma_1- \sigma_2]_+ &= \dyad{0}\\
        \varepsilon^{-1}[\sigma_1 - \sigma_2]_- &= \dyad{1}
    \end{aligned}.
\end{equation}
We will now show that there exists no $\tau \in \cS(\cH)$ such that $\Delta^\pm(\sigma_1, \sigma_2, \tau) \in \cS_0$ again, meaning $\cS_0$ is not $s$-perturbed $\Delta$-invariant. Assume there is an operator $\tau \ge 0$ such that $\Delta^\pm(\sigma_1, \sigma_2, \tau)\in \cS_0$ we then would have 
\begin{equation}
    \begin{aligned}
        \dyad{0}^\perp \Delta^+(\sigma_1, \sigma_2, \tau) \dyad{0}^\perp &= \dyad{0}^\perp s \tau \dyad{0}^\perp \ge \widetilde{m}\dyad{0}^\perp\rho \dyad{0}^\perp\\
        \dyad{1}^\perp \Delta^-(\sigma_1, \sigma_2, \tau) \dyad{1}^\perp &= \dyad{1}^\perp s \tau \dyad{1}^\perp \ge \widetilde{m}\dyad{1}^\perp\rho \dyad{1}^\perp
    \end{aligned}
    \label{eq:eq_delta_state_projections}
\end{equation}
where $\dyad{i}^\perp := P_{\rho} - \dyad{i}$ for $i =0, 1$. Here $P_\rho$ is the projection on the support of $\rho$. We further used $\Delta^\pm(\sigma_1, \sigma_2, \tau) \ge \widetilde{m}\rho$ as $\Delta^\pm(\sigma_1, \sigma_2, \tau)$ are in $\cS_0$ by assumption. To fulfil \cref{eq:eq_delta_state_projections} we clearly need to choose $s > 0$ and since $s < \widetilde{m}$ we directly obtain the conditions
\begin{equation}
    \dyad{0}^\perp \tau \dyad{0}^\perp \gvertneqq  \dyad{0}^\perp \rho \dyad{0}^\perp \quad\text{and}\quad \dyad{1}^\perp \tau \dyad{1}^\perp \gvertneqq \dyad{1}^\perp\rho \dyad{1}^\perp \,.
\end{equation}
This gives us, 
\begin{equation}
    \begin{aligned}
        \tr[\tau] &\ge \tr[\dyad{0}^\perp \tau \dyad{0}^\perp + \dyad{0} \tau \dyad{0}] = \tr[\dyad{0}^\perp \tau \dyad{0}^\perp + \dyad{0} \dyad{1}^\perp \tau \dyad{1}^\perp \dyad{0}]\\
        &> \tr[\dyad{0}^\perp \rho \dyad{0}^\perp + \dyad{0} \dyad{1}^\perp \rho \dyad{1}^\perp \dyad{0}] = \tr[\dyad{0}^\perp \rho \dyad{0}^\perp + \dyad{0}  \rho  \dyad{0}]\\
        &= \tr[P_\rho \rho] = \tr[\rho] =  1,
    \end{aligned}
\end{equation}
where we used that $\ket{0}$ and $\ket{1}$ are orthogonal, hence $\dyad{0} \dyad{1}^\perp = \dyad{1}^\perp \dyad{0} = \dyad{0}$ and $\dyad{0}^2 = \dyad{0}, (\dyad{0}^\perp)^2 = \dyad{0}^\perp$. We thus conclude $\tau \not\in \cS(\cH)$ proving the claim.

\section{Proof of \texorpdfstring{\cref{prop:prop_almost_concave_estimate_bs_entropy_well_behaved}}{Proposition 6.5}}\label{sec:sec_proof_almost_concave_estimate_BS_entropy_well_behaved}

\begin{enumerate}
    \item If $\sigma_1 = \sigma_2 = \sigma$, then for $j = 1, 2$
    \begin{equation}
        \begin{aligned}
            \hat{c}_j &= \int\limits_{-\infty}^\infty dt \beta_0(t) \tr[\rho_j(\rho_j^{1/2} \sigma^{-1} \rho_j^{1/2})^{\frac{it + 1}{2}} \rho_j^{-1/2}\sigma\rho_j^{-1/2}(\rho_j^{1/2} \sigma^{-1} \rho_j^{1/2})^{\frac{-it + 1}{2}}]\\
            &= \int\limits_{-\infty}^\infty dt \beta_0(t) \tr[\rho_j] = \int\limits_{-\infty}^\infty dt \beta_0(t) = 1
        \end{aligned}
    \end{equation}
    which gives us immediately $f_{\hat{c}_1, \hat{c}_2} + \hat{c}_0h = \hat{c}_0 h$.
    \item For $j, k = 1, 2$ with $j \ne k$ we first have $\sigma_k \le m^{-1} \sigma_j$ giving us
    \begin{equation}
        \begin{aligned}
            \hat{c}_j &\le \int\limits_{-\infty}^\infty dt \beta_0(t) \tr[\rho_j(\rho_j^{1/2} \sigma_j^{-1} \rho_j^{1/2})^{\frac{it + 1}{2}} \rho_j^{-1/2} m^{-1}\sigma_j\rho_j^{-1/2}(\rho_j^{1/2} \sigma_j^{-1} \rho_j^{1/2})^{\frac{-it + 1}{2}}]\\
            &= m^{-1} \int\limits_{-\infty}^\infty dt \beta_0(t) \tr[\rho_j] = m^{-1} \, .
        \end{aligned}
    \end{equation}
    Since $\hat{c}_0 \le m^{-1}$ and because the logarithm is monotone this immediately gives $f_{\hat{c}_1, \hat{c}_2} + \hat{c}_0 h \le f_{m^{-1}, m^{-1}} + m^{-1} h$.
    \item The proof is along the same lines as the one for 2., however with $\sigma_j = d_A^{-1} \identity_A \otimes \rho_{j,B}$. We just have to show that the minimal eigenvalue of $\sigma_j$ is bounded from below by $m$. We use that $T_A: \tau \mapsto d_A^{-1}\identity_A \otimes \tau_B$ is a conditional expectation and that $d_A^{-1} \identity_A \otimes \tau_B$ is full rank if $\tau$ was full rank \cite[Theorem 4.13]{Carlen-TraceInequalities-2009}. This means, however, 
    \begin{equation}
        (d_A^{-1} \identity_A \otimes \rho_{j, B})^{-1}= T_A(\rho_j)^{-1} \le T_A(\rho_j^{-1}) \, ,
    \end{equation}
    where we used \cite[Theorem 4.16]{Carlen-TraceInequalities-2009}. This gives us
    \begin{equation}\label{eq:eq_minimal_eigenvalue_partial_trace}
        \norm{(d_A^{-1} \identity_A \otimes \rho_{j, B})^{-1}}_\infty \le \norm{T_A(\rho^{-1})}_\infty \le \norm{\rho^{-1}}_\infty \le  m^{-1} \, .
    \end{equation} 
    Hence, we have that $\norm{(d_A^{-1} \otimes \rho_{j, B})^{-1}}_\infty^{-1}$ the minimal eigenvalue of $d_A^{-1} \otimes \rho_{j, B}$ is bounded from below by $m$. From here on the proof is analogous to the one in 2. We obtain $f_{\hat{c}_1, \hat{c}_2} + \hat{c}_0 h \le f_{m^{-1}, m^{-1}} + \hat{c}_0 h$ and again use \cref{eq:eq_minimal_eigenvalue_partial_trace} to get $f_{m^{-1}, m^{-1}} + \hat{c}_0 h \le f_{m^{-1}, m^{-1}} + m^{-1} h$.
    \item The proof is completely analogous to the one in 4. of \cref{sec:sec_proof_almost_concave_estimate_relative_entropy_well_behaved}.
\end{enumerate}

\section{Proof of \texorpdfstring{\cref{prop:prop_bound_BS_entropic_quantities}}{Proposition 6.6}}\label{sec:sec_proof_bound_BS_entropic_quantities}

\begin{enumerate}
    \item We begin with the BS-conditional information. The upper bound on $\widehat{H}_\cdot(A|B)$ can be obtained by
    \begin{equation}
        \widehat{H}_\rho(A|B) = -\widehat{D}(\rho_{AB} \Vert d_A^{-1} \identity_A \otimes \rho_B) + \log d_A \le \log d_A \, .
    \end{equation}
    where we used the non-negativity of $\widehat{D}(\cdot\Vert \cdot)$ on quantum states. The bound is attained if one inserts the maximally mixed state, i.e., $\rho_{AB} = d_{AB}^{-1} \identity_{AB}$. For the lower bound we use that $-\widehat{D}(\cdot \Vert \cdot)$ is jointly concave and $\tr_A[\cdot]$ linear which means without loss of generality one can assume $\rho$ to be pure, i.e., a rank one projection. Then
    \begin{equation}
        \begin{aligned}
             \widehat{H}_{\dyad{\psi}}(A|B) &= -\widehat{D}(\dyad{\psi} \Vert \identity_A \otimes P_B) = -\tr[\dyad{\psi} \log \dyad{\psi}^{1/2} (\identity_A \otimes P_B^{-1}) \dyad{\psi}^{1/2}]\\
             &= -\log \tr[\dyad{\psi} (\identity_A \otimes P_B^{-1})] = -\log\tr[P_B P_B^{-1}] \, ,
        \end{aligned}
    \end{equation}
    with $P_B = \tr_A[\dyad{\psi}]$. Employing the Schmidt decomposition to $\dyad{\psi}$ we find that 
    \begin{equation}
        P_B = \sum\limits_{i = 1}^d \lambda_i^2 P_i
    \end{equation}
    with $P_i$ orthogonal rank one projections on $\mathcal{H}_{B}$, $\lambda_i^2 > 0$ and $\sum\limits_{i = 1}^d \lambda_i^2 = 1$. Further $d \le \min\{d_A, d_B\}$ the Schmidt rank. This gives us that
    \begin{equation}
        \tr[P_B P_B^{-1}] = \sum\limits_{i = 1}^d \lambda_i^2 \lambda_i^{-2} = d \le \min\{d_A, d_B\}.
    \end{equation}
    Through monotonicity of the logarithm, we obtain the lower bound, i.e.,
    \begin{equation}
        \widehat{H}_\rho(A|B) \ge -\log\min\{ d_A, d_B\}.
    \end{equation}
    This bound is attained for $\rho$ a pure state with full Schmidt rank, which can directly be seen from the above calculations.
    \item We now tackle the BS-mutual information. The lower bound, i.e. $\widehat{I}_\rho(A:B) \geq 0$, is a direct consequence of the data processing inequality \cite{HiaiMosonyi_2011}. Applying $\tr_A[\,\cdot\,]$, we find
    \begin{equation}
        \widehat{I}_\rho(A:B) = \widehat{D}(\rho_{AB}\Vert \rho_A \otimes \rho_B) \ge \widehat{D}(\rho_B \Vert \rho_B) = 0.
    \end{equation}
    To proof the upper bound, we w.l.o.g assume that $\norm{\rho_A^{-1}}_\infty \le \norm{\rho_B^{-1}}_\infty$. We then use that $\rho_A \otimes \rho_B \ge \norm{\rho_A^{-1}}_\infty^{-1} P_{\rho_A} \otimes \rho_B$, where $P_{\rho_A}$ is the projection to the support of $\rho_A$. This gives us 
    \begin{equation}
        \begin{aligned}
            \widehat{I}_\rho(A:B) &= \widehat{D}(\rho_{AB} \Vert \rho_A \otimes \rho_B) \le \widehat{D}(\rho_{AB} \Vert P_{\rho_A}  \otimes \rho_B) + \log \norm{\rho_A^{-1}}_\infty\\
            &=\widehat{D}(\rho_{AB} \Vert \identity_A \otimes \rho_B) + \log\norm{\rho_A^{-1}}_\infty = - \widehat{H}_\rho(A|B) + \log\norm{\rho_A^{-1}}_\infty\\
            &\le \log\min\{d_A, d_B\} + \log \norm{\rho_A^{-1}}_\infty\\
            &\le\log\min\{d_A, d_B\} + \log\min\{ \norm{\rho_A^{-1}}_\infty, \norm{\rho_B^{-1}}_\infty\}
        \end{aligned}
    \end{equation}
    In the second equality we used that $(\ker \rho_A) \otimes \cH_{B} \subseteq \ker \rho_{AB}$, so extending $P_{\rho_A}$ to $\identity_A$ has no effect. With the next example, we will see that the bound is tight and scales with $\log\max\{\norm{\rho_A^{-1}}_\infty, \norm{\rho_B^{-1}}_\infty\}$ in some cases. For that purpose let $d_A \in \N$, $d_A \ge 2$ and a bipartite space $\cH_A\otimes \cH_B$ with $\cH_A$ having dimension $d_A$ and $\cH_B$ dimension $d_B = d_A + 1$. Furthermore, let $\varepsilon \in (0,1)$. We then consider sets of orthonormal vectors $\{\ket{i_A}\}_{i = 1}^{d_A} \subset \cH_A$, $\{\ket{i_B}\}_{i = 1}^{d_A} \subset \cH_B$ and define
    \begin{equation}
        \ket{\psi} := \sum\limits_{i = 1}^{d_A - 1} \sqrt{\frac{\varepsilon}{d_A - 1}} \ket{i_A} \otimes \ket{i_B} + \sqrt{1 - \varepsilon} \ket{(d_{A})_A} \otimes \ket{(d_A)_B} = \sum\limits_{i = 1}^{d_A} \sqrt{\lambda_i} \ket{i_A} \otimes \ket{i_B} \, .
    \end{equation}
    with the $\lambda_i$ defined accordingly. We find that
    \begin{equation}
        \begin{aligned}
            \rho_A := \tr_B[\dyad{\psi}] &= \sum\limits_{i = 1}^{d_A} \lambda_i \dyad{i_A}\,,\\
            \rho_B := \tr_A[\dyad{\psi}]&= \sum\limits_{i = 1}^{d_A}\lambda_i \dyad{i_B}\,,
        \end{aligned}
    \end{equation}
    and the Moore-Penrose pseudoinverse (in the case of $\rho_A$ it is an inverse)
    \begin{equation}
        \begin{aligned}
            \rho_A^{-1} &= \sum\limits_{i = 1}^{d_A} \lambda_i^{-1} \dyad{i_A}\,,\\
            \rho_B^{-1} &= \sum\limits_{i = 1}^{d_A}\lambda_i^{-1} \dyad{i_B}\,.
        \end{aligned}
    \end{equation}
    We find
    \begin{equation}
        \begin{aligned}
            \tr[\dyad{\psi} \rho_A^{-1}\otimes \rho_B^{-1}] &= \sum\limits_{i, j, k, l} \frac{\sqrt{\lambda_i}\sqrt{\lambda_j}}{\lambda_k \lambda_l} \braket{i_A | k_A} \braket{k_A | j_A} \braket{i_B | l_B} \braket{l_B | j_B}\\
            &= \sum\limits_{i, j, k, l} \frac{\sqrt{\lambda_i}\sqrt{\lambda_j}}{\lambda_k \lambda_l}  \delta_{ik} \delta_{kj} \delta_{il} \delta_{lj}\\
            &= \sum\limits_{i} \frac{1}{\lambda_i} = \frac{(d_A - 1)^2}{\varepsilon} + \frac{1}{1 - \varepsilon} \, ,
        \end{aligned}
    \end{equation}
    with which, as $\dyad{\psi}$ is a rank one projection
    \begin{equation}
        \begin{aligned}
            \widehat{I}_{\dyad{\psi}}(A:B) &= \tr[\dyad{\psi}\log(\dyad{\psi}^{1/2} \rho_A^{-1} \otimes \rho_B^{-1} \dyad{\psi}^{1/2})]\\
            &= \log \tr[\dyad{\psi} \rho_A^{-1} \otimes \rho_B^{-1}]\\
            &= \log(\frac{(d_A - 1)^2}{\varepsilon} + \frac{1}{1 - \varepsilon}) \ge \log(\frac{(d_A - 1)^2}{\varepsilon}) \, . \label{eq:eq_exact_BS_mutual_information_proof}
        \end{aligned}
    \end{equation}
    We directly obtain $\norm{\rho_A^{-1}}_\infty = \norm{\rho_B^{-1}}_\infty = \frac{d_A - 1}{\varepsilon}$ and by construction $d_A < d_B$, hence the bound in \cref{eq:eq_bound_BS_mutual_information} gives
    \begin{equation}
        \widehat{I}_{\dyad{\psi}}(A:B) \le \log(\frac{d_A(d_A - 1)}{\varepsilon})\, . \label{eq:eq_upper_bound_BS_mutual_information_proof}
    \end{equation}
    We first note that for $\varepsilon = 1 - \frac{1}{d_A}$ we get equality in \cref{eq:eq_upper_bound_BS_mutual_information_proof}. What is, however, more interesting is the fact that 
    \begin{equation}
        \log(\frac{(d_A - 1)^2}{\varepsilon}) \le \widehat{I}_{\dyad{\psi}}(A:B) \le \log(\frac{d_A(d_A - 1)}{\varepsilon}) \, ,
    \end{equation}
    with 
    \begin{equation}
        \Big|\log(\frac{d_A(d_A - 1)}{\varepsilon}) - \log(\frac{(d_A - 1)^2}{\varepsilon})\Big| = \log(\frac{d_A}{d_A - 1}) \, .
    \end{equation}
    I.e., the error of the bound is of order $\log(\frac{d_A}{d_A - 1})$ independent of the $\varepsilon$. This means, that the scaling behaviour of the bound, in terms of the minimal non-zero eigenvalue of $\rho_A$ and $\rho_B$ respectively is the best one can do. 
    \item The lower bound of the BS-CMI is again a consequence of the data processing inequality. The upper bound is a direct consequence of the bounds obtained for the BS-conditional information due to the definition of the conditional mutual information in \cref{eq:eq_BS_conditional_mutual_information}
    \begin{equation}
        \begin{aligned}
             \widehat{I}_\rho(A:B|C) &= \widehat{H}_\rho(A|C) - \widehat{H}_\rho(A|BC)\\ 
             &\le \log d_A + \log\min\{d_A, d_{BC}\}\\
             &= \log\min\{d_A^2, d_{ABC}\} \, .
        \end{aligned}
    \end{equation}
    We expect that the tightness of such a bound can be proven in a similar way to the one for the BS-mutual information.
\end{enumerate}

\section{Behavior of \texorpdfstring{$g_d$}{g\_d}}

In this section, we study the function $g_d(p) := \frac{d}{p^{1/d}}h(p) - \log(1 - p^{1/d})$ for $p \in (0,1)$ and a fixed $d \in \mathbb{N}$, $d \ge 2$.  This function appears in some of the continuity bounds in \cref{sec:min-dist-sep-states}.

\begin{lemma} \label{lem:gd-cont}
    Let $d \in \mathbb N$, $d \geq 2$. Then, $\lim_{p \to 0^+} g_d(p) = 0$. In particular, $g_d$ is continuous on $p \in [0,1)$.
\end{lemma}
\begin{proof}
    Since $\lim_{p \to 0^+} \log(1 - p^{1/d}) = 0$, we can focus on $\frac{d}{p^{1/d}}h(p)$. The assertion follows from applying L'Hospital's rule twice. Indeed,
    \begin{align}
        \lim_{p \to 0^+} \frac{d}{p^{1/d}}h(p) &= \lim_{p \to 0^+} \frac{d (\log(1-p) - \log(p))}{p^{1/d-1} /d} \\
        &= \lim_{p \to 0^+} \frac{d (-(1-p)^{-1} - p^{-1})}{(1-d)p^{1/d-2} /d^2} \\
        &= \lim_{p \to 0^+} \frac{d^3}{d-1}  \left(\frac{p^{2-1/d}}{1-p} + p^{1-1/d} \right) \\
        &=0.
    \end{align}
    Continuity, therefore, follows from the definition of the function.
\end{proof}

\begin{lemma}\label{lem:gd-mono}
    Let $d \in \mathbb N$, $d \geq 2$. Then, the function $g_d$ is non-decreasing on $[0,1/2]$.
\end{lemma}
\begin{proof}
    We can differentiate $g_d(p)$ on $(0,1/2)$. This yields
    \begin{align}
        \frac{\partial}{\partial p} g_d(p) &= \frac{1}{p^{1/d}} \left(\frac{p^{2/d-1}}{d(1-p^{1/d})}+ (d-1+p^{-1})\log(1-p) - (d-1)\log(p) \right)\\
        &=: \frac{1}{p^{1/d}} g'_d(p)\, . \label{eq:eq_g_d_derivative}
    \end{align}
    We will now show monotonicity in $d$ of $g'_d(p)$ for all $p \in (0, 1/2)$. This will allow us to show the non-negativity of \cref{eq:eq_g_d_derivative} on $(0, 1/2)$ only for $d = 2$ and conclude it for all $d \ge 2$. We have 
    \begin{equation}
        \frac{\partial}{\partial d} g'_d(p) = \frac{p^{2/d - 1}\left(d(p^{1/d} - 1) + (p^{1/d} - 2)\log(p)\right)}{d^3(p^{1/d} - 1)^2} + \log(1 - p) - \log(p) \, .
    \end{equation}
    The above is non-negative for $p \in (0, 1/2)$, if 
    \begin{equation}
        (2 - p^{1/d})\log\frac{1}{p} \ge d(1 - p^{1/d}) \quad \Leftrightarrow \quad \big(1 + \frac{1}{1 - p^{1/d}}\big)\log\frac{1}{p} \ge d
    \end{equation}
    One obtains the last inequality by substitution of $p = e^{d t}$ with $t \in (-\infty, \frac{-\log(2)}{d})$ giving us
    \begin{equation}
        -dt \big(1 + \frac{1}{1 - e^t}\big) \ge d \quad \Leftrightarrow \quad -t\big(1 + \frac{1}{1 - e^t}\big) \ge 1
    \end{equation}
    which is true for $t \in (-\infty , 0)$ hence in particular on $(-\infty, -\frac{\log(2)}{d})$. We thereby have that for $d \ge 2$ $p \in (0, 1/2)$ $g'_d(p) \ge g'_2(p)$. It is straightforward to see that $g'_2(p) \ge 0$ on $p \in (0, 1/2)$. This finally lets us conclude the claim that $g_d(p)$ is non-decreasing on $p \in [0, 1/2]$ as $g_d(p)$ is continuous on $[0,1/2]$ by \cref{lem:gd-cont}.
\end{proof}

\begin{lemma} \label{lem:gd/(1-p)-mono}
    Let $d \in \mathbb N$, $d \geq 2$. Then, the function $p \mapsto g_d(p)/(1-p)$ is non-decreasing on $[0,1)$.
\end{lemma}
\begin{proof}
    The argument follows similar lines as the one in \cref{lem:gd-mono}. We first note that $p \mapsto \frac{1}{1 - p}$ is non decreasing on $[0, 1/2)$ and $p \mapsto g_d(p)$ is as well, as proven in \cref{lem:gd-mono}. Hence $p \mapsto \frac{1}{1 - p}g_d(p)$ is non decreasing on $[0, 1/2]$. What now remains to show is that it is non-decreasing on $[1/2, 1)$. 
    We can differentiate the function on the interval $[1/2,1)$ and obtain
    \begin{align}
        \frac{\partial }{\partial x} \frac{g_d(x)}{1-x}\Big|_{x = p} &= \frac{1}{1 - p}\left(\frac{d p^{-1/d}h(p) - \log(1 - p^{1/d})}{1 - p} \right. \\
        &\hspace{1cm}\left.+ \frac{p^{1/d - 1}}{d(1 - p^{1/d})} - p^{-1/d} p^{-1} h(p) + dp^{-1/d}(\log(1 - p) - \log(p))\right).\\
        &\ge \frac{1}{1 - p}\left(p^{-1/d}h(p)\left(\frac{1}{1 - p} - \frac{1}{p}\right) + (d - 1) p^{-1/d} \left(\frac{h(p)}{1 - p} + \log(1 - p)\right)\right.\\
        &\hspace{1cm} \left. + p^{-1/d} \log(1 - p) - \frac{\log(1 - p^{1/d})}{1 - p}\right)\\
        &\ge 0 \, .
    \end{align}
    The last inequality holds since $p \ge \frac{1}{2}$ and $d \ge 2$ hence 
    \begin{align}
        \frac{1}{1 - p} - \frac{1}{p}    &\ge 0 \, , \\
        \frac{h(p)}{1 - p} + \log(1 - p) &\ge 0 \, ,\\
        p^{-1/d} \log(1 - p) - \frac{\log(1 - p^{1/d})}{1 - p} & \ge 0 \, .
    \end{align}
    To see the last inequality, one can verify that $ p^{-1/d} \leq 1/(1-p)$ in this regime and that $\log(1 - p) \geq \log(1 - p^{1/d})$. Thus $p \mapsto \frac{g_d(p)}{1 - p}$ is non-decreasing on $[1/2,1)$, which concludes the proof.
\end{proof}